\documentclass[10pt]{article}
\usepackage{lmodern}
\usepackage{amsfonts,amsmath,amssymb,amsthm}
\usepackage[pdftex]{graphicx}
\usepackage[hmargin=1in,vmargin=1in]{geometry}
\usepackage{natbib} 
\usepackage{setspace}
\usepackage{enumerate}
\usepackage[title]{appendix}
\usepackage[colorlinks,citecolor=blue,hypertexnames=false]{hyperref}
\usepackage{etoolbox}
\usepackage{booktabs}
\usepackage{multirow}
\usepackage[font=small]{caption}
\allowdisplaybreaks
\usepackage{algorithm,algpseudocode}
 
\renewcommand{\qed}{\rule{2mm}{2mm}}
\newcommand{\indep}{\perp \!\!\! \perp}

\newtheorem{theorem}{Theorem}[section]
\newtheorem{lemma}{Lemma}[section]

\newtheorem{corollary}{Corollary}[section]

\theoremstyle{definition}
\newtheorem{example}{Example}[section]
\newtheorem{remark}{Remark}[section]
\newtheorem{assumption}{Assumption}[section]

\AtEndEnvironment{remark}{~\qed}
\AtEndEnvironment{example}{~\qed}

\DeclareMathOperator*{\argmin}{argmin}

\newenvironment{packed_enum}{\begin{enumerate} \setlength{\itemsep}{1pt}\setlength{\parskip}{0pt}\setlength{\parsep}{0pt}}{\end{enumerate}}

\begin{document}

\title{\vspace{-1em}Inference for Treatment Effects Conditional on Generalized Principal Strata using Instrumental Variables \thanks{We thank the editor and three anonymous referees, as well as Jing Cheng, Jiaying Gu, Sukjin Han, Nathan Kallus, Soonwoo Kwon, Sokbae Lee, Derek Neal, Kirill Ponomarev, Vitor Possebom, Bernard Salani\'e, and Panos Toulis for helpful comments. Shaikh acknowledges support from NSF Grant SES-2419008; Moon acknowledges support from NSF Grant No.\ 2141064; Vytlacil acknowledges financial support from the Tobin Center for Economic Policy at Yale.}}

\author{Yuehao Bai \\
Department of Economics \\
University of Southern California \\
\url{yuehao.bai@usc.edu}
\and
Shunzhuang Huang \\
Booth School of Business \\
University of Chicago \\
\url{shunzhuang.huang@chicagobooth.edu}
\and
Sarah Moon\\
Department of Economics \\
Massachusetts Institute of Technology \\
\url{sarahmn@mit.edu}
\and
Andres Santos \\
Department of Economics \\
University of California--Los Angeles \\
\url{andres@econ.ucla.edu}
\and
Azeem M.\ Shaikh \\
Department of Economics \\
University of Chicago \\
\url{amshaikh@uchicago.edu}
\and
Edward J.\ Vytlacil \\
Department of Economics \\
Yale University \\
\url{edward.vytlacil@yale.edu}
}

\begin{spacing}{1.25}
\maketitle
\end{spacing}

\vspace{-0.3in}
\begin{spacing}{1.1}
\begin{abstract}
We propose a general approach to inference for a broad class of models that arise in the analysis of treatment effects with discrete-valued treatments and instruments and a general-valued outcome. In addition to instrument exogeneity, the main substantive assumption in our class of models rules out certain response types by assuming that they occur with probability zero. Here, the response type refers to the vector of potential outcomes and potential treatments, and we refer to a set of possible values for the response type as a generalized principal stratum. Through a series of examples, we show that this framework encompasses a wide variety of assumptions that have been considered in the previous literature. Our framework allows inference on any treatment effect parameter that can be expressed as the expectation of a function of the response type conditional on a generalized principal stratum. We develop methods for inference on such parameters under these assumptions, as well as methods for testing the validity of the assumptions themselves. A key result of our analysis is a characterization of the identified set for such parameters under these assumptions and the testable restrictions for the assumptions themselves in terms of existence of a nonnegative solution to linear systems of equations with a special structure. We propose methods for inference exploiting this special structure and recent results in \cite{fang2023inference}.
\end{abstract}


\end{spacing}

\thispagestyle{empty} 
\newpage
\setcounter{page}{1}
\fontsize{12pt}{14pt}\selectfont
\section{Introduction}\label{sec:intro}

We propose a general approach to inference for a broad class of models that arise in the analysis of treatment effects with discrete-valued treatments and instruments and a general-valued outcome.  In addition to instrument exogeneity, the main substantive restriction imposed in our analysis is that certain values for the response types occur with probability zero.  Here, the response type refers to the vector of potential outcomes and potential treatments, and we refer to a set of possible values for the response type as a generalized principal stratum. Through a series of examples, we show that our framework encompasses a wide variety of assumptions that have been considered in the previous literature, including the following: (i) restrictions in the analysis of randomized controlled trials (RCTs) under noncompliance considered, e.g., in \cite{imbens1994identification} and \cite{cheng2006bounds}; (ii) generalizations of these restrictions considered in \cite{bai2026identifying}; (iii) revealed preference-type restrictions on response types considered in \cite{kirkeboen2016field}, \cite{kline2016evaluating}, and \cite{heckman2018unordered}; and (iv) restrictions on the ordering of potential treatments or on the ordering of potential outcomes considered in \cite{manski1997monotone}, \cite{manski2000monotone}, and \cite{machado2019instrumental}.\footnote{In the context of mediation analysis, \cite{kwon2026testing} show that ``full mediation'' is equivalent to certain assumptions like those we consider.}

Our framework allows inference on any treatment effect parameter that can be expressed as the expectation of a function of the response type conditional on a generalized principal stratum.  In this way, our framework accommodates many parameters that have been considered previously in the literature, including both average and distributional treatment effect parameters, such as the probability of being strictly helped and the probability of not being hurt by the treatment.
It further accommodates versions of these parameters conditional on sets of possible values of potential treatments, such as the local average treatment effect in \cite{imbens1994identification} and average effects conditional on principal strata in \cite{frangakis2002principal}.  It also accommodates parameters conditional on sets of possible values of potential outcomes, such as those considered in \cite{heckman1997making} and \cite{heckman1998evaluating}, as well as the population persuasion rate in \cite{jun2023identifying}.
By contrast, each of the papers cited in the preceding paragraph tailors its method to specific examples of such treatment effect parameters.  Furthermore, because the types of restrictions considered in these papers are typically insufficient for identification of many parameters of interest, some of these papers focus specifically on those treatment effect parameters that are identified from the two-stage least squares estimand. A characterization of the parameters that are point identified by such restrictions is given by \cite{navjeevan2026identification} and \cite{goff2025doesividentificationrestrict}.  

A key result in our analysis is a characterization of the identified set for such parameters under these assumptions and the testable restrictions for the assumptions themselves in terms of existence of a nonnegative solution to linear systems of equations with a special structure. Using this result, we build on results in \cite{fang2023inference} to develop a test for the null hypothesis that a pre-specified value for the parameter of interest lies in the identified set as well as the null hypothesis that the testable restrictions are satisfied.  Importantly, the resulting test remains well behaved even in ``high-dimensional'' settings, meaning it controls size uniformly over a large class of distributions satisfying weak assumptions and permitting, among other things, the support of the treatment and instrument or the support of the response types to be ``large'' relative to the sample size. Through test inversion, we may also construct confidence regions for such parameters that are uniformly consistent in level. 
We highlight that while our test builds on \cite{fang2023inference}, it leverages the structure of our problem to obtain novel strong approximations that significantly improve on the coupling rates of \cite{fang2023inference}. 
These coupling results for high-dimensional vectors of estimates of conditional probabilities and their bootstrap counterparts may be of independent interest.

Other methods for inference have been proposed in the prior literature for some special cases of our framework. For inference on certain treatment effect parameters, see, e.g., \cite{cheng2006bounds}, \cite{bhattacharya2008treatment,bhattacharya2012treatment}, and \cite{machado2019instrumental}; for inference on the validity of the assumptions, see, e.g., \cite{kitagawa2015test}, \cite{heckman2018unordered}, \cite{sun2023instrument}, \cite{bai2025sharp}, and \cite{kaido2025testing}.  These approaches rely upon closed-form expressions for the identified set for the parameter of interest or testable inequalities, which must be proved or computed on a case-by-case basis.  When such expressions need to be computed, the methods are feasible only when the support of response types is small. In contrast, our approach does not rely on such closed-form expressions and remains computationally feasible even in high-dimensional settings. A key strength of our approach is its broad applicability, including in settings for which no existing methods apply. 

The remainder of the paper is organized as follows. Section \ref{sec:setup} introduces our formal setup and notation. In Section \ref{sec:examples}, we provide several examples of parameters and restrictions previously considered in the literature that are nested by our framework. We present our inference method when the support of the outcome variable is restricted to be finite in Section \ref{sec:inference}, and without such a restriction in Section \ref{sec:contY}. 
We illustrate our method through a simulation study in Section \ref{sec:sims}, and demonstrate its usefulness by revisiting the empirical study of \cite{blattman2017reducing} in Section \ref{sec:app}. Proofs of our main results can be found in Appendix \ref{app:proofs}. An Online Appendix contains the proofs of auxiliary results.

\section{Setup and Notation} \label{sec:setup}

Denote by $Y \in \mathcal Y$ an outcome, by $D \in \mathcal D$ a discrete-valued endogenous regressor (i.e., the treatment), and by $Z \in \mathcal Z$ a discrete-valued instrumental variable. To rule out degenerate cases, we assume throughout that $2 \leq |\mathcal Y|$, $2 \leq |\mathcal D| < \infty$, and $2 \leq |\mathcal Z| < \infty$. We emphasize that neither $\mathcal D$ nor $\mathcal Z$ needs to be a subset of the real line; in this way, our framework accommodates vector-valued treatments and instruments provided they only take a finite number of values. In what follows, we will consider both settings in which $|\mathcal Y|$ is finite and those in which it is not. Further denote by $Y(d) \in \mathcal Y$ the potential outcome if $D = d \in \mathcal D$ and by $D(z) \in \mathcal D$ the potential treatment if $Z = z \in \mathcal Z$.  As usual, we assume that
\begin{equation}\label{eq:potential}
Y = \sum_{d \in \mathcal D} Y(d) I \{D = d\} \text{~~and~~} D = \sum_{z \in \mathcal Z} D(z) I \{Z = z\}~, 
\end{equation}
where, throughout the paper, $I\{E\}$ denotes the indicator function equal to one if the event $E$ holds and zero otherwise.

In our discussion below, it will be convenient to define $R_o \equiv (Y(d) : d \in \mathcal D)$ and $R_t \equiv (D(z) : z \in \mathcal Z)$, and set $R \equiv (R_o,R_t)$.  Following \cite{heckman2018unordered}, we will refer to $R_t$ as the ``treatment response type.''  By analogy with this terminology, we will also refer to $R_o$ as the ``outcome response type'' and to $R$ as simply the ``response type.'' 
We let $Q$ denote the distribution of $(R_o,R_t,Z)$ and note that by \eqref{eq:potential} we have
\begin{equation} \nonumber
(Y, D, Z) = T(R_o,R_t,Z)~,
\end{equation}
for a transformation $T$ implicitly defined through \eqref{eq:potential}. 
The distribution $Q$ therefore determines, through the transformation $T$, the distribution of $(Y,D,Z)$ which we denote by $P$.
We notationally formalize this relation by writing $P = QT^{-1}$, which signifies that $P$ is the pushforward distribution of $Q$ by the function $T$.  

Below we will require that $Q \in \mathbf Q$, where $\mathbf Q$ is a class of distributions satisfying assumptions that we will specify. Different specifications of $\mathbf Q$ represent different assumptions that we impose on the distribution of potential outcomes and potential treatments. In this sense, $\mathbf Q$ may be viewed as a model for potential outcomes and potential treatments.   

Given a distribution $P$ of $(Y,D,Z)$ and a model $\mathbf Q$, we define the identified set for the distribution of $(R_o,R_t,Z)$ as the set of $Q \in \mathbf Q$ that rationalizes $P$. Formally, this set equals
\[ \mathbf Q_0(P, \mathbf Q) \equiv \{Q \in \mathbf Q: P = Q T^{-1}\}~. \]
We say $\mathbf Q$ is consistent with $P$ if and only if $\mathbf Q_0(P, \mathbf Q )$ is not empty. For every model $\mathbf Q$ considered in this paper, every $Q \in \mathbf Q$ is assumed to satisfy the restriction:
\begin{assumption}[Instrument Exogeneity] \label{as:exog}
$R \indep Z$ under $Q$.
\end{assumption}

Our remaining restrictions on $\mathbf Q$ will be formulated in terms of restrictions on possible values of the response type $R$.  
These restrictions will be expressed by specifying a set $\mathcal R \subseteq \mathcal Y^{|\mathcal D|} \times \mathcal D^{|\mathcal Z|}$ characterizing the possible values of $R$.  In Section \ref{sec:examples}, we provide several different choices of $\mathcal R$ that have been previously considered in the literature. For a given choice of $\mathcal R$, we will impose the following assumption on every $Q \in \mathbf Q$:
\begin{assumption}[Response Type Restrictions]\label{ass:generalizedstrata} 
$Q\{R \in \mathcal R \} = 1$.
\end{assumption}

Following \cite{frangakis2002principal}, sets of the form $\{R_t = r_t\}$ for $r_t$ a possible value of $R_t$ are referred to as ``principal strata.''  By analogy with this terminology, we will refer to sets of the form $\{R \in \mathcal R'\}$ for $\mathcal R' \subseteq \mathcal R$ as ``generalized principal strata.''
Using this notation, we introduce our parameters of interest, which are those that can be written as
\begin{equation} \label{eq:parameter}
\theta(Q) \equiv E_Q[g(R) \mid R \in \mathcal R']
\end{equation}
for different choices of integrable function $g : \mathcal R \rightarrow \mathbf R$ and generalized principal strata $\mathcal R' \subseteq \mathcal R$. To ensure the parameter $\theta(Q)$ is well-defined we require $Q$ to satisfy $Q\{R \in \mathcal R'\} > 0$.
A wide variety of parameters can be accommodated in this way. In Section \ref{sec:examples}, we show that specific parameters considered in the literature have the structure in \eqref{eq:parameter}. 
We note now, however, that natural choices of $g$ correspond to the effect of one treatment versus another (i.e., $g(R) = Y(d) - Y(d')$ for $d \in \mathcal D$ and $d' \in \mathcal D$) and the probability that one treatment leads to a larger outcome than another treatment (i.e., $g(R) = I\{Y(d) > Y(d')\}$ for $d \in \mathcal D$ and $d' \in \mathcal D$). We further note that when $\mathcal R' = \mathcal R$, the parameter $\theta(Q)$ simplifies to $E_Q[g(R)]$.  

For a given distribution $P$ and model $\mathbf Q$, the identified set for $\theta(Q)$ is given by
\begin{equation} \label{eq:idset}
\Theta_0(P, \mathbf Q) \equiv \{ \theta(Q) : Q \in \mathbf Q_0(P, \mathbf Q) ~ \mbox{and}~ Q\{R \in \mathcal R'\} > 0 \} ~. 
\end{equation}
The set $\Theta_0(P, \mathbf Q) $ is nonempty whenever there exists at least one $Q \in \mathbf Q_0(P, \mathbf Q)$ satisfying $Q\{R \in \mathcal R'\} > 0$. By construction, this set is ``sharp'' in the sense that for any value in the set there exists a distribution $Q$ that is consistent with $P$, satisfies the restrictions of the model, and for which $\theta(Q)$ equals the prescribed value.

\section{Examples} \label{sec:examples}

In this section, we show how to accommodate several examples from the literature in our framework. Our discussion focuses in particular on Assumption \ref{as:exog} and the specification of $\mathcal R$ in Assumption \ref{ass:generalizedstrata} but, where the cited literature has emphasized specific parameters of interest, we additionally describe how those parameters of interest are special cases of $\theta(Q)$, as defined in \eqref{eq:parameter}, for suitable choices of $g$ and $\mathcal R'$.
Throughout, denote by $r_o = (y(d): d\in \mathcal D)$ and $r_t = (d(z) : z\in \mathcal Z)$, respectively, a possible value of an outcome response type and a treatment response type.

\begin{example}[\textit{RCT with one-sided noncompliance}] \label{eg:onesided}
Consider a multi-arm randomized controlled trial (RCT) with noncompliance, where $\mathcal D = \mathcal Z = \{0, \dots, K\}$, $Z=d$ denotes random assignment to treatment $d$, and $D(d)=d$ denotes that the subject would comply with assignment if assigned to  treatment $d$. In this example, $Q$ satisfies Assumption \ref{as:exog} because $Z$ is randomly assigned. Suppose noncompliance to the assignment is one-sided \citep{bloom1984accounting} in the sense that one can always take the control $d=0$, but, for any other treatment $d \ne 0$, one can only take that treatment if assigned to it.
This restriction can be expressed in terms of Assumption \ref{ass:generalizedstrata} with 
$\mathcal R = \{(r_o,r_t) : d(j) \in \{0, j\} \text{ for all } j \in \mathcal D\}$.  
Analyses of causal parameters that condition on generalized principal strata in this context follow \cite{imbens1994identification} in using the Wald estimand to identify local average treatment effect (LATE) parameters of the form $E_Q[Y(j)-Y(0) \mid D(j)=j]$ for $ j \in \{1, \dots, K\}$. Note that identification of such parameters does not allow comparison of $Y(j)$ to $Y(k)$ for $j, k \ne 0$ for any subgroup of subjects. A notable exception is \cite{cheng2006bounds}, who provide limited results for $K = 2$ and other treatment effect parameters; further discussion of their approach is provided in Remark \ref{remark:bootstrap}.  Our framework nests these parameters as well as others, such as  relative treatment effectiveness, for any subgroup defined by a generalized principal stratum.
\end{example}

\begin{example}[\textit{Encouragement design}] \label{eg:encouragement}
Consider a multi-arm RCT with possibly two-sided noncompliance, where $\mathcal D = \mathcal Z = \{0, \dots, K\}$, $Z=d$ denotes random assignment to treatment $d$, and $D(d)=d$ denotes that the subject would comply with  assignment if assigned to  treatment  $d$. More generally, not necessarily in the context of an RCT, one can interpret $Z=d$ as random encouragement to treatment $d$ and interpret $D(d)=d$  as the subject taking treatment $d$ if encouraged to do so. In this example, $Q$ satisfies Assumption \ref{as:exog} because $Z$ is randomly assigned. \cite{bai2026identifying} generalize the ``no-defier'' restriction of \cite{imbens1994identification} to the condition
\begin{equation} \label{eq:nodefier}
Q\{D(d) \ne d ,   ~D(d')=d ~ \mbox{for some}  ~ d^{\prime} \ne d \}=0 \text{ for all } d \in \mathcal D~,
\end{equation}
i.e., a subject that would not take treatment $d$ if assigned to (encouraged to take) $d$ but who would take $d$ if assigned (encouraged) to some other treatment $d^{\prime} \ne d$ occurs with probability zero.  This restriction can be formulated in terms of Assumption \ref{ass:generalizedstrata} with 
$\mathcal R = \{(r_o,r_t) : d(j) \ne j ~ \Rightarrow d(k) \ne j ~ \text{ for all }~ j, k \in \mathcal{D} \}$.   
\cite{bai2026identifying} derive the identified set for unconditional average treatment effects in this context.  While the two-stage least squares (TSLS) estimand in this setting generally does not correspond to any well-defined causal parameter when $K \ge 2$, \citet{BHULLER2024105785} show that under strong, additional assumptions, it identifies a particular weighted average of strata-specific causal effects that may not be of a priori interest. In contrast, our framework allows inference on a broad class of parameters of substantive interest including those that condition on generalized principal strata and that need not coincide with the TSLS estimand.
\end{example}

\begin{example}[\textit{RCT with close substitute}]\label{eg:close}
\cite{kline2016evaluating} consider an RCT with a ``close substitute'' to study the effects of preschooling on educational outcomes. In their setting, $D \in \mathcal{D} = \{c,h,n\}$, where $D= h$ denotes a preschool program called Head Start, $D = c$ denotes a preschool other than Head Start (the close substitute), and $D = n$ denotes home care (no preschool). Let $Z \in \mathcal{Z}=\{0,1\}$ indicate a randomized offer to attend Head Start. Assumption \ref{as:exog} holds because $Z$ is randomly assigned. 
In evaluating the cost-effectiveness of Head Start, \cite{kline2016evaluating} impose the restriction 
\begin{equation} \label{eq:kw}
Q \{ D(1) = h \mid D(0) \ne D(1)\}=1~
\end{equation}
and the instrument relevance requirement $Q\{D(0) \ne D(1)\}>0$.
Condition \eqref{eq:kw} states that if a family's schooling choice changes upon receiving a Head Start offer, then they must choose Head Start when receiving the offer. Equivalently, it cannot be the case that upon receiving a Head Start offer, a family switches from no preschool to preschools other than Head Start, or the other way around. The restriction can be formulated in terms of Assumption \ref{ass:generalizedstrata} with
\[ \mathcal R = \{(r_o,r_t): (d(0), d(1)) \in \{(c, c), (h, h), (n, n), (c, h), (n, h)\}\}~. \]
\cite{kline2016evaluating} show that the Wald estimand identifies a weighted combination of ``sub-LATEs'':
\[ \frac{E[Y \mid Z=1] - E[Y \mid Z=0]}{P\{D = h \mid Z=1\} - P\{D=h \mid Z=0\}} =  S_c \mathrm{SubLATE}_{ch} + (1 - S_c) \mathrm{SubLATE}_{nh}~, \]
where
\begin{align*}
\mathrm{SubLATE}_{ch} & = E_Q[Y(h)-Y(c) \mid D(1)=h, D(0)=c] \\
\mathrm{SubLATE}_{nh} & = E_Q[Y(h)-Y(n) \mid D(1)=h, D(0)=n] 
\end{align*}
and $S_c = Q\{ D(1) = h, D(0)=c \mid D(1) = h, D(0) \ne h\}$.
Here, $S_c$ and the sub-LATEs can all be written as \eqref{eq:parameter} for appropriate choices of $g$ and $\mathcal R'$. \cite{kline2016evaluating} show that these separate sub-LATE parameters are of substantive interest, and identify them under additional assumptions including imposing a multinomial normal selection model. Our framework allows inference on these sub-LATE parameters without imposing such additional assumptions sufficient to identify them. 
\end{example}

\begin{example}[\textit{Ordered monotonicity with known direction}] \label{eg:ordered}
Consider a setting in which both the treatment and the instrument admit natural orderings, and write $\mathcal D = \{0, \ldots, K_d\}$ and $\mathcal Z = \{0, \ldots, K_z\}$.  Suppose further that $Z$ is randomly assigned, so Assumption \ref{as:exog} holds.  
In many applications, it is reasonable to assume the monotonicity restriction $Q\{D(j) \geq D(k)\} = 1$ for any pair of instrument values $j,k \in \mathcal Z$ with $j \geq k$.  For example, \cite{dupas2014short} reports an experiment in which $Z$ represents different randomly assigned subsidy levels for insecticide-treated bed nets, and $D$ denotes the total number of bed nets purchased over a two-year period.
The monotonicity restriction in this context is that one would purchase at least as many insecticide-treated bed nets with  a greater subsidy as with a lower subsidy.   This  restriction can be expressed in terms of Assumption \ref{ass:generalizedstrata} by setting  
$\mathcal R = \{(r_o,r_t): d(K_z) \geq \ldots \geq d(0)\}$. 
The results of \cite{angrist1995two-stage} imply that the TSLS estimand in this context identifies a specific weighted average of strata-specific causal effects, whereas our analysis allows inference on a broader class of parameters—including partially identified parameters—that need not coincide with the TSLS estimand.
\end{example}

\begin{example}[\textit{Harmless treatment}] \label{eg:harmless}
Consider a setting in which $Z\in \mathcal Z = \{0,\ldots, K_z\}$ is randomly assigned, so Assumption \ref{as:exog} holds, and there is a baseline treatment, i.e., the control, corresponding to $0 \in \mathcal D = \{0,\ldots, K_d\}$. In many applications, it is reasonable to assume that the remaining treatments are  harmless relative to the control.  For example, in \cite{angrist2009incentives}, the outcome is academic performance, the control is no treatment, and the noncontrol treatments are (i) providing students with academic peer-advising service; (ii) financial incentives for good academic performance; and (iii) both (i) and (ii).  It is therefore natural to assume that $Q\{Y(d) \geq Y(0)\} =1$ for all $d \in \mathcal D$.  This 
restriction can be formulated in terms of Assumption \ref{ass:generalizedstrata} by setting 
$\mathcal R = \{(r_o,r_t) : y(d) \geq y(0) \text{ for all } d \in \mathcal D \}$. 
This restriction is an example of what \cite{manski1997monotone} termed a semi-monotone ordering of outcomes.
While \cite{manski1997monotone} develops partial-identification results for unconditional treatment parameters in this context without restrictions on treatment-response types, our analysis additionally allows such restrictions to be incorporated and enables inference on parameters that condition on generalized principal strata.
\end{example}

Appendix \ref{app:examples} contains additional examples that fit our framework, including RCTs with restrictions on noncompliance considered in \cite{cheng2006bounds}, revealed preference-type restrictions on response types considered in \cite{kirkeboen2016field} and \cite{heckman2018unordered}, and monotone treatment response considered in \cite{manski1997monotone}.

In many of the examples discussed above, $Y$ may be an ordinal outcome.  
In such instances, average treatment effects may not be interpretable.
Nonetheless, researchers may consider other parameters that fit our framework, such as: (i) the conditional probability of benefit of treatment $j$ versus treatment $k$ (i.e.\ $Q\{Y(j) > Y(k) \mid R \in \mathcal R' \}$); (ii) the conditional probability of no harm of treatment $j$ versus treatment $k$ (i.e.\ $Q\{Y(j) \ge Y(k) \mid R \in \mathcal R' \}$); or (iii) the conditional relative treatment effect of treatment $j$ versus $k$ (i.e.\ $Q\{Y(j) > Y(k) \mid  R \in \mathcal R'\} - Q\{Y(k) > Y(j) \mid R \in \mathcal R'\}$). Each of these parameters has the structure in \eqref{eq:parameter} for appropriate choices of $g$ and therefore fits our framework.  
These parameters have been previously studied for the special case in which $\mathcal{R}'=\mathcal{R}$ and $|\mathcal{D}|$ equals two or three by \cite{lu2018treatment}, \cite{huang2019constructing}, and \cite{gabriel2024sharp}. 

It is also instructive to discuss examples of restrictions that \emph{do not} fit our framework.
\cite{angrist1995two-stage}, for instance, consider a generalization of the monotonicity assumption in Example \ref{eg:ordered} in which the direction of the monotonicity is not known \textit{a priori}; i.e., they consider the restriction that, for any pair of values for the instrument $j,k \in \mathcal Z$, either $Q\{D(j) \geq D(k)\} = 1$ or $Q\{D(j) \leq D(k)\} = 1$. Such a restriction is also analyzed in \cite{vytlacil2006ordered} and in \cite{heckman2018unordered}, where the latter paper terms it ``ordered monotonicity.''  
These ordered monotonicity conditions fall outside our framework because they correspond to  specifying a collection of restrictions $\mathcal R_\pi$ indexed by $\pi \in \Pi$, and requiring that $Q\{R\in \mathcal R_\pi\} = 1$ for some $\pi \in \Pi$ (instead of $Q\{R\in \mathcal R\} =1$ for some \emph{known} $\mathcal R$).

\section{Inference with Discrete Outcomes} \label{sec:inference}

In this section, we propose a test for conducting inference on the parameter of interest $\theta(Q)$. 
To this end, recall that $P$ denotes the distribution of $(Y,D,Z)$. In this section only, we suppose $|\mathcal Y| < \infty$. Doing so allows us to introduce our method succinctly. 
This assumption will be modified in Section \ref{sec:contY}, where we allow $\mathcal Y$ to be a compact interval in $\mathbf R$.
In what follows, we assume that $P\in \mathbf P$, where $\mathbf P$ is a ``large'' class of distributions that we specify below. 

We define $\mathcal M\subseteq \mathcal Y \times \mathcal D \times \mathcal Z$ as the set of possible realizations of $(Y,D,Z)$ that is consistent with $\mathcal R$ being the support of the response types $R$. 
Formally, recall that $r_o=(y(d):d\in\mathcal D)$ and $r_t=(d(z):z\in\mathcal Z)$, and let  $r=(r_o,r_t)$. For any $z\in\mathcal Z$, let
\begin{equation}\label{eq:Mzdef}
\mathcal M_z \equiv \{(y,d) \in \mathcal Y \times \mathcal D : (y(d),d(z)) = (y,d) \text{ for some } r = (r_o,r_t) \in \mathcal R\}
\end{equation}
and
\begin{equation}\label{eq:Mdef}
\mathcal M \equiv \bigcup_{z\in \mathcal Z} (\mathcal M_z \times \{z\})~.
\end{equation}
The set $\mathcal M$ need not equal $\mathcal Y \times \mathcal D \times \mathcal Z$ --- as in, for example, RCTs with one-sided noncompliance in Example \ref{eg:onesided}. 
For $(y, d, z) \in \mathcal M$, we further define $P_{ydz} \equiv P \{Y = y, D = d, Z = z\}$, $P_z \equiv P \{Z = z\}$, and $P_{yd|z} \equiv P\{Y=y,D=d|Z=z\}$. 
The class $\mathbf P$ can depend on the sample size $n$, but we suppress the dependence from our notation.
We consider testing
\begin{equation} \label{eq:testingprob} 
H_0 : P \in \mathbf P_0 \text{ versus } H_1 : P \in \mathbf P \setminus \mathbf P_0
\end{equation}
at level $\alpha \in (0,1)$, where, for a pre-specified value of $\theta_0$, the null hypothesis we consider is
\begin{equation} \label{eq:null}
\mathbf P_0 = \{ P \in \mathbf P : \theta_0 \in \Theta_0(P,\mathbf Q) \}~.   
\end{equation} 
A test of this null hypothesis can be used to construct confidence regions for $\theta(Q)$ through test inversion over the pre-specified value $\theta_0$. Moreover, as we show below,  a test of this null hypothesis can also be used to conduct a specification test by appropriately selecting $\theta(Q)$.

The key insight underlying the construction of our test is the following theorem, which provides a convenient reformulation of $\mathbf P_0$ in terms of existence of a nonnegative solution to a (possibly under-determined) system of linear equations in which the ``slope coefficients'' are known. The statement of the theorem involves the parameter
\begin{equation} \label{eq:beta}
\beta(P) \equiv ((P_{yd|z} : (y,d,z) \in \mathcal{M}), 1, 0)'.
\end{equation}
We also need to define a matrix $A$ with $|\mathcal M| + 2$ rows and $|\mathcal R|$ columns. Indexing the first $|\mathcal M|$ rows of $A$ by $(y,d,z) \in \mathcal M$ and the columns by $r = (r_o, r_t) \in \mathcal R$, recalling $r_o = (y(d) : d \in \mathcal D)$ and $r_t = (d(z) : z \in \mathcal Z)$, we then set the coordinates of $A$ as
\begin{equation} \label{eq:A}
    \begin{aligned}
    A_{(y, d, z), r} & = I \{y(d) = y, d(z) = d\} \\
    A_{(|\mathcal M| + 1), r} & = 1 \\
    A_{(|\mathcal M| + 2), r} & = (g(r) - \theta_0)I \{r \in \mathcal R'\}~.   
    \end{aligned}
\end{equation}
Table \ref{table:A} illustrates the construction of $A$ and $\beta(P)$ when the parameter of interest is the ATE, the model assumes one-sided noncompliance, and the outcome, treatment, and instrument are binary. Following Example \ref{eg:onesided} with $K = 1$, here $\mathcal R = \{r = (r_o,r_t) : r_o \in \{0,1\}\times \{0,1\}, ~ r_t \in \{(0,0),(0,1)\}\}$ with $\mathcal R^\prime =\mathcal R$ and  $g(r) = y(1)-y(0)$.

\begin{table}[t!]
    \centering
    \begin{tabular}{ccccccccc}
    \toprule
    \multicolumn{8}{c}{$(y(0),y(1),d(0),d(1))$} & \multirow{2}{*}{$\beta(P)$} \\
    \cmidrule(rr){1-8}
    $(0,0,0,0)$ & $(0,0,0,1)$ & $(1,0,0,0)$ & $(1,0,0,1)$ & $(0,1,0,0)$ & $(0,1,0,1)$ & $(1,1,0,0)$ & $(1,1,0,1)$ & \\
    \midrule
    1 & 1 & & & 1 & 1 & & & $P_{00|0}$ \\
    & & 1 & 1 & & & 1 & 1 & $P_{10|0}$ \\
    1 & & & & 1 & & & & $P_{00|1}$ \\
    & & 1 & & & & 1 & & $P_{10|1}$ \\
    & 1 & & 1 & & & & & $P_{01|1}$ \\
    & & & & & 1 & & 1 & $P_{11|1}$ \\
    1 & 1 & 1 & 1 & 1 & 1 & 1 & 1 & $1$ \\
    $-\theta_0$ & $-\theta_0$ & $-1-\theta_0$ & $-1-\theta_0$ & $1-\theta_0$ & $1-\theta_0$ & $-\theta_0$ & $-\theta_0$ & $0$ \\
    \bottomrule
    \end{tabular}
    \caption{Matrix $A$ and $\beta(P)$ in a one-sided noncompliance model, as in Example \ref{eg:onesided}, with $\mathcal Y = \mathcal D = \mathcal Z = \{0,1\}$, and the parameter of interest being the ATE so that  $\mathcal R^\prime = \mathcal R$, and $g(r) = y(1)-y(0)$. All omitted entries are zero.}
    \label{table:A}
\end{table}

Our first theorem employs $\beta(P)$ and $A$ to characterize the null hypothesis.

\begin{theorem} \label{thm:equiv}
Suppose $|\mathcal Y| < \infty$. Let $\mathbf Q$ be the set of all distributions of $(R,Z)$ satisfying Assumptions \ref{as:exog} and \ref{ass:generalizedstrata}, and $\mathbf P$ be the set of all distributions of $(Y,D,Z)$ satisfying $P\{(Y,D,Z)\in \mathcal M\} =1 $ and $P\{Z=z\} > 0$ for all $z\in \mathcal Z$. Then, for $\beta(P)$ as  defined in \eqref{eq:beta} and $A$ as defined in \eqref{eq:A},
it follows that
\begin{equation} \label{eq:equiv}
\mathbf P_0 \subseteq \{P \in \mathbf P : Ax = \beta(P) \text{ for some } x \geq 0 \}~.
\end{equation} 
Furthermore, if $\mathbf P_0\neq \emptyset$, then $ \{P\in \mathbf P : Ax = \beta(P) \text{ for some } x \geq 0\}\subseteq {\rm cl}_{\mathbf P}(\mathbf P_0) $, where $\mathrm{cl}_{\mathbf P}({\mathbf P_0})$ denotes the closure of $\mathbf P_0$ in $\mathbf P$ (understood as a subset of $\mathbf R^{|\mathcal M|}$).  
\end{theorem}

The key implication of Theorem \ref{thm:equiv} is that we may test the null hypothesis of interest by examining whether there exists a nonnegative vector $x$ satisfying $Ax=\beta(P)$.  Because $\theta(Q)$ is not well-defined under $Q$ assigning zero probability to $\mathcal R'$, result \eqref{eq:equiv} involves an inclusion rather than an equality.  
From a testing perspective, however, $\mathbf P_0$ and its closure are indistinguishable, in the sense that  any level $\alpha$ test of whether $P$ belongs to $\mathbf P_0$ necessarily has power no larger than $\alpha$ against any $P$ in the closure $\mathrm{cl}_{\mathbf P}(\mathbf P_0)$ \citep{romano2004non}.
Therefore, the second part of Theorem \ref{thm:equiv} may be interpreted as showing that testing whether $P\in \mathbf P_0$ is in fact equivalent to testing whether $\beta(P) = Ax$ for some $x\geq 0$.
This conclusion holds under the condition that $\mathbf P_0$ is not empty.
However, by definition, $\mathbf P_0$ is empty whenever there is no possible distribution of the data $P\in \mathbf P$ that is compatible with the null hypothesis.
In other words, we may focus on testing whether $\beta(P) = Ax$ for some $x\geq 0$ \emph{unless} the null hypothesis specifies restrictions that are incompatible with any $P$ and, in this sense, may be interpreted as being logically inconsistent with each other.

\begin{remark} \label{rm:dropzdisc}
The system of equations $A x = \beta(P)$ contains linearly dependent restrictions.
This is seen in Table \ref{table:A}, for example, where summing the first two restrictions yields restriction number seven (the adding up constraint). 
More generally, the set of solutions to $A x = \beta(P)$ is unchanged if, for each $z\in \mathcal Z$, we select a point $(y_z,d_z) \in \mathcal M_z$ and remove the rows $((y_z,d_z,z) : z\in \mathcal Z)$ from both $A$ and $\beta(P)$.
In particular, it follows that Theorem \ref{thm:equiv} continues to hold if we redefine $A$ and $\beta(P)$ in this manner; see Lemma \ref{lm:dropz} in the Appendix. 
As we discuss below, this observation can be helpful when constructing our test.
\end{remark}

Theorem \ref{thm:equiv} can be easily extended to test a joint hypothesis on multiple parameters with the structure in \eqref{eq:parameter}. Adding an additional parameter simply corresponds to including an additional row in $A$ (with the same structure as the last equation in \eqref{eq:A}) and adding a zero to $\beta(P)$.
Through test inversion, this observation provides the basis for constructing joint confidence regions for multiple parameters.
Moreover, Theorem \ref{thm:equiv} also allows us to test whether the model is correctly specified in the sense that $\mathbf Q_0(P, \mathbf Q) \neq \emptyset$.  To see this, consider setting $g = 1$, $\theta_0 = 1$, and $\mathcal R' = \mathcal R$.  For such values of these quantities, it follows
\[
\Theta_0(P,\mathbf Q) =
\begin{cases}
\{1\} & \text{ if } \mathbf Q_0(P,\mathbf Q) \neq \emptyset \\
\emptyset & \text{ otherwise }
\end{cases}
~.
\]
Therefore, $\mathbf P_0 = \{P \in \mathbf P : 1 \in \Theta_0(P,\mathbf Q)\} = \{P \in \mathbf P : \mathbf Q_0(P,\mathbf Q) \neq \emptyset \}$.  Under these choices of $g$ and $\theta_0$, the equality in the last row of the equation $Ax = \beta(P)$ in \eqref{eq:equiv} is always satisfied, and we therefore let $\tilde A$ and $\tilde \beta(P)$ denote all but the last rows of $A$ and $\beta(P)$. 
Using this notation, the next result formalizes the implications of Theorem \ref{thm:equiv} for specification testing.

\begin{corollary} \label{cor:validity}
Suppose $|\mathcal Y| < \infty$. Let $\mathbf Q$ be the set of all distributions of $(R,Z)$ satisfying Assumptions \ref{as:exog} and \ref{ass:generalizedstrata}, and $\mathbf P$ be the set of all distributions of $(Y,D,Z)$ satisfying $P\{(Y,D,Z)\in \mathcal M\} = 1$ and $P\{Z=z\} > 0$ for all $z\in \mathcal Z$. Then, for $g = 1$, $\theta_0 = 1$, and $\mathcal R' = \mathcal R$, 
\begin{equation} \label{eq:testabilityrestriction}
\mathbf P_0 = \{P \in \mathbf P : \mathbf Q_0(P, \mathbf Q) \neq \emptyset \} = \{P \in \mathbf P : \tilde A x = \tilde \beta(P) \text{ for some } x \geq 0 \}~.
\end{equation} 
\end{corollary}

Unlike Theorem \ref{thm:equiv}, the conclusion of Corollary \ref{cor:validity} holds without needing to close $\mathbf P_0$. This contrast is due to $\mathcal R^\prime$ equaling $\mathcal R$ in the context of Corollary \ref{cor:validity}, which ensures that any $Q$ in our model assigns probability one to $\mathcal R^\prime$.
The main implication of Corollary \ref{cor:validity} is a general characterization of correct model specification that can be used for specification testing in general settings.
As we show in Appendix \ref{app:testability}, Corollary \ref{cor:validity} can also be employed to derive a set of analytical inequalities that characterize correct model specification. We emphasize, however, that our inference methods below do not rely on these analytical inequalities and instead directly leverage the characterization of correct specification in Corollary \ref{cor:validity}.    

By Theorem \ref{thm:equiv}, we may test the null hypothesis of interest by testing whether $P$ satisfies
\begin{equation}\label{eq:test1}
\beta(P) = Ax \text{ for some } x\geq 0 
\end{equation}
or, equivalently, the smaller set of restrictions in Remark \ref{rm:dropzdisc}.
In recent work, \cite{fang2023inference} derived a test for whether $P$ satisfies the restriction in \eqref{eq:test1} for a general class of problems; see also \cite{bai2022testing} and \cite{cox2023simple} for other approaches.
In what follows, we build on \cite{fang2023inference} by employing the structure of our problem to sharpen rate conditions and establish the validity of their test under weaker requirements on the dimension of $\beta(P)$ (relative to the sample size). Following Corollary \ref{cor:validity}, it is straightforward to modify the test described below to test the null hypothesis that the model is correctly specified in the sense that $\mathbf Q_0(P, \mathbf Q) \neq \emptyset$.

In order to describe the test statistic, we require some additional notation.
We will assume the availability of an i.i.d.\ sample $\{V_i\}_{i=1}^n$ with $V_i=(Y_i,D_i,Z_i)$, $V_i \sim P$, and let $\hat P_n$ denote the empirical distribution corresponding to the sample $\{V_i\}_{i=1}^n$.
We further set $\hat \beta_n \equiv \beta(\hat P_n)$ to denote the plug-in estimator for the parameter $\beta(P)$ and let $\hat \Omega_n \equiv \Omega(\hat P_n)$, where $\Omega(P)$ is a diagonal weighting matrix whose $(i,i)$-entry equals the asymptotic standard deviation of the $i^{th}$ row of $\sqrt n(\hat \beta_n - \beta(P))$.
Specifically, the first $|\mathcal M|$ diagonal entries of $\Omega(P)$ equal
$$\left(\frac{P_{yd|z}(1-P_{yd|z})}{P_z}\right)^{1/2}$$
for the corresponding $(y,d,z)\in \mathcal M$, and the final two diagonal entries of $\Omega(P)$ are equal to zero.
Finally, we let $A^\dagger$ denote the Moore-Penrose pseudoinverse of $A$.

Given the introduced notation, we define $\hat{\mathcal V}_n \equiv  \{s : A^\dagger s \leq 0, \|\hat \Omega_n (AA')^\dagger s\|_1 \leq 1\}$ and set
\begin{equation} \label{eq:teststat-main}
T_n \equiv \sup_{s \in \hat{\mathcal V}_n} \sqrt n \langle A^\dagger s, A^\dagger \hat \beta_n \rangle~
\end{equation}
as our test-statistic --- here, $\|a\|_1 = \sum_{i=1}^d |a_i|$ for any vector $a = (a_1,\ldots, a_d) \in \mathbf R^d$ and $\langle a,b\rangle = a^\prime b$.
As shown in \cite{fang2023inference}, the test statistic $T_n$ represents a sample analogue to an equivalent formulation of the restriction in \eqref{eq:test1} obtained from Farkas' lemma.
In particular, the population analogue to $T_n$ equals zero whenever \eqref{eq:test1} indeed holds and, as a result, under the null hypothesis $T_n$ should not be ``too large.''
We highlight that the test statistic can be computed through linear programming. 
As a result, $T_n$ can be reliably and quickly computed even in high-dimensional applications.

Unfortunately, the asymptotic distribution of $T_n$ is not pivotal.
In particular, its asymptotic distribution depends on the ``directions'' $s \in \hat {\mathcal V}_n$ at which the population analogue $\langle A^\dagger s, A^\dagger \beta(P)\rangle$ is ``close'' to zero.
In order to obtain a valid critical value, we therefore rely on the nonparametric bootstrap and a construction that asymptotically excludes directions $s$ that do not play a role in the distribution of $T_n$; i.e., directions $s$ for which $\langle A^\dagger s, A^\dagger \beta(P)\rangle $ is ``far'' from zero.
Specifically, for the latter purpose we introduce
\begin{align*}
\hat \beta_n^{\mathrm{r}} \in \;& \argmin_b \sup_{s \in \hat{\mathcal V}_n} |\langle A^\dagger s, A^\dagger (\hat \beta_n - b) \rangle| \\
&\text{ s.t. } Ax = b \text{ for some } x \geq 0 \text{ and } b = (b_u', 1, 0)' \text{ for } b_u \in \mathbf R^{|\mathcal M|},
\end{align*}
which represents an estimator for $\beta(P)$ that is restricted to satisfy the null hypothesis.
Importantly, the restricted estimator $\hat \beta_n^{\mathrm{r}}$ can be computed through linear programming, which preserves the computational scalability of the test; see Appendix M.1 of \cite{fang2021inference} for details.
Letting $\{V_i^*\}_{i=1}^n$ denote a bootstrap sample (i.e., a sample drawn i.i.d.\ from the empirical distribution $\hat P_n$), and $\hat \beta_n^*$ the bootstrap analogue to $\hat \beta_n$, we set 
\begin{equation}\label{eq:test3}
T_n^* \equiv \sup_{s\in \hat {\mathcal V}_n} \left\{\sqrt n\langle A^\dagger s, A^\dagger(\hat \beta_n^* - \hat \beta_n)\rangle + \lambda_n \sqrt n \langle A^\dagger s, A^\dagger \hat \beta_n^{\mathrm{r}}\rangle \right\}
\end{equation}
as our bootstrap statistic.
Here, $\lambda_n \in [0,1]$ represents a penalty converging to zero with the sample size, and whose role is to ensure that directions $s$ that do not impact the finite-sample distribution of $T_n$ do not impact the bootstrap statistic either.
Following \cite{fang2023inference} and the default option in their \texttt{lpinfer} package in R, we use the data-dependent choice of $\lambda_n$ given by $\min\{1, 1/\hat \tau_n\}$, where $\hat \tau_n$ is the $1 - 1 / \sqrt{\log(e \vee \log(e \vee n))}$ bootstrap quantile of 
$$\sup_{s \in \hat{\mathcal V}_n} \langle A^\dagger s, \sqrt n A^\dagger(\hat \beta_n^* - \hat \beta_n) \rangle .$$
Importantly, just like $T_n$, the bootstrap statistic $T_n^*$ can also be computed through linear programming and as a result remains computationally tractable in high-dimensional problems.
We also consider the other recommended choice of $\lambda_n$ in \cite{fang2023inference}, motivated by the law of the iterated logarithm, and find that our empirical results are largely robust to this choice; see Section \ref{sec:app} and Table \ref{tab:blattman_empirical_LATE_allcomp_onesided_LIL} in Appendix \ref{app:simdetails} for additional details.

The critical value for a level $\alpha$ test is then the $1-\alpha$ quantile of the distribution of $T_n^*$ conditional on the sample $\{V_i\}_{i=1}^n$.
Formally, the critical value for our test is given by
\begin{equation}\label{eq:test4}
\hat c_n(1-\alpha) \equiv \inf\left\{ u\in \mathbf R : \hat P_n\left\{ T_n^* \leq u\right\} \geq 1-\alpha \right\}.    
\end{equation}
As usual, the critical value $\hat c_n(1-\alpha)$ can be computed through simulation by drawing multiple i.i.d.\ samples $\{V_i^*\}_{i=1}^n$ from the empirical distribution $\hat P_n$ and computing the $1-\alpha$ quantile of the corresponding sample of bootstrap statistics.
Alternative sampling schemes, such as clustered (instead of i.i.d.) data, can be accommodated by appropriately modifying the bootstrap sampling --- e.g., by drawing clusters (instead of observations) with replacement to construct the bootstrap sample.
Given the critical value, we next formally introduce our test, which rejects whenever the test statistic $T_n$ exceeds the critical value $\hat c_n(1-\alpha)$:
\begin{equation} \label{eq:test}
\phi_n(\theta_0) \equiv I \{T_n > \hat c_n(1 - \alpha)\}~.    
\end{equation}

We establish the asymptotic validity of our test under the following main assumption.

\begin{assumption}\label{ass:test}
(i) $V_i \equiv (Y_i,D_i,Z_i), i = 1,\ldots, n$ is an i.i.d.\ sample with $V_i \sim P\in \mathbf P$;
(ii) $(\hat \beta_n - \beta(P)) \in \text{range}(A)$ with probability tending to one uniformly in $P\in \mathbf P$;
(iii) there exists a constant $c_0 > 0$  such that $\sup_{P\in \mathbf P} P_{yd|z}\leq 1-c_0$ for all $(y,d,z)\in \mathcal M$, and there exist constants $c_1, c_2$,     $0 < c_1 < c_2$, satisfying
\begin{align*}
c_1 \leq \inf_{P\in \mathbf P} \min_{(y,d,z)\in \mathcal M} |\mathcal M| P_{ydz} & \leq \sup_{P\in \mathbf P} \max_{(y,d,z)\in \mathcal M} |\mathcal M| P_{ydz} \leq c_2\\
c_1 \leq \inf_{P \in \mathbf P} \min_{z\in \mathcal Z} |\mathcal Z| P_z & \leq \sup_{P \in \mathbf P} \max_{z\in \mathcal Z}|\mathcal Z| P_z \leq c_2~. 
\end{align*}
\end{assumption}

Assumption \ref{ass:test}(i) formalizes our requirement that the sample be independent and identically distributed.
In Assumption \ref{ass:test}(ii) we impose the additional condition that $(\hat \beta_n - \beta(P))$ belong to the range of the matrix $A$ with probability tending to one.
This requirement is satisfied in our examples, where $(\hat \beta_n - \beta(P))$ belongs to the range of $A$ whenever $\hat P_z > 0$ for all $z\in \mathcal Z$ --- i.e., whenever $\hat \beta_n$ is well-defined.
More generally, Assumption \ref{ass:test}(ii) can be verified numerically once the matrix $A$ is constructed; see Remark \ref{rm:check}.
We note that Assumption \ref{ass:test}(ii) is not imposed in \cite{fang2023inference}, but we deploy it here due to its wide applicability and its importance in obtaining simple primitive conditions for establishing the validity of our test.
In applications in which Assumption \ref{ass:test}(ii) is violated, it is still possible to conduct inference by using the results in \cite{fang2023inference}.
Finally, Assumption \ref{ass:test}(iii) represents our main regularity condition on the set $\mathbf P$.
Assumption \ref{ass:test}(iii) requires that all support points have probabilities of the same order ($1/|\mathcal M|$) and that no conditional probability is dominant, in the sense that $P_{yd|z}$ is bounded away from one uniformly in the $(y,d,z)\in \mathcal M$ and $P\in \mathbf P$.

We next establish the validity of our test under Assumption \ref{ass:test} and an additional regularity condition that we formally state in the appendix and discuss after the theorem.

\begin{theorem} \label{thm:inference}
Suppose Assumption \ref{ass:test} holds, $\alpha \in (0, 0.5)$, $\log^3(n)|\mathcal M|/n = o(1)$, and $0 \leq \lambda_n \leq 1$ satisfies $\lambda_n\sqrt{\log(|\mathcal M|)} = o(1)$. Further suppose Assumption \ref{ass:tuning} in the appendix is satisfied. Then, it follows that the test defined in \eqref{eq:test} satisfies $$\limsup_{n \rightarrow \infty} \sup_{P \in \mathbf P_0} E_P[\phi_n(\theta_0)] \leq \alpha.$$
\end{theorem}

Theorem \ref{thm:inference} establishes the validity of our test under weak conditions on the number of support points for $(Y,D,Z)$. 
In the appendix, we also show that an analogous result holds if we instead drop linearly dependent restrictions from $A$ when constructing our test, as discussed in Remark \ref{rm:dropzdisc}; see Theorem \ref{thm:inference:dropz}.
The main rate requirement imposed by Theorem \ref{thm:inference} is that $|\mathcal M|/n$ tend to zero (up to logs).
We view this rate condition as close to minimal, in the sense that it is necessary for the empirical distribution $\hat P_n$ to assign positive probability to each support point $(y,d,z)\in \mathcal M$ with probability tending to one.\footnote{Formally, it can be shown using the negative association of the multinomial distribution \citep{joag1983negative} and Assumption \ref{ass:test}(iii) that in order for $P\{\hat P_n\{(Y,D,Z) = (y,d,z)\} > 0 \text{ for all } (y,d,z)\in \mathcal M\}\to 1$ we must have $|\mathcal M|/n = o(1)$.}
Our rate restriction is in addition a significant improvement over \cite{fang2023inference}, whose best attainable rate is that $|\mathcal M|^2/n $ tend to zero (up to logs).
We improve on their results by employing the specific structure of our problem to derive better coupling rates for our test and bootstrap statistics.
In particular, by relying on coupling results in  \cite{massart1989strong} for the empirical process, we are able to couple our test statistic to a Gaussian counterpart at a rate $r_n$ given by
$$r_n \equiv \frac{\log^{3/2}(n)\sqrt{|\mathcal M|}}{\sqrt n}.$$
The bootstrap statistic is in turn shown to be coupled to a Gaussian distribution at the same rate  $r_n$, which yields the asymptotic validity of our test under the condition that $r_n$ tend to zero; see Lemma \ref{lm:coup} in the Online Appendix for additional details. 

Showing weak convergence of a test statistic and its bootstrap counterpart to a common distribution is in general not sufficient for establishing consistency of the corresponding critical values.
In high-dimensional settings, consistency of the critical values requires a condition termed \emph{anti-concentration} by \cite{chernozhukov2015comparison}.
Intuitively, anti-concentration ensures that the c.d.f.\ of the bootstrap statistic is suitably well behaved at the quantile of interest.
Assumption \ref{ass:tuning}, stated in the appendix for ease of exposition, is a sufficient condition for anti-concentration to hold in our setting.
Alternatively, Assumption \ref{ass:tuning} can be dispensed with by modifying our critical value to include a ``correction factor'' introduced by \cite{andrews2013inference}.
We opt to impose Assumption \ref{ass:tuning} instead, however, to avoid introducing an additional tuning parameter.

Finally, we emphasize that while we have focused our discussion on hypothesis testing, our results readily deliver confidence regions for the parameter of interest as well.
In particular, through test inversion it is straightforward to show that the confidence region
\begin{equation} \label{eq:conf}
C_n = \{\theta_0 \in \mathbf R: \phi_n(\theta_0) = 0\}
\end{equation}
covers the parameter of interest with asymptotic probability of at least $1-\alpha$ uniformly in $P\in \mathbf P$.
Implementing test inversion is facilitated by the observation that only the last row of the matrix $A$ depends (linearly) on $\theta_0$. 
As a result, the computational cost of obtaining $A$ for test inversion is comparable to the computational cost of computing $A$ for a single test.

\begin{remark} \label{remark:weaker}
Assumption \ref{ass:generalizedstrata} imposes that response types take on certain values with probability zero. Our framework can be modified to accommodate the restriction that response types belong to $\mathcal R$ with at most (or at least) some pre-specified probability.  
This modification is useful for conducting sensitivity analysis, as in \cite{kline2013sensitivity} and \cite{masten2020inference}. For instance, in Example \ref{eg:encouragement}, one may relax restriction \eqref{eq:nodefier} so that it instead specifies that the left-hand side (the probability of defiers) is at most some pre-specified $\epsilon > 0$, and explore how inferences on $\theta(Q)$ change as one varies $\epsilon$. 
Similarly, if we maintain Assumption \ref{as:exog} and reject correct model specification by following Corollary \ref{cor:validity}, then our framework also allows us to quantify the violation of the response type restrictions encoded in $\mathcal R$. 
In particular, we may follow \cite{goff2025doesividentificationrestrict} in specifying a richer set of response types $\tilde{\mathcal R}$ satisfying $\mathcal R \subsetneqq \tilde{\mathcal R}$ and employ our procedure to obtain a confidence region for $Q\{R \in \tilde{\mathcal R} \setminus \mathcal R\}$ under the assumption that $Q\{R\in \tilde {\mathcal R}\} =1$.
\end{remark}

\begin{remark} \label{remark:bootstrap}
In applications with closed-form expressions for the identified set, $\Theta_0(P,\mathbf Q) = [L(P), U(P)]$, the plug-in estimator $[L(\hat P_n),U(\hat P_n)]$ provides a natural approach for estimating the identified set. Such approaches are employed by \cite{cheng2006bounds}, who rely on bounds for $L(P)$ and $U(P)$ that are necessarily sharp.
However, the confidence regions based on the plug-in bootstrap in general do not provide correct coverage.
Formally, the endpoint maps $L(P)$ and $U(P)$ are typically maxima or minima of finitely many smooth functionals of the observable conditional-probability vector: affine functionals for unconditional parameters and linear-fractional functionals for conditional parameters.  Consequently, they are directionally differentiable but not fully differentiable at distributions where multiple component functionals attain the relevant maximum or minimum with different derivatives.
Hence, it follows from results in \cite{fang2019inference} that plug-in bootstrap estimates of the quantiles of $L(\hat P_n)$ and $U(\hat P_n)$ are not consistent. While one could use the methods in \cite{fang2019inference} to consistently estimate these quantiles pointwise in $P\in \mathbf P$, the resulting confidence regions would not necessarily be valid uniformly in $P\in \mathbf P$. Moreover, we emphasize that our framework provides valid inference when a closed-form expression of the identified set is not available, which is the case for most of the parameters of interest considered in this paper.
\end{remark}

\begin{remark} \label{rm:check}
Assumption \ref{ass:test}(ii) can, under the remaining conditions of Theorem \ref{thm:inference}, be shown to be implied by a requirement on $A$ that is simple to verify.
Specifically, let $e_j\in \mathbf R^{|\mathcal M|+2}$ be the unit vector with its $j^{th}$ coordinate equal to one and $v_{\tilde z}\in \mathbf R^{|\mathcal M|+2}$ be such that its first $|\mathcal M|$ coordinates (indexed by $(y,d,z)\in \mathcal M$) equal one when $z = \tilde z$, and all other coordinates equal zero.
Assumption \ref{ass:test}(ii) then holds whenever the null space of $A'$ is contained in the span of $\{v_z : z\in \mathcal Z\} \cup \{e_{|\mathcal M|+1}, e_{|\mathcal M|+2}\}$.
In low dimensional settings, this condition is also necessary for Assumption \ref{ass:test}(ii) to hold; see Lemma \ref{lm:equiv} in the Online Appendix.
\end{remark}

\subsection{Inference under Conditional Instrument Exogeneity} \label{subsec:W}

The framework of Theorem \ref{thm:equiv} can be extended to allow instrument exogeneity to hold conditional on an observable discrete-valued covariate $W \in \mathcal W$ with $|\mathcal W| < \infty$. 
To this end, we continue to assume that $(Y,D)$ are determined by $(R,Z)$ according to \eqref{eq:potential} and denote this relation by $(Y, D, Z, W) = T^{\rm w}(R, Z, W)$. 
With a slight abuse of notation, in this subsection only we let $Q$ and $P$ denote the distributions of $(R, Z, W)$ and  $(Y, D, Z, W)$. 
Therefore, for a given model $\mathbf Q^{\rm w}$, the set of distributions $Q$ compatible with a given $P$ equals
$$\mathbf Q_0^{\rm w}(P, \mathbf Q^{\rm w}) \equiv \{Q \in \mathbf Q^{\rm w}: P = Q (T^{\rm w})^{-1}\}.$$ 
We continue to consider models $\mathbf Q^{\rm w}$ that, for a given choice of $\mathcal R \subseteq \mathcal Y^{|\mathcal D|}\times \mathcal D^{|\mathcal Z|}$, satisfy Assumption \ref{ass:generalizedstrata} for every $Q\in \mathbf Q^{\rm w}$.
However, we now replace the instrument exogeneity requirement imposed in Assumption \ref{as:exog} with the condition that every $Q\in \mathbf Q^{\rm w}$ satisfies:
\begin{assumption}[Conditional Instrument Exogeneity] \label{as:exog_w}
$R \indep Z \mid W$ under $Q$.
\end{assumption}

We extend the parameters of interest we consider by allowing them to depend on the covariates $W$.
Formally, we focus on parameters that can be written as
\begin{equation} \label{eq:param_w}
\theta^{\rm w}(Q) = E_Q[g(R, W) | R \in \mathcal R', W \in \mathcal W']
\end{equation}
for different choices of function $g: \mathcal R \times \mathcal W \to \mathbf R$, and sets $\mathcal R' \subseteq \mathcal R$ and $\mathcal W' \subseteq \mathcal W$. 
To ensure the parameter $\theta^{\rm w}(Q)$ is well-defined we require $Q$ to satisfy $Q\{R \in \mathcal R', W \in \mathcal W'\} > 0$. 
Hence, for a given distribution $P$ and model $\mathbf Q^{\rm w}$, the identified set for $\theta^{\rm w}(Q)$ is given by
\[ \Theta_0^{\rm w}(P, \mathbf Q^{\rm w}) \equiv \{\theta^{\rm w}(Q): Q \in \mathbf Q_0^{\rm w}(P, \mathbf Q^{\rm w}) \text{ and } Q\{R \in \mathcal R', W \in \mathcal W'\} > 0\}~. \]
In parallel with Theorem \ref{thm:equiv}, we consider testing the null hypothesis that $P \in \mathbf P_0^{\rm w}$, where
\[ \mathbf P_0^{\rm w} = \{P \in \mathbf P^{\rm w}: \theta_0 \in \Theta_0^{\rm w}(P, \mathbf Q^{\rm w})\} \]
and $\mathbf P^{\rm w}$ is a ``large'' class of distributions satisfying conditions that we specify below. 

We will show that the relationship established in Theorem \ref{thm:equiv} between the null hypothesis  and the existence of a nonnegative solution  to a linear system  extends to this setting. 
The principal challenge in doing so is to ensure that the ``slopes" in the linear system (e.g., the matrix $A$ in Theorem \ref{thm:equiv}) do not depend on $P$ so that the inference framework of Theorem \ref{thm:inference}, based on \cite{fang2023inference}, continues to apply.
To this end, we let $\mathcal M$ be as defined in \eqref{eq:Mdef}, set $\mathcal M^{\rm w} \equiv \mathcal{M} \times \mathcal W$, and for any $(y, d, z, w) \in \mathcal M^{\rm w}$ we let $P_{yd|zw} \equiv P\{Y = y, D = d | Z = z, W = w\}$ and $P_w \equiv P\{W = w\}$. 
Given this notation, we replace the vector $\beta(P)$, defined in \eqref{eq:beta}, with a vector $\beta^{\rm w}(P)$ of dimension $|\mathcal M^{\rm w}| + |\mathcal W| + 1$ given by
\begin{equation} \label{eq:beta-w}
\beta^{\rm w}(P) \equiv ((P_w \times P_{yd|zw}: (y, d, z, w) \in \mathcal M^{\rm w}), (P_w: w \in \mathcal W), 0)'~.
\end{equation}
Similarly, we replace the matrix $A$, defined in \eqref{eq:A}, with a matrix $A^{\rm w}$ having $|\mathcal M^{\rm w}| + |\mathcal W| + 1$ rows and $|\mathcal R| \times |\mathcal W|$ columns. 
Writing the support of the covariates as $\mathcal W = \{w_1, \dots, w_{|\mathcal W|}\}$ and indexing the first $|\mathcal M^{\rm w}|$ rows of $A^{\rm w}$ by $(y, d, z, w) \in \mathcal M^{\rm w}$ and the columns of $A^{\rm w}$ by $(r, \tilde w) \in \mathcal R \times \mathcal W$, we set the coordinates of $A^{\rm w}$ to be given by
\begin{equation} \label{eq:A-w}
\begin{aligned}
A^{\rm w}_{(y, d, z, w), (r, \tilde w)} & = I\{y(d) = y, d(z) = d, \tilde w = w\} \\
A^{\rm w}_{|\mathcal M^{\rm w}| + j, (r, \tilde w)} & = I\{\tilde w = w_j\} \quad \text{for } 1 \leq j \leq |\mathcal W| \\
A^{\rm w}_{|\mathcal M^{\rm w}| + |\mathcal W| + 1, (r, \tilde w)} & = (g(r, \tilde w) - \theta_0) I\{r \in \mathcal R', \tilde w \in \mathcal W'\}~,
\end{aligned}
\end{equation}
where recall that $r_o = (y(d) : d\in \mathcal D)$ and $r_t = (d(z) : z\in \mathcal Z)$ for any $(r_o,r_t) = r\in \mathcal R$.

The next result shows that the conclusion of Theorem \ref{thm:equiv} continues to hold in the presence of discrete-valued covariates, provided we replace $\beta(P)$ and $A$ by $\beta^{\rm w}(P)$ and $A^{\rm w}$.

\begin{theorem} \label{thm:W}
Suppose $|\mathcal Y| < \infty$ and $|\mathcal W| < \infty$. Let $\mathbf Q^{\rm w}$ be the set of all distributions satisfying Assumptions \ref{as:exog_w} and \ref{ass:generalizedstrata}, and $\mathbf P^{\rm w}$ be the set of all distributions of $(Y, D, Z, W)$ satisfying $P\{(Y, D, Z, W) \in \mathcal M^{\rm w}\} = 1$ and $P\{Z = z, W = w\} > 0$ for all $(z, w) \in \mathcal Z \times \mathcal W$. Then, for $\beta^{\rm w}(P)$ as defined in \eqref{eq:beta-w} and $A^{\rm w}$ as defined in \eqref{eq:A-w}, it follows that
\begin{align*}
\mathbf P_0^{\rm w} & \subseteq \{P \in \mathbf P^{\rm w}: A^{\rm w} x = \beta^{\rm w}(P) \text{ for some } x \geq 0\}~.
\end{align*}
Furthermore, if $\mathbf P_0^{\rm w} \neq \emptyset$, then $\{P \in \mathbf P^{\rm w}: A^{\rm w} x = \beta^{\rm w}(P) \text{ for some } x \geq 0\} \subseteq \mathrm{cl}_{\mathbf P^{\rm w}}(\mathbf P_0^{\rm w})$, where $\mathrm{cl}_{\mathbf P^{\rm w}}(\mathbf P_0^{\rm w})$ denotes the closure of $\mathbf P_0^{\rm w}$ in $\mathbf P^{\rm w}$ (understood as a subset of $\mathbf R^{|\mathcal M^{\rm w}|}$).
\end{theorem}

Theorem \ref{thm:W} justifies testing the null hypothesis that $P\in \mathbf P_0^{\rm w}$ by instead examining whether there exists a (weakly positive) vector $x$ solving the linear system $A^{\rm w} x = \beta^{\rm w}(P)$. 
Because $A^{\rm w}$ is constructed to not depend on $P$, the latter null hypothesis can again be tested by following the approach in Theorem \ref{thm:inference} for the case in which there are no covariates.

While we have focused on scalar parameters, the characterization of Theorem \ref{thm:W} readily extends to vector-valued parameters.
For example, consider the conditional parameter curve
 \[
\boldsymbol\theta(Q) \equiv \bigl(\theta_{w}(Q):w\in\mathcal W\bigr)', \qquad \theta_{w}(Q) \equiv E_Q[g(R,W)\mid R\in\mathcal R',W=w], \]
where the curve is defined for $Q$ satisfying $ Q\{R\in\mathcal R', W=w\}>0 $ for every $w\in\mathcal W$. 
To test the null hypothesis that a candidate curve $\boldsymbol\theta_0=(\theta_{w,0}:w\in\mathcal W)'$ belongs to the identified set for $\boldsymbol \theta(Q)$ we then replace the final row of  $A^{\rm w}$ (as in \eqref{eq:A-w}) with $|\mathcal W|$ rows with the structure
$$A^{\rm w}_{|\mathcal M^{\rm w}| + |\mathcal W| + j,(r,\tilde w)} = (g(r,\tilde w) - \theta_{w_j,0})I\{r\in \mathcal R^\prime, \tilde w = w_j\}$$
for each $w_j \in \mathcal W = \{w_1,\ldots, w_{|\mathcal W|}\}$, and similarly replace the final row of $\beta^{\rm w}(P)$ (as in \eqref{eq:beta-w}) with a vector of zeros of dimension $|\mathcal W|$.
By conducting test inversion we may then obtain a joint confidence region for the conditional parameter curve. Similarly, Theorem \ref{thm:W} can be extended to study vector-valued parameters defined by a vector-valued $g$.

\begin{remark}
A simpler linear system can be employed in applications in which the parameter of interest $\theta^{\rm w}(Q)$ (as in \eqref{eq:param_w}) is specified for a singleton $\mathcal W^\prime = \{w^\prime\}$.
In such cases, it is possible to replace $\beta(P)$ in \eqref{eq:beta} by $((P_{yd|zw'}: (y, d, z) \text{ such that } (y, d, z) \in \mathcal M), 1, 0)'$ and replace the last row of $A$ in \eqref{eq:A} with $(g(r, w') - \theta_0) I\{r \in \mathcal R'\}$. Then, Theorem \ref{thm:equiv} applies  to the distribution of $(Y,D,Z)$ given $W=w'$ and, through test inversion, we may obtain a confidence region for the parameter of interest.
\end{remark}

\section{Inference with General-Valued Outcomes}\label{sec:contY}

Section \ref{sec:inference} proposed an inference framework in which the support $\mathcal Y$ of $Y$ is taken to be finite. 
We now modify that framework to allow for continuous outcomes. 
In this section, we take $\mathcal Y$ to be a subset of $\mathbf R$ that is no longer restricted to be finite. The treatment $D$ and instrument $Z$ remain discrete, with supports $\mathcal D$ and $\mathcal Z$ satisfying $2\leq |\mathcal D|<\infty$ and $2\leq |\mathcal Z|<\infty$.

Our approach partitions the outcome space $\mathcal Y$, which is now a subset of $\mathbf R$, into a finite number of subsets $\{B_{\ell,L}\}_{\ell=1}^L$ and reduces the problem into a sequence of discrete ones indexed by $L$. 
To this end, we set 
$$\mathbb L \equiv \sum_{\ell = 1}^L \ell I\{Y \in B_{\ell, L}\}$$
to indicate the subset $B_{\ell,L}$ to which $Y$ belongs.
The random variable $\mathbb L$ takes values in the set $\mathcal L \equiv \{1,\ldots, L\}$ and its distribution therefore depends on $L$ but, for simplicity, we omit the dependence from the notation.
Paralleling Section \ref{sec:inference}, we define $\mathcal M_L$ to be the set of values of $(\mathbb L,D,Z)$ that is consistent with response types $R$ belonging to $\mathcal R$.
Specifically, recalling that $r = (r_o,r_t)$ with $r_o = (y(d) : d \in \mathcal D)$ and $r_t = (d(z) : z\in \mathcal Z)$, we let 
\begin{equation}\label{eq:MzLdef}
\mathcal M_{L,z} \equiv \{(\ell,d) \in \mathcal L \times \mathcal D: (y(d),d(z)) \in B_{\ell,L} \times \{d\} \text{ for some } r = (r_o,r_t) \in \mathcal R\}
\end{equation}
and
\begin{equation}\label{eq:MLdef}
\mathcal M_{L} \equiv \bigcup_{z\in \mathcal Z} (\mathcal M_{L,z} \times \{z\}) .
\end{equation}
Our analysis mimics the discrete case studied in Theorem \ref{thm:equiv}, but replaces $\beta(P)$ with
\begin{equation} \label{eq:beta_cont}
\beta_L(P) \equiv ((P\{Y\in B_{\ell,L}, D=d|Z=z\} : (\ell,d,z) \in \mathcal M_L), 1, 0, 0)' ~;
\end{equation}
i.e., the first $|\mathcal M_L|$ elements of $\beta_L(P)$ match our construction of $\beta(P)$ in Section \ref{sec:inference} for the discrete case when the observations are given by $(\mathbb L,D,Z)$ (instead of $(Y,D,Z)$).
Reflecting Theorem \ref{thm:equiv}, in what follows we relate our inference problem with general $Y$ to the problem of testing whether $\beta_L(P)$ equals $A_L x$ for some $x\geq 0$ and a suitably defined matrix $A_L$.

In order to formally state our results, it is helpful to note that because the support of the outcome response type $R_o$ is contained in $\mathcal Y^{|\mathcal D|}$, any partition $\{B_{\ell,L}\}_{\ell = 1}^L$ of $\mathcal Y$ induces a partition for the support of $R_o$ as well.
Specifically, writing any $c_o\in \mathcal L^{|\mathcal D|}$ as $c_o \equiv (c_o(d): d \in \mathcal D)$, and setting $B_{c_o} \equiv \prod_{d \in \mathcal D} B_{c_o(d), L}$  we have that $\{B_{c_o}\}_{c_o\in \mathcal L^{|\mathcal D|}}$ forms a partition for $\mathcal Y^{|\mathcal D|}$.
Our next assumption uses this notation in introducing the regularity conditions we require for the main result in this section.

\begin{assumption}\label{ass:cont}
(i) $\mathcal Y$ is a compact subset of $\mathbf R$;
(ii) $g$ is continuous; 
(iii) $\{B_{\ell,L}\}_{\ell=1}^L$ is a partition of $\mathcal Y$ satisfying $\max_{1\leq \ell\leq L} \text{diam}\{B_{\ell,L}\} = o(1)$ as $L\to \infty$; 
(iv) $\mathcal R_o^\prime(r_t) \equiv \{r_o\in \mathcal Y^{|\mathcal D|} : (r_o,r_t)\in \mathcal R^\prime\}$ is closed and $B_{c_o}\cap \mathcal R_o^\prime(r_t)$ is connected for all $c_o\in \mathcal L^{|\mathcal D|}$ and $r_t\in \mathcal D^{|\mathcal Z|}$.
\end{assumption}

Assumption \ref{ass:cont}(i) imposes that $\mathcal Y$ be closed and bounded, which ensures that the partition $\{B_{\ell,L}\}_{\ell=1}^L$ can be chosen so that each set $B_{\ell,L}$ is bounded.
In Assumption \ref{ass:cont}(ii), we further restrict attention to the case in which $g$ is continuous. 
As we discuss in Remark \ref{remark:piecewise} below, some forms of discontinuity can be accommodated, though they can require specific choices of partitions $\{B_{\ell,L}\}_{\ell=1}^L$ that satisfy restrictions beyond those imposed by Assumption \ref{ass:cont}.
Assumption \ref{ass:cont}(iii) imposes that the partition $\{B_{\ell,L}\}_{\ell = 1}^L$ become finer as $L$ becomes larger.
Finally, for every treatment response type $r_t\in\mathcal D^{|\mathcal Z|}$, Assumption \ref{ass:cont}(iv) imposes two restrictions on the corresponding set $\mathcal R_o^\prime(r_t)$ of outcome response types consistent with $\mathcal R^\prime$: first, $\mathcal R_o^\prime(r_t)$ must be closed; second, its intersection with $B_{c_o}$ must be connected for every $c_o\in\mathcal L^{|\mathcal D|}$.
When $\mathcal R_o^\prime(r_t)$ is closed and convex for all response types $r_t$, Assumption \ref{ass:cont}(iv) is satisfied whenever the sets $\{B_{\ell,L}\}_{\ell=1}^L$ are chosen to be convex. 
In applications in which $\mathcal Y$ is a compact interval in $\mathbf R$, examples of $\mathcal R_o^\prime(r_t)$ being closed and convex include, for instance, applications in which $\mathcal R^\prime$ does not restrict $R_o$ (i.e., $\mathcal R^\prime_o(r_t) = \mathcal Y^{|\mathcal D|}$) or $\mathcal R^\prime$ imposes (weak) monotonicity restrictions on $R_o$ (i.e., $\mathcal R^\prime_o(r_t) \equiv \{r_o\in \mathcal Y^{|\mathcal D|}: y(d) \geq y(d^\prime) \text{ when } d\geq d^\prime\}$).
In contrast, Assumption \ref{ass:cont}(iv) fails when we impose (strong) monotonicity restrictions on $R_o$ (i.e., $\mathcal R^\prime_o(r_t) \equiv \{r_o\in \mathcal Y^{|\mathcal D|}: y(d) > y(d^\prime) \text{ when } d> d^\prime\}$) because $\mathcal R^\prime_o(r_t)$ then fails to be closed.

Assumption \ref{ass:cont} enables us to relate the inference problem with general outcomes to the problem of testing whether $\beta_L(P) = A_L x$ for some $x\geq 0$ and matrix $A_L$.
For ease of exposition, and to focus on the main ideas, we describe $A_L$ in the special case in which the outcome response types are unrestricted so that $\mathcal R = \mathcal Y^{|\mathcal D|}\times \mathcal R_t$ and $\mathcal R^\prime = \mathcal Y^{|\mathcal D|}\times \mathcal R_t^\prime$. The more general case in which outcome response types are restricted is discussed in the appendix.
The matrix $A_L$ has $(|\mathcal M_L|+3)$ rows and $(L^{|\mathcal D|}\times |\mathcal R_t| + 2)$ columns.
Indexing the first $|\mathcal M_L|$ rows by $(\ell,d,z)\in \mathcal M_L$ and first $L^{|\mathcal D|}\times |\mathcal R_t|$ columns by $(c_o,r_t)$ we then set
\begin{align*}
A_{L,(\ell,d,z), (c_o,r_t)} & = I\{c_o(d) = \ell, d(z) = d\} \\
A_{L,|\mathcal M_L|+1, (c_o,r_t)} & = 1 
\end{align*}
and let the final two entries of the $j^{th}$ row of $A_L$ be zero for any $j \leq |\mathcal M_L|+1$.
The upper $(|\mathcal M_L|+1 ) \times (L^{|\mathcal D|} |\mathcal R_t|)$ submatrix of $A_L$ therefore matches the structure of the first $|\mathcal M|+1$ rows of the matrix $A$ in Theorem \ref{thm:equiv} when the data is given by $(\mathbb L,D,Z)$ (instead of $(Y,D,Z)$).
The final two rows of $A_L$ impose the hypothesized value on $\theta(Q)$ by setting
\begin{align}
A_{L,|\mathcal M_L|+2,(c_o,r_t)} & = \left(\sup_{r_o\in B_{c_o}} g(r_o,r_t) - \theta_0\right)I\{r_t \in \mathcal R^\prime_t\} \notag  \\
A_{L,|\mathcal M_L|+3, (c_o,r_t)} & = \left(\inf_{r_o\in B_{c_o}} g(r_o,r_t) - \theta_0\right)I\{r_t \in \mathcal R^\prime_t\} \label{cont:adj}
\end{align}
and letting the last two entries of the $|\mathcal M_L|+2$ row be $(-1,0)$ and the last two entries of the $|\mathcal M_L|+3$ row be $(0,1)$.
Intuitively, this construction approximates $E_Q[(g(R)-\theta_0)I\{R\in \mathcal R^\prime\}]$, which equals zero under the null hypothesis, from above and below with the approximation becoming more accurate as $L$ becomes larger and the sets $B_{c_o}$ smaller.

The next theorem may be seen as a generalization of Theorem \ref{thm:equiv}.
We emphasize that the result relies on a general formulation of $A_L$ that allows for restrictions on $R_o$.

\begin{theorem} \label{thm:equiv_cont}
Let $\mathbf Q$ be the set of all distributions of $(R,Z)$ satisfying Assumptions \ref{as:exog} and \ref{ass:generalizedstrata} and $\mathbf P$ be the set of all distributions of $(Y,D,Z)$ satisfying $P\{(\mathbb L,D,Z)\in \mathcal M_L\} = 1$ for all $L$, and $P \{Z = z\} > 0$ for all $z \in \mathcal Z$. Further suppose $\{B_{\ell,L}\}_{\ell=1}^L$ is a partition of $\mathcal Y$, and $g$ in the definition of $\theta(Q)$ in \eqref{eq:parameter} is bounded. Then, for $\beta_L(P)$ as in \eqref{eq:beta_cont} and a matrix $A_L$ defined in  Appendix \ref{app:equiv_cont}, it follows 
\begin{equation} \label{eq:equiv_cont}
\mathbf P_0 \subseteq \{P \in \mathbf P : A_Lx = \beta_L(P) \text{ for some } x \geq 0 \}\equiv \mathbf P_{0,L}~.
\end{equation} 
If in addition Assumption \ref{ass:cont} is satisfied and $\mathbf P_0$ is not empty, then it follows that
$$\limsup_{L \to \infty} \mathbf P_{0,L} \equiv \bigcap_{M=1}^\infty \bigcup_{L\geq M} \mathbf P_{0,L} \subseteq {\rm cl}_{\mathbf P}(\mathbf P_0)~,$$ where ${\rm cl}_{\mathbf P}(\mathbf P_0)$ denotes the closure of $\mathbf P_0$ in $\mathbf P$ under weak convergence.
\end{theorem}

The first part of Theorem \ref{thm:equiv_cont} shows that $\mathbf P_0 \subseteq \mathbf P_{0,L}$ for $\mathbf P_{0,L}$ as defined in \eqref{eq:equiv_cont}.
Therefore, any (asymptotically) $\alpha$-level test of the null hypothesis that $P\in \mathbf P_{0,L}$ is also an (asymptotically) $\alpha$-level test of the null hypothesis that $P\in \mathbf P_{0}$.
Importantly, this conclusion does not depend on any of the regularity conditions imposed in Assumption \ref{ass:cont}.
Any such test of the null hypothesis that $P\in \mathbf P_{0,L}$, however, may fail to reject for certain $P\in \mathbf P \setminus \mathbf P_0$ --- i.e., it is not consistent against alternatives $P\in \mathbf P_{0,L}\setminus \mathbf P_0$.
The second part of Theorem \ref{thm:equiv_cont} addresses this concern.
In particular, it establishes that, under Assumption \ref{ass:cont}, the limit superior of $\mathbf P_{0,L}$ is included in the closure of the set $\mathbf P_0$. 
The second part of Theorem \ref{thm:equiv_cont} therefore shows that, as our discretization becomes finer, we are able to detect whether $P$ belongs to $\mathbf P_0$ or not (up to closure).
We also note that, in analogy to Remark \ref{rm:dropzdisc}, these conclusions continue to hold if we remove linearly dependent restrictions from $A_L$ and $\beta_L(P)$ --- see Lemma \ref{lm:dropz:cont} in the Appendix for a formal statement.

The main implication of Theorem \ref{thm:equiv_cont} is that, with a general outcome $Y$, we still can  test the null hypothesis of interest by testing whether $P$ is such that $A_Lx = \beta_L(P)$ for some $x \geq 0$.
The latter null hypothesis can be tested by using the test defined in \eqref{eq:test} for the finite support case but replacing $\hat \beta_n$ with $\hat \beta_{L,n} \equiv \beta_L(\hat P_n)$ and $A$ with $A_L$. 
Implementing such a test requires selecting the partition $\{B_{\ell,L}\}_{\ell = 1}^L$.
In our simulations, we set the partition to consist of $L$ equal-length intervals of the support of $Y$.
Such a choice is consistent with Assumption \ref{ass:cont}(iii) and, under regularity conditions on the conditional density of $Y$ given $(D,Z)$, with the requirements of Theorem \ref{thm:inference} provided $(L/n) \times |\mathcal D|\times |\mathcal Z|$ tends to zero (up to logs).
In particular, if $\mathcal D$ and $\mathcal Z$ are fixed, then we require $L/n$ to tend to zero (up to logs), allowing the partition to quickly become finer with the sample size.
In this regard, Theorem \ref{thm:equiv_cont} suggests that, for power considerations, it is desirable to set $L$ to be large.
Two concerns about setting $L$ to be large are the increased computational cost and the potential for our Gaussian approximations to become inaccurate.
Because $\hat \beta_{L,n}$ is a vector of (conditional) empirical probabilities and the Gaussian approximation to a Binomial distribution fails for ``small" probabilities, a simple rule of thumb for ensuring that $L$ is not ``too large" is to check that all but the last three coordinates of the bootstrap draws of $\hat \beta_{L,n}^*$ are strictly positive.

While we have focused our discussion on hypothesis testing, we again emphasize that our results readily deliver confidence regions. Specifically, test inversion arguments imply that the confidence region in \eqref{eq:conf}, implemented with $A_L$ and $\hat \beta_{L,n}$ in place of $A$ and $\hat \beta_n$, is valid in the sense that it covers the parameter of interest with asymptotic probability of at least $1-\alpha$ uniformly in $P\in \mathbf P$.

\begin{remark} \label{remark:piecewise}
An important example of a parameter of interest that corresponds to a discontinuous $g$, and hence violates Assumption \ref{ass:cont}(ii), is the probability of an event conditional on generalized principal strata.
However, in this case, $g$ is piecewise constant and the conclusion of the second part of Theorem \ref{thm:equiv_cont} continues to hold if we redefine $A_L$ suitably. As an illustration, we continue to consider the special case in which the outcome response types are unrestricted so that $\mathcal R = \mathcal Y^{|\mathcal D|} \times \mathcal R_t$ and $\mathcal R' = \mathcal Y^{|\mathcal D|} \times \mathcal R_t'$. The general case can be handled similarly. Note by an inspection of the proof of Theorem \ref{thm:equiv_cont} that a solution to the linear system in \eqref{eq:equiv_cont} is a vector whose coordinates are the probabilities of the cells $B_{c_o} \times \{r_t\}$.
If $\mathcal Y^{|\mathcal D|} = \bigcup_{1 \leq j \leq K} \mathcal S_j$ and $g(\cdot, r_t)$ is constant on $\mathcal S_j$ for each $r_t$ and $1 \leq j \leq K$, then we may redefine $A_L$ so that the solution to the linear system is instead a vector of probabilities of $(B_{c_o} \cap \mathcal S_j) \times \{r_t\}$ for $1 \leq j \leq K$.
\end{remark}

\begin{remark}\label{remark:condmean}
In a binary treatment setting with $R_o = (Y(0),Y(1))$, for example, we may define the neighborhood $B_\epsilon(q) \equiv \{y : | y -q|\leq \epsilon\}$ for a given point $q$ and $\epsilon > 0$, and set $\mathcal R^\prime = \{(r_o,r_t)  \in \mathcal{R} : y(0) \in B_\epsilon(q)\}$ and $g(r) = y(1)$ to conduct inference on 
\begin{equation}\label{remark:conmean1}
E_Q[Y(1) \mid Y(0) \in B_\epsilon(q) ] ~.
\end{equation}
Our framework applies to \eqref{remark:conmean1} provided $Q\{Y(0) \in B_\epsilon(q)\}$ is positive.
Conducting inference on $E_Q[Y(1) \mid Y(0) = q]$, however, additionally requires conditions such as smoothness restrictions even if a sample $\{R_i\}_{i=1}^n$ is directly observed \citep{bertanha2020impossible}.
An explicit bound $K$ on the derivative of the conditional mean, for example, allows us to obtain a confidence region for $E_Q[Y(1) \mid Y(0) = q]$ by simply expanding a confidence region for \eqref{remark:conmean1}, obtained through our framework, by $K\epsilon$. 
\end{remark}

Finally, we note that Theorem \ref{thm:equiv_cont} provides a useful characterization of whether the model is correctly specified in the sense that $\mathbf Q_0(P, \mathbf Q) \neq \emptyset$.  
As in Section \ref{sec:inference}, we may set $g = 1$, $\theta_0 = 1$, and $\mathcal R' = \mathcal R$, which implies $\mathbf P_0$ then equals $\{P \in \mathbf P : \mathbf Q_0(P,\mathbf Q) \neq \emptyset \}$.  Under these choices, the equalities in the last two rows of the equation $A_Lx = \beta_L(P)$ in \eqref{eq:equiv_cont} are always satisfied and we therefore let $\tilde A_L$ denote all but the last two rows and last two columns of $A_L$ and $\tilde \beta_L(P)$ all but the last two rows of $\beta_L(P)$, respectively. Our next corollary to Theorem \ref{thm:equiv_cont} formalizes these observations and provides the basis for conducting a specification test with general outcomes $Y$.

\begin{corollary} \label{cor:validity_cont}
Let $\mathbf Q$ be the set of all distributions of $(R,Z)$ satisfying Assumptions \ref{as:exog} and \ref{ass:generalizedstrata} and $\mathbf P$ be the set of all distributions of $(Y,D,Z)$ satisfying $P\{(\mathbb L,D,Z)\in \mathcal M_L\} = 1$ for all $L$, and $P \{Z = z\} > 0$ for all $z \in \mathcal Z$. Then, for $g = 1$, $\theta_0 = 1$, and $\mathcal R' = \mathcal R$, $$\mathbf P_0 = \{P \in \mathbf P : \mathbf Q_0(P,\mathbf Q) \neq \emptyset \} \subseteq \{P \in \mathbf P : \tilde A_L x = \tilde \beta_L(P) \text{ for some } x \geq 0\} \equiv \tilde {\mathbf P}_{0,L}~.$$  If in addition Assumption \ref{ass:cont} holds and $\mathbf P_0$ is not empty, then $\limsup_{L\to \infty} \tilde {\mathbf P}_{0,L} \subseteq {\rm cl}_{\mathbf P}(\mathbf P_0)$, where ${\rm cl}_{\mathbf P}(\mathbf P_0)$ denotes the closure of $\mathbf P_0$ in $\mathbf P$ under weak convergence.
\end{corollary}

\section{Simulations} \label{sec:sims}
In this section, we present simulation results that illustrate the finite-sample performance of our proposed inference methods. Throughout, we let $\mathcal D = \mathcal Z = \{0, 1, 2\}$ while setting $\mathcal Y$ to be discrete in Sections \ref{sec:sims_disc} and \ref{sec:sims_compare} and an interval in Section \ref{sec:sims_cont}. The parameters of interest in Sections \ref{sec:sims_disc} and \ref{sec:sims_cont} are
\begin{equation}\label{sim:paramdef}
\theta_{\mathrm{ATE}}(Q) = E_Q[Y(2) - Y(1) \mid R \in \mathcal R'] \quad \text{and} \quad \theta_{\mathrm{Prob}}(Q) = Q\{Y(2) \geq Y(1) \mid R \in \mathcal R'\}~,
\end{equation}
where $\mathcal R' = \{r\in \mathcal R : (d(0),d(1),d(2)) = (0,1,2)\}$ and recall $r = (r_o,r_t)$ with $r_o =(y(d) : d\in \mathcal D)$ and $r_t = (d(z) : z\in \mathcal Z)$. We consider the following two specifications of $\mathcal R$:
\begin{enumerate}[]
	\item {\it One-sided noncompliance}: $\mathcal R$ is as specified in Example \ref{eg:onesided} and denoted by $\mathcal R_{\mathrm{1s}}$, with the corresponding model $\mathbf Q_{\mathrm{1s}} = \{Q: Q \text{ satisfies Assumptions \ref{as:exog} and \ref{ass:generalizedstrata} with } \mathcal R = \mathcal R_{\mathrm{1s}}\}$;
	\item {\it Encouragement design}: $\mathcal R$ is as specified in Example \ref{eg:encouragement} and denoted by $\mathcal R_{\mathrm{enc}}$, with the corresponding model $\mathbf Q_{\mathrm{enc}} = \{Q: Q \text{ satisfies Assumptions \ref{as:exog} and \ref{ass:generalizedstrata} with } \mathcal R = \mathcal R_{\mathrm{enc}}\}$.
\end{enumerate}
In Section \ref{sec:sims_compare}, we consider $\mathcal R' = \mathcal R = \mathcal R_{\mathrm{enc}}$ and $\theta_{\rm ATE}(Q)$, so closed-form expressions of the identified set are readily available, and compare our method with existing approaches based on such closed-form expressions.

Throughout, we implement our test using the \texttt{lpinfer} package in R with the default choice of penalty  $\lambda_n$ (discussed in Section \ref{sec:inference}), five hundred bootstrap draws, and using the specification of $A$ that removes linearly dependent restrictions (as described in Remark \ref{rm:dropzdisc}). In unreported simulations we find that the results are nearly unaffected by the exact identity of the removed restrictions, and that a procedure based on $A$ without removing such restrictions also controls size but delivers slightly larger confidence regions. All results are averaged over three thousand replications, which were run in parallel using 5 cores of an AMD EPYC 9335 processor and 20 GB of memory. The details of the distributions employed in the simulation designs can be found in the online repository at \url{https://github.com/edward-vytlacil/STRATA-ATE/blob/main/replication.pdf}.

\subsection{Discrete Outcomes}\label{sec:sims_disc}

We first present simulation results when $\mathcal Y = \{-1, 0, 1\}$. For each model, $\mathbf Q_{\rm 1s}$ and $\mathbf Q_{\rm enc}$, we select a distribution $Q$. The chosen distributions $Q$ are described in detail in the online repository. 
Employing these selected distributions $Q$, we then repeatedly draw i.i.d.\ samples $(Y_i,D_i,Z_i)$, $i = 1,\ldots, n$ from $P = Q T^{-1}$ (i.e., by using relation \eqref{eq:potential}).
For each replication, we obtain the confidence intervals for the parameters $\theta_{\mathrm{ATE}}(Q)$ and $\theta_{\mathrm{Prob}}(Q)$ by inverting the test of Theorem \ref{thm:inference} implemented at the five percent level.

\begin{table}[t!]
\centering
\begin{tabular}{llccccccl}
\toprule
 \multicolumn{9}{c}{One-Sided Noncompliance Design}\\
 \midrule
 & \multirow{2}{*}{$\dim(A)$} & \multirow{2}{*}{$|\Theta_0(P, \mathbf{Q}_{\rm 1s})|$} & \multirow{2}{*}{$n$} & $95\%$ CI & Cov.\ & Cov.\  & Cov.\  & \multirow{2}{*}{Runtime (s.d.)} \\ 
 & & & & length & $\theta_0$ & $\min \Theta_0$ & $\max \Theta_0$ & \\
  \midrule
\multirow{2}{*}{$\theta_{\rm ATE}(Q)$} & $(15+2)\times$ & \multirow{2}{*}{0.881} & 3000 & 1.035 & 1.000 & 0.948 & 0.947 & 97.75 (18.62)\\
 & $108$ &  & 12000 & 0.961 & 1.000 & 0.955 & 0.954 & 109.15 (21.66)\\ 
 \midrule
\multirow{2}{*}{$\theta_{\rm Prob}(Q)$} & $(15+2)\times$ & \multirow{2}{*}{0.381} & 3000 & 0.414 & 1.000 & 0.945 & 1.000 & 67.85 (12.85) \\
& $108$ & & 12000 & 0.397 & 1.000 & 0.950 & 1.000 & 75.85 (15.15)\\
\addlinespace[1.25em]
\multicolumn{9}{c}{Encouragement Design}\\
\midrule
 & \multirow{2}{*}{$\dim(A)$} & \multirow{2}{*}{$|\Theta_0(P, \mathbf{Q}_{\rm enc})|$} & \multirow{2}{*}{$n$} & $95\%$ CI & Cov.\ & Cov.\  & Cov.\  & \multirow{2}{*}{Runtime (s.d.) } \\ 
 & & & & length & $\theta_0$ & $\min \Theta_0$ & $\max \Theta_0$ & \\
 \midrule
\multirow{2}{*}{$\theta_{\rm ATE}(Q)$} & $(27+2)\times$ & \multirow{2}{*}{2.836} & 3000 & 3.340 & 1.000 & 0.958 & 0.945 & 85.35 (21.41)\\
 & $270$ &  & 12000 & 3.083 & 1.000 & 0.945 & 0.950 & 97.91 (24.89)\\
 \midrule
\multirow{2}{*}{$\theta_{\rm Prob}(Q)$} & $(27+2)\times$ & \multirow{2}{*}{0.894} & 3000 & 0.984 & 1.000 & 0.951 & 1.000 & 63.97 (15.89)\\
& $270$ & & 12000 & 0.969 & 1.000 & 0.947 & 1.000 & 73.28 (18.65)\\
\bottomrule
\end{tabular}
\caption{Summary results for confidence intervals under $Q \in \mathbf Q_{\mathrm{1s}}$ and $Q \in \mathbf Q_{\rm enc}$ as specified in the online repository.
Each row reports length of the identified set, average length of the $95\%$ confidence intervals, coverage rate for the true parameter, coverage rates for the lower and upper endpoints of the identified set, and average runtime in seconds for obtaining one confidence region (with standard deviation of runtime in parentheses).}
\label{tab:FSST_MC_1sided_fixQ}
\end{table}

Table \ref{tab:FSST_MC_1sided_fixQ} reports results for sample sizes of $n = 3000$ and $n = 12000$.
For each parameter and each model, we report the average length of the 95\% level confidence intervals and the coverage rates for the true parameter and the endpoints of the identified set. We also report the lengths of the identified sets, $|\Theta_0(P, \mathbf{Q}_{\rm 1s})|$ and $|\Theta_0(P, \mathbf{Q}_{\rm enc})|$. The difference between the average length of the confidence intervals and the length of the identified set can be attributed to sampling uncertainty. To illustrate the computational cost, we also report the average runtime required to compute one confidence region, as well as the dimension of the matrix $A$, which corresponds to the number of constraints and number of decision variables in each of the linear systems being tested. For clarity, we report $\dim(A)$ without removing linearly dependent restrictions.

In these designs, the confidence interval covers the true parameter with probability one in all cases. This phenomenon is to be expected because in each case the true value of the parameter lies in the interior of the identified set.
For example, in the one-sided noncompliance design, the identified set $\Theta_0(P,\mathbf Q_{\mathrm{1s}})$ for $\theta_{\rm ATE}(Q)$ equals $[0.148, 1.030]$, while $\theta_{\rm ATE}(Q)$ itself equals $0.492$ (see Table \ref{tab:sims-trueval} in Appendix \ref{app:simdetails}).
On the other hand, we obtain nominal coverage for the lower and upper endpoints of the identified set in all cases except for the upper endpoint of the identified set for $\theta_{\mathrm{Prob}}(Q)$ in both models.
In the latter cases, the upper endpoints of both the identified set and the confidence interval equal the logical bound of one for probabilities.
Computational costs for obtaining a single confidence interval are low and range between one and two minutes for these designs.
While we generally expect the computational cost to increase with the dimension of $A$, we see in Table \ref{tab:FSST_MC_1sided_fixQ} that it is also influenced by the design.

\subsection{General-Valued Outcomes} \label{sec:sims_cont}

To examine the performance of our procedure when $Y$ is continuous, we next set $\mathcal Y = [-5,5]$ while retaining $\mathcal D = \mathcal Z = \{0,1,2\}$. 
We focus on the one-sided noncompliance design (i.e., $\mathcal R = \mathcal R_{\mathrm{1s}}$) and the parameter $\theta_{\mathrm{ATE}}(Q)$ as defined in \eqref{sim:paramdef}. 
As in Section \ref{sec:sims_disc}, we select a distribution $Q \in \mathbf Q_{\rm 1s}$, specified in detail in the online repository, and repeatedly draw i.i.d.\ samples $(Y_i,D_i,Z_i)$, $i = 1, \ldots, n$  from $P = Q T^{-1}$.

Implementing our inference approach with continuous outcomes requires a choice of partition of the support of $Y$. 
We first consider the following three partitions of $\mathcal Y$:
\begin{enumerate}[]
    \item {\tt bin1}: $\mathcal Y = [-5, 0) \cup [0, 5]$;
    \item {\tt bin2}: $\mathcal Y = [-5,-3) \cup [-3, -1) \cup \cdots \cup [3,5]$;
    \item {\tt bin3}: $\mathcal Y = [-5, -4) \cup [-4,-3) \cup \cdots \cup [4,5]$.
\end{enumerate}
In Table \ref{tab:FSST_MC_cont_Ibound}, we report results for sample sizes $n = 3000$ and $n= 12000$. 
For each level of discretization, we report the average length of the $95\%$ level confidence intervals and the coverage rates of the true parameter and the endpoints of $\Theta_{0,L}(P,\mathbf Q_{\mathrm{1s}})$, which is the set of values that are compatible with $P$ satisfying the restrictions corresponding to a partition of size $L$.
In order to examine computational cost, we again report average runtime and the dimension of the matrix $A_L$ in Theorem \ref{thm:equiv_cont}.

\begin{table}[t!]
\centering
\begin{tabular}{llccccccl}
\toprule
\multirow{2}{*}{Partition} 
& \multirow{2}{*}{$\dim(A_L)$} 
& \multirow{2}{*}{$|\Theta_{0,L}(P,\mathbf Q_{\mathrm{1s}})|$} 
& \multirow{2}{*}{$n$} 
& $95\%$ CI
& Cov. & Cov. & Cov.
& \multirow{2}{*}{Runtime (s.d.)} \\
& & & & length
& $\theta_0$ 
& $\min \Theta_{0,L}$ 
& $\max \Theta_{0,L}$
& \\
\midrule

{\tt bin1}
 & $(10+3)\times$ 
& \multirow{2}{*}{12.385} 
& 3000  & 12.966  & 1.000 & 0.946 & 0.955 & 103.91 (29.14)  \\
($L=2$) & $(32+2)$ 
& 
& 12000  & 12.675  & 1.000 & 0.959 & 0.956 & 115.60 (33.48)  \\

\midrule
{\tt bin2}
& $(25+3)\times$  
& \multirow{2}{*}{5.788} 
& 3000  & 6.100  & 1.000 & 0.987 & 0.987 & 119.80 (33.97) \\
 ($L=5$) & $(500+2)$ 
& 
& 12000  & 5.920  & 1.000 & 0.982 & 0.979 & 136.95 (39.52) \\

\midrule
{\tt bin3} 
& $(50+3)\times$ 
& \multirow{2}{*}{3.668} 
& 3000  & 4.080  & 1.000 & 0.996 & 0.997 & 206.50 (56.13) \\
($L=10$) & $(4000+2)$ 
& 
& 12000  & 3.848  & 1.000 & 0.995 & 0.995 & 238.83 (65.61) \\

\bottomrule
\end{tabular}
\caption{Summary results for confidence intervals under $Q \in \mathbf Q_{\mathrm{1s}}$ as specified in the online repository. Each row includes the value of $L$, dimension of $A_L$, length of $\Theta_{0,L}(P,\mathbf Q_{\mathrm{1s}})$, average length of the $95\%$ confidence intervals, coverage rate for the true $\theta_{\mathrm{ATE}}(Q)$, coverage rates for the lower and upper endpoints of $\Theta_{0,L}(P,\mathbf Q_{\mathrm{1s}})$, and average runtime in seconds for obtaining one confidence region (with standard deviation of runtime in parentheses).}
\label{tab:FSST_MC_cont_Ibound}
\end{table}

As in Table \ref{tab:FSST_MC_1sided_fixQ}, the confidence intervals cover the parameter with probability one for each level of discretization.  This phenomenon is again to be expected because in each case the true value of the parameter lies strictly in the interior of $\Theta_{0,L}(P,\mathbf Q_{\mathrm{1s}})$ (see Table \ref{tab:sims-trueval} in Appendix \ref{app:simdetails}). 
For $L = 2$, we continue to obtain nominal coverage for the lower and upper endpoints of the identified sets. For $L = 5$ and $L = 10$, however, the coverage rates of the lower and upper endpoints of $\Theta_{0,L}(P, \mathbf Q_{\mathrm{1s}})$ exceed the nominal level.
We cannot calculate the identified set $\Theta_0(P,\mathbf Q_{\mathrm{1s}})$ explicitly, but because $\Theta_0(P,\mathbf Q_{\mathrm{1s}}) \subseteq \Theta_{0,L}(P,\mathbf Q_{\mathrm{1s}})$ for each $L$ by Theorem \ref{thm:equiv_cont}, the coverage rates for the endpoints of $\Theta_0(P,\mathbf Q_{\mathrm{1s}})$ must also exceed the nominal level for $L = 5$ and $L = 10$.
Increasing $L$ decreases the length of the identified set and the average length of our confidence intervals.
However, because the number of constraints and decision variables grows as the partition becomes finer, the computational cost also increases with $L$.

\begin{figure}[t!]
    \centering
    \includegraphics[width=0.75\linewidth]{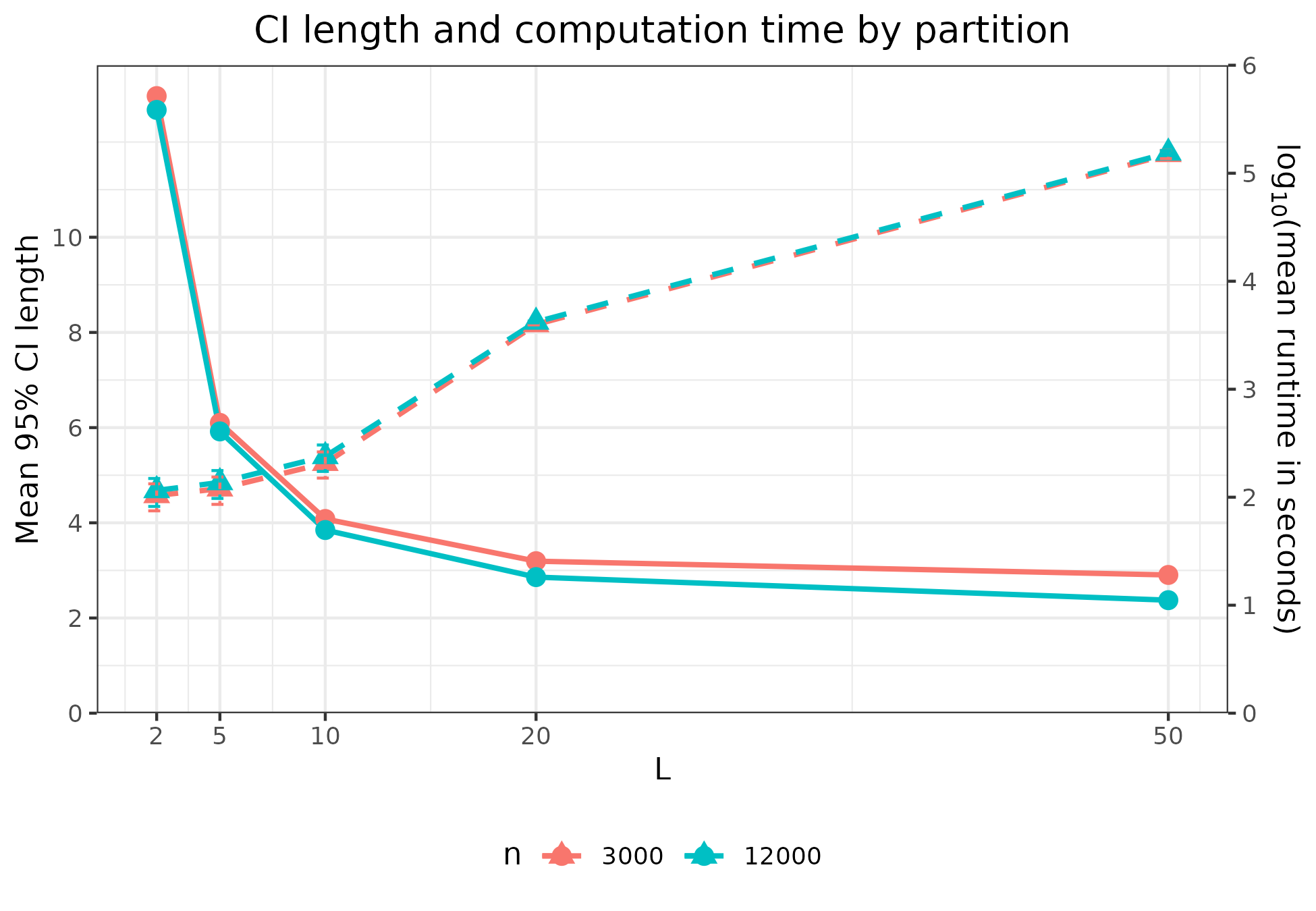}
    \caption{Average length and computation time of a confidence interval for $\theta_{\rm ATE}(Q)$ for different choices of partition size $L$. Results for $L = 2,5,10$ are averaged over 3000 replications (also reported in Table \ref{tab:FSST_MC_cont_Ibound}), while results for $L = 20,50$ correspond to two replications.}
    \label{fig:CIlen_comptime_byL}
\end{figure}

To further understand the tradeoff between computation time and the length of the confidence intervals, we increase $L$ to 20 (\texttt{bin4}) and 50 (\texttt{bin5}). Repeated simulation draws in these designs are computationally prohibitive, so we instead report results based on two replications in Figure \ref{fig:CIlen_comptime_byL}. 
As expected, when $L$ further increases, the length of the confidence interval further decreases (with the caveat that we only compute the confidence interval in one replication).
However, while the marginal decrease in confidence interval length quickly diminishes with $L$, the corresponding computational cost increases rapidly in $L$.
For $n = 12000$, for example, the computational cost for $L = 20$ and $L = 50$ is approximately eleven and five hundred times larger than for $L = 10$.
This significant increment in computational cost reflects the important change in the dimension of $A_L$, which equals $103\times 32{,}002$ and $253 \times 500{,}002$ for $L = 20$ and $L = 50$ (vs. $53\times 4{,}002$ for $L=10$).


\subsection{Comparison with Existing Approaches}\label{sec:sims_compare}
We conclude our simulation study by comparing our methods to existing approaches based on closed-form expressions of the identified sets for the parameters of interest. Before proceeding, we emphasize that closed-form expressions are not available for most parameters considered in this paper. The comparison in this subsection serves to illustrate the advantages and disadvantages of our methods when closed-form expressions happen to be available.

We retain $\mathcal D = \mathcal Z = \{0, 1, 2\}$ and consider either $\mathcal Y = \{-1, 0, 1\}$ or $\mathcal Y = \{-1, 0, 1, 2\}$. We focus on the encouragement design (i.e., $\mathcal R = \mathcal R_{\mathrm{enc}}$) in Section \ref{sec:sims_disc} and set $\mathcal R' = \mathcal R$, so that the parameter of interest is the average treatment effect $$\theta(Q) = E_Q[Y(2) - Y(1)].$$
Following our analysis of Sections \ref{sec:sims_disc} and \ref{sec:sims_cont}, for each $\mathcal Y$, we again fix a distribution $Q \in \mathbf Q_{\mathrm{enc}}$, and draw i.i.d.\ samples $(Y_i,D_i,Z_i)$, $i = 1, \ldots, n$ through the relation \eqref{eq:potential} --- i.e., by drawing from $P = QT^{-1}$. The $Q$ for $\mathcal Y = \{-1,0,1\}$ is the same $Q \in \mathbf Q_{\mathrm{enc}}$ that we used in Section \ref{sec:sims_disc}. Details of $Q$ for both $\mathcal Y$ can be found in the online repository.

In this case, closed-form expressions of the identified set are readily available. Letting $y^L$ and $y^U$ denote the minimum and maximum of the elements of the support $\mathcal Y$, it follows from Corollary 3.1 in \cite{bai2026identifying} that $\Theta_0(P, \mathbf Q_{\mathrm{enc}}) = [L(P), U(P)]$, where
\begin{align*}
L(P) & = \sum_{y \in \mathcal Y} y P \{Y = y, D = 2 | Z = 2\} + y^L P \{D \neq 2 | Z = 2\} \\
& \hspace{2em} - \sum_{y \in \mathcal Y} y P \{Y = y, D = 1 | Z = 1\} - y^U P\{D \neq 1 | Z = 1\}~, \\
U(P) & = \sum_{y \in \mathcal Y} y P \{Y = y, D = 2 | Z = 2\} + y^U P \{D \neq 2 | Z = 2\} \\
& \hspace{2em} - \sum_{y \in \mathcal Y} y P \{Y = y, D = 1 | Z = 1\} - y^L P\{D \neq 1 | Z = 1\}~.    
\end{align*}
We compare our proposed method in Section \ref{sec:inference} based on the implicit characterization of the identified set (labeled ``FSST latent'') with a procedure that exploits this closed-form expression for the identified set.
Specifically, we compare our method to a confidence region obtained through test inversion of a test that the moment inequalities (i) $\theta_0 - L(P)\geq 0$ and (ii) $U(P) - \theta_0 \geq 0$ hold.
To help isolate the effect of having a closed-form solution for the identified set in our comparison,  we again employ the procedure in \cite{fang2023inference} (labeled ``FSST ineq'') to test these moment inequalities. 
For this purpose, we note that the moment inequalities can be expressed as the restriction that $\beta(P) = Ax$ for some $x \geq 0$ by  setting $\beta(P) = (\theta_0-L(P), U(P)-\theta_0)^\prime$ and $A$ to be the $2\times 2$ identity matrix.

\begin{table}[t!]
\centering
\begin{tabular}[t]{llccccccl}
\toprule
\multicolumn{9}{c}{Encouragement Design with $\mathcal Y = \{-1,0,1\}$}\\ \midrule
\multirow{2}{*}{Method} & \multirow{2}{*}{$\dim(A)$} & \multirow{2}{*}{$|\Theta_0(P, \mathbf{Q}_{\mathrm{enc}})|$} & \multirow{2}{*}{$n$} & \multicolumn{1}{c}{95\% CI} & Cov. & Cov. & Cov. & \multirow{2}{*}{Runtime (s.d.)} \\
& & & & length & $\theta_0$ & $\min \Theta_0$ & $\max \Theta_0$ \\
\midrule
\multirow{2}{*}{FSST latent} & $(27+2) \times$ & \multirow{2}{*}{1.061} & 3000 & 1.180 & 1.000 & 0.953 & 0.953 & 91.58 (30.03) \\
& $270$ & & 12000 & 1.120 & 1.000 & 0.953 & 0.956 & 108.20 (36.17) \\
\cmidrule{1-9}
\multirow{2}{*}{FSST ineq} & $2 \times$ & \multirow{2}{*}{1.061} & 3000 & 1.178 & 1.000 & 0.953 & 0.952 & 83.66 (27.49) \\
& $2$ & & 12000 & 1.119 & 1.000 & 0.951 & 0.955 & 101.27 (33.64) \\
\addlinespace[1.25em]
\multicolumn{9}{c}{Encouragement Design with $\mathcal Y = \{-1,0,1,2\}$}\\ \midrule
\multirow{2}{*}{Method} & \multirow{2}{*}{$\dim(A)$} & \multirow{2}{*}{$|\Theta_0(P, \mathbf{Q}_{\mathrm{enc}})|$} & \multirow{2}{*}{$n$} & \multicolumn{1}{c}{95\% CI} & Cov. & Cov. & Cov. & \multirow{2}{*}{Runtime (s.d.)} \\
& & & & length & $\theta_0$ & $\min \Theta_0$ & $\max \Theta_0$ \\
\midrule
\multirow{2}{*}{FSST latent} & $(36+2) \times$ & \multirow{2}{*}{2.392} & 3000 & 2.562 & 1.000 & 0.949 & 0.952 & 122.92 (31.69) \\
& $640$ & & 12000 & 2.476 & 1.000 & 0.951 & 0.960 & 147.48 (37.59) \\
\cmidrule{1-9}
\multirow{2}{*}{FSST ineq} & $2 \times$ & \multirow{2}{*}{2.392} & 3000 & 2.562 & 1.000 & 0.952 & 0.952 & 101.97 (24.86) \\
& $2$ & & 12000 & 2.476 & 1.000 & 0.952 & 0.960 & 125.77 (30.95) \\
\bottomrule
\end{tabular}
\caption{Summary results for confidence intervals under $Q\in \mathbf Q_{\rm enc}$ for $\mathcal Y = \{-1, 0, 1\}$ and $\mathcal Y = \{-1, 0, 1, 2\}$. Each row reports the length of the identified set, average length of the 95\% confidence intervals, coverage rate of the true parameter, coverage rates of the lower and upper endpoints of the identified set, and average runtime in seconds for obtaining one confidence region (with standard deviation of runtime in parentheses).}
\label{tab:BP_MC_3Y_summary}
\end{table}

Table \ref{tab:BP_MC_3Y_summary} reports the properties of 95\% confidence intervals obtained through ``FSST latent" and ``FSST ineq"  for samples of size $n = 3000$ and $n = 12000$.
For both the design with $\mathcal Y = \{-1,0,1\}$ and the design with $\mathcal Y = \{-1,0,1,2\}$, the confidence regions for both procedures exhibit remarkably similar average lengths and coverage properties. However, the computational cost of ``FSST latent" is larger than that of ``FSST ineq," reflecting the larger dimension of $A$ in the former procedure.
In Table \ref{tab:BP_MC_2Y_summary} in Appendix \ref{app:simdetails}, we consider a related design and find the ``FSST ineq" confidence intervals are sometimes shorter than those based on ``FSST latent" while retaining a lower computational cost. These results illustrate that, when available, methods based on closed-form expressions can be preferable, especially for larger specifications of $A$. 
We again emphasize, however, that closed-form expressions are not available for most of the parameters studied in this paper.
We also note that ``FSST latent'' is a valid test jointly for the value of the parameter as well as correct model specification, whereas ``FSST ineq'' is only for testing the value of the parameter and only when the model is correctly specified.

\section{Impact of Therapy on Reducing Crime} \label{sec:app}

In an important study, \cite{blattman2017reducing} examined the impact of cognitive behavioral therapy and cash transfers on criminal behavior through an RCT conducted in Liberia. 
Their analysis showed that intent-to-treat effects of cash transfers or therapy alone led to reductions in criminal behavior in the short term (2--5 weeks following randomization), but that these effects disappeared over the long term.
In contrast, offering both cash transfers and therapy led to intent-to-treat effects that persisted in the long term (12--13 months following randomization), which suggests an unexpected and important long-term complementarity between these treatments in reducing criminal behavior.

We employ our theoretical results to go beyond intent-to-treat effects in \cite{blattman2017reducing} and examine whether the long-term importance of receiving both cash and therapy (vs.\ only cash or only therapy) is present in average and distributional causal effects of treatment receipt as well, and, if so, for which subpopulations.
The experiment in \cite{blattman2017reducing} employed a $2 \times 2$ factorial design with the factors being cognitive behavioral therapy and cash transfer. 
We use $Z \in \mathcal Z = \{0, t, c, tc\}$ to denote the assignment to control, therapy, cash, and both therapy and cash, and $D \in \mathcal D = \{0, t, c, tc\}$ to denote the take-up\footnote{We define compliance to the therapy assignment to be ``attending at least 8 days of therapy,'' which is the first of the two alternative definitions of compliance used in \citep[Appendix A.4]{blattman2017reducing}.} of the corresponding treatment.
The sample size is $n = 999$. Table \ref{tab:blattman_DZ} presents the cell count of each pair of values of the treatment and the instrument.

\begin{table}[t!]
\centering
\begin{tabular}{crrrr}
\toprule
& \multicolumn{4}{c}{$Z$} \\
\cmidrule{2-5}
$D$ & 0 & $t$ & $c$ & $tc$ \\
\midrule
0 & 220 &  28 &   6 &   4 \\
$t$ &   0 & 252 &   0 &   0 \\
$c$ &   0 &   0 & 244 &  18 \\
$tc$ &   0 &   0 &   0 & 227 \\
\bottomrule
\end{tabular}
\caption{Cell counts for each pair of values of the treatment and the instrument.}\label{tab:blattman_DZ}
\end{table}
In order to focus on the benefits of receiving both treatments, we examine the incremental effect of receiving therapy in addition to cash and the incremental effect of receiving cash in addition to therapy.
Among a wide variety of outcomes related to crime and violence considered by \cite{blattman2017reducing}, we focus on the binary indicator of whether they engaged in any drug-selling activities in the short term and the long term.

Throughout our analysis, we maintain the assumption that noncompliance for either therapy or cash is one-sided.
Because the design is factorial, this means that if the subject is assigned both therapy and cash, then they can choose neither, one of them, or both. 
To formally map this restriction into our notation, we set $\mathcal R$ to equal
$$\mathcal R_{{\rm 1s}} = \{(r_o,r_t) : d(0) = 0, d(t) \in \{0,t\}, d(c) \in \{0,c\}, d(tc)\in \{0,t,c,tc\}\}, $$
where recall $r_o = (y(d) : d\in \mathcal D)$ and $r_t = (d(z) : z\in \mathcal Z)$.
We also consider complementing the one-sided noncompliance restriction with the assumption of harmless treatment --- i.e., that receiving cash transfers, therapy, or both, does not induce individuals into selling drugs. 
Since $Y(d)$ is here an indicator for whether the individual sold drugs, imposing the no-harm restriction in addition to one-sided noncompliance corresponds to setting $\mathcal R$ to equal
$$\mathcal R_{{\rm 1s, nh}} = \{(r_o,r_t) \in \mathcal R_{\rm 1s}: y(0)\geq y(d) \geq y(tc) \text{ for } d\in \{t,c\}\}. $$
Specification tests based on Corollary \ref{cor:validity} for both of these restrictions fail to reject the null hypothesis that the restrictions are valid at conventional significance levels.
In the case of one-sided noncompliance, we obtain $p$-values of one when $Y$ is the short-term outcome, and also of one when $Y$ is the long-term outcome. When we additionally impose the no-harm restriction, the $p$-values decrease to 0.91 when $Y$ is the short-term outcome, and to $0.18$ when $Y$ is the long-term outcome.

For each of these sets of restrictions, and different subsets of generalized principal strata $\mathcal R^\prime$, we study average and distributional causal effects by conducting inference on 
\begin{align}\label{app:causaldef}
\theta_{\rm ATE}(tc,d|\mathcal R^\prime) & = E_Q[Y(tc) - Y(d)|R\in \mathcal R^\prime]    \notag \\
\theta_{\rm Pr}(tc,d|\mathcal R^\prime) & = Q\{Y(tc) < Y(d)|R\in \mathcal R^\prime\}
\end{align}
for $d\in \{t,c\}$.
Thus, $\theta_{\rm ATE}(tc,d|\mathcal R^\prime)$ reveals the average incremental effect of receiving both cash and therapy over just one of them, while $\theta_{\rm Pr}(tc,d|\mathcal R^\prime)$ informs us about the corresponding distributional effects.
As subsets of generalized principal strata $\mathcal R^\prime$ we consider:
\begin{itemize}
    \item[] Unrestricted: $\mathcal R' = \mathcal R$.
    \item[] Pairwise compliers: $\mathcal R' = \mathcal R_{j, k}$ for $\mathcal R_{j, k} = \{(r_o,r_t)\in \mathcal R:  d(j) = j, d(k) = k\}$.
    \item[] Full compliers: $\mathcal R' = \mathcal R_{\mathrm{full}}$ for $\mathcal R_{\mathrm{full}} = \{(r_o,r_t) \in \mathcal R: d(j) = j \text{ for all } j\in \mathcal Z\}$.
\end{itemize}

\begin{table}[t!]
\centering
{\small
\begin{tabular}{ccllll}
\midrule
\multicolumn{6}{c}{One-Sided Noncompliance ($\mathcal R = \mathcal R_{\rm{1s}}$)}  \\ \midrule
Strata $(\mathcal R^\prime)$ &  Outcome & $\theta_{\mathrm{ATE}}(tc,t|\mathcal R^\prime)$ &  $\theta_{\mathrm{ATE}}(tc,c|\mathcal R^\prime)$ &  $\theta_{\mathrm{Pr}}(tc,t|\mathcal R^\prime)$ &  $\theta_{\mathrm{Pr}}(tc,c|\mathcal R^\prime)$ \\
\cmidrule(lr){1-6}
Unrestricted & Short Term & [-0.122, 0.158] & [-0.098, 0.100] & [0.000, 0.243] & [0.000, 0.213] \\
$(\mathcal R^\prime = \mathcal R)$ & Long Term & [-0.159, 0.116] & [-0.198, 0.020] & [0.000, 0.258] & [0.048, 0.298] \\
\addlinespace[0.75em]
Pairwise Compliers & Short Term & [-0.161, 0.208] & [-0.118, 0.044] & [0.000, 0.165] & [0.000, 0.230] \\
$(\mathcal R^\prime = \mathcal R_{j,k})$ & Long Term & [-0.198, 0.157] & [-0.232, -0.047] & [0.000, 0.196] & [0.052, 0.305]  \\
\addlinespace[0.75em]
Full Compliers & Short Term & [-0.165, 0.213] & [-0.253, 0.204] & [0.000, 0.164] & [0.000, 0.257]  \\
$(\mathcal R^\prime = \mathcal R_{\rm full})$ & Long Term & [-0.205, 0.181] & [-0.338, 0.058] & [0.000, 0.202] & [0.000, 0.340]  \\
\addlinespace[1.25em]
\multicolumn{6}{c}{One-sided Noncompliance + No-Harm Restriction ($\mathcal R = \mathcal R_{\rm{1s,nh}}$)} \\ \midrule
Strata $(\mathcal R^\prime)$ &  Outcome & $\theta_{\mathrm{ATE}}(tc,t|\mathcal R^\prime)$ &  $\theta_{\mathrm{ATE}}(tc,c|\mathcal R^\prime)$ &  $\theta_{\mathrm{Pr}}(tc,t|\mathcal R^\prime)$ &  $\theta_{\mathrm{Pr}}(tc,c|\mathcal R^\prime)$ \\
\cmidrule(lr){1-6}
Unrestricted & Short Term & [-0.069, 0.000] & [-0.109, 0.000] & [0.000, 0.069] & [0.000, 0.122] \\
$(\mathcal R^\prime = \mathcal R)$  & Long Term & [-0.116, 0.000] & [-0.148, -0.024] & [0.000, 0.112] & [0.029, 0.147] \\
\addlinespace[0.75em]
Pairwise Compliers & Short Term & [-0.078, 0.000] & [-0.131, 0.000] & [0.000, 0.091] & [0.000, 0.130] \\
$(\mathcal R^\prime = \mathcal R_{j,k})$ & Long Term & [-0.145, 0.000] & [-0.169, -0.025]& [0.000, 0.142] & [0.028, 0.165] \\
\addlinespace[0.75em]
Full Compliers & Short Term & [-0.062, 0.000] & [-0.139, 0.000] & [0.000, 0.063] & [0.000, 0.142] \\
$(\mathcal R^\prime = \mathcal R_{\rm full})$ & Long Term & [-0.146, 0.000] & [-0.183, -0.011] & [0.000, 0.147] & [0.007, 0.180] \\
\bottomrule
\end{tabular}}
\caption{95\% confidence intervals for causal parameters $\theta_{\rm ATE}(tc,d|\mathcal R^\prime)$ and $\theta_{\rm Pr}(tc,d|\mathcal R^\prime)$ (as defined in \eqref{app:causaldef}) for different subsets of generalized principal strata $\mathcal R^\prime$. Pairwise complier subset $\mathcal R_{j,k}$ corresponds to the $(j,k)$ in the corresponding causal parameter --- e.g., $\mathcal R_{j,k} = \mathcal R_{tc,t}$ in column 3 and $\mathcal R_{j,k} = \mathcal R_{tc,c}$ in column 4.  Confidence regions were obtained by inverting the test in Theorem \ref{thm:inference} after removing linearly dependent rows (as in Remark \ref{rm:dropzdisc}), and using two thousand bootstrap draws, the default choice of penalty $\lambda_n$ in the \texttt{lpinfer} R package, and the same machine configuration as in the simulations.}
\label{tab:blattman_empirical}
\end{table}

The first panel in Table \ref{tab:blattman_empirical} reports $95\%$ confidence intervals for the different causal parameters while imposing only the one-sided noncompliance restriction $(i.e., \mathcal R = \mathcal R_{\rm 1s})$.
We find no evidence that receiving both cash and therapy is better in the short term than receiving cash alone or   therapy alone: all our confidence regions for causal effects based on pairwise comparisons between receiving both therapy and cash against receiving cash alone or therapy alone cover zero.
However, we find evidence of long-term benefits of receiving  therapy in addition to cash, relative to receiving  cash alone.
In particular, we find a close to significant average treatment effect in the population (the $p$-value for testing $\theta_{\rm ATE}(tc,c|\mathcal R) = 0$ is 0.126; see Table \ref{tab:blattman_empirical_pval} in Appendix \ref{app:simdetails}) and that at least $5\%$ of the population would not sell drugs in the long term when receiving both cash and therapy, but would if receiving cash alone (i.e., $\theta_{\rm Pr}(tc,c|\mathcal R)$).
These effects are accentuated among individuals that are both $tc$-compliers and $c$-compliers (i.e., $\mathcal R' = \mathcal R_{tc,c}$), but not statistically significant among the group of full compliers (i.e., $\mathcal R' = \mathcal R_{\rm full})$.
The results for the long-term causal effects of receiving both therapy and cash, relative to receiving therapy alone, are inconclusive. 
While these results are compatible with the long-term positive effects of receiving both therapy and cash over receiving cash alone being driven by the benefits of therapy itself, in unreported results we find no statistically significant evidence that receiving therapy alone is beneficial over control.
Our analysis therefore suggests an important complementarity in offering both cash and therapy together, in particular over offering cash alone and for the subgroup of individuals that are both $tc$-compliers and $c$-compliers.

In the second panel of Table \ref{tab:blattman_empirical}, we strengthen our restrictions by imposing both one-sided noncompliance and the no-harm assumption --- i.e., we set $\mathcal R = \mathcal R_{\rm 1s, nh}$.
The conclusions remain largely in line with those obtained under one-sided noncompliance only, though the additional no-harm restriction proves informative in shortening our confidence intervals.
For example, the confidence interval for the long-term average treatment effect of receiving both cash and therapy instead of cash alone now excludes zero both for the entire population (i.e., $\mathcal R^\prime = \mathcal R$) and for the set of full compliers (i.e., $\mathcal R^\prime = \mathcal R_{\rm full}$).
We caution, however, that the data are less supportive of imposing both one-sided noncompliance and no-harm than of imposing one-sided noncompliance only --- recall the specification test $p$-values were $0.18$ and $1$ respectively when employing long-term outcomes.

\begin{table}[t!]
\centering
{\small
\begin{tabular}{ccllll}
\toprule
\multicolumn{6}{c}{One-Sided Noncompliance ($\mathcal R = \mathcal R_{\rm{1s}}$)}  \\ \midrule
Strata $(\mathcal R^\prime)$
    & Outcome
    & $\theta_{\mathrm{ATE}}(tc,t | \mathcal R^\prime)$
    & $\theta_{\mathrm{ATE}}(tc,c | \mathcal R^\prime)$
    & $\theta_{\mathrm{Pr}}(tc,t | \mathcal R^\prime)$
    & $\theta_{\mathrm{Pr}}(tc,c | \mathcal R^\prime)$ \\
\cmidrule(lr){1-6}
Unrestricted + No Harm
    & Short Term
    & [-0.309, 0.000]
    & [-0.272, 0.000]
    & [0.000, 0.308]
    & [0.000, 0.273] \\
$(\mathcal R^\prime =  \mathcal R^{+}_{tc,k})$
    & Long Term
    & [-0.322, 0.000]
    & [-0.362, -0.045]
    & [0.000, 0.322]
    & [0.047, 0.361] \\
\addlinespace[0.75em]
Pairwise Compliers + No Harm
    & Short Term
    & [-0.198, 0.000]
    & [-0.270, 0.000]
    & [0.000, 0.198]
    & [0.000, 0.268] \\
$(\mathcal R^\prime = \mathcal R_{j,k} \cap \mathcal R^{+}_{tc,k})$
    & Long Term
    & [-0.227, 0.000]
    & [-0.349, -0.048]
    & [0.000, 0.228]
    & [0.048, 0.349] \\
\addlinespace[0.75em]
Full Compliers + No Harm
    & Short Term
    & [-0.200, 0.000]
    & [-0.312, 0.000]
    & [0.000, 0.197]
    & [0.000, 0.306] \\
$(\mathcal R^\prime = \mathcal R_{\rm full} \cap \mathcal R^{+}_{tc,k})$
    & Long Term
    & [-0.238, 0.000]
    & [-0.395, 0.000]
    & [0.000, 0.234]
    & [0.000, 0.391] \\
\bottomrule
\end{tabular}
}
\caption{95\% confidence intervals for causal parameters $\theta_{\mathrm{ATE}}(tc, d | \mathcal R')$ and $\theta_{\mathrm{Pr}}(tc, d | \mathcal R')$. Pairwise complier subset $\mathcal R_{j,k}$ corresponds to the $(j,k)$ in the corresponding causal parameter --- e.g., $\mathcal R_{j,k} = \mathcal R_{tc,t}$ in column 3 and $\mathcal R_{j,k} = \mathcal R_{tc,c}$ in column 4. Restriction $\mathcal R_{tc,k}^+$ corresponds to the $k$ in the corresponding causal parameter --- e.g., $\mathcal R_{tc,k}^+ = \mathcal R_{tc,t}^+$ in column 3 and $\mathcal R_{tc,k}^+ = \mathcal R_{tc,c}^+$ in column 4. Confidence regions were obtained by inverting the test in Theorem \ref{thm:inference} after removing linearly dependent rows (as in Remark \ref{rm:dropzdisc}) and using two thousand bootstrap draws, the default choice of the penalty parameter $\lambda_n$ in the \texttt{lpinfer} R package, and the same machine configuration as in the simulations.}
\label{tab:blattman_condY_onesided_vs0}
\end{table}

As an alternative to the no-harm restriction, we may maintain one-sided noncompliance but instead examine causal effects conditional on not being harmed by the treatment.
These parameters correspond to conditioning on generalized principal strata groups $\mathcal R^\prime$ that incorporate a restriction on the outcome response type $r_o = (y(d):d\in \mathcal D)$. Formally, set
$$\mathcal R^{+}_{tc,k} = \{(r_o,r_t)\in \mathcal R : y(tc) \leq y(k)\}$$
for $k\in \{t,c\}$ to indicate the set of individuals who see a (weakly) better outcome under treatment $tc$ than under alternative treatment $k$.
Notice this generalized principal stratum has probability one under the no-harm restriction, but not necessarily under only one-sided noncompliance.
Table \ref{tab:blattman_condY_onesided_vs0} reports causal effects for this generalized principal stratum under the one-sided noncompliance restriction only.
Interestingly, the confidence interval for the average treatment effect of receiving both cash and therapy rather than  cash alone, among individuals who would not be harmed by receiving both rather than cash alone  (i.e., $\mathcal R^\prime = \mathcal R_{tc,c}^+$) now includes only strictly negative (beneficial) values.
With an upper endpoint of $-4.5\%$, this confidence interval is more informative than the confidence region obtained under the no-harm assumption, and, perhaps more surprisingly, almost identical to the upper end for the average treatment effect among individuals that are $tc$-compliers and $c$-compliers.
The similarity suggests that the group of individuals who are both $tc$-compliers and $c$-compliers may heavily overlap with the group of individuals who benefit from receiving therapy and cash instead of only cash. Indeed, the $95\%$ confidence interval for $Q\{R\in\mathcal R_{tc,c}\mid Y(tc)\leq Y(c)\}$ has a lower endpoint of $85.1\%$, providing evidence of substantial overlap between individuals who would not be harmed by receiving therapy in addition to cash and individuals who are both $tc$-compliers and $c$-compliers.

\clearpage
\begin{appendices}

\section{Proof of Main Results} \label{app:proofs}
To reduce notational clutter, throughout the appendix we suppress the dependence of $\hat P_n$ on $n$ and write $\hat P$ instead.

\subsection{Proof of Theorem \ref{thm:equiv}} \label{app:equiv}

We first establish Theorem \ref{thm:equiv} and subsequently show that the result continues to hold if certain rows of $A$ and $\beta(P)$ are removed, as stated in Remark \ref{rm:dropzdisc}.

\begin{proof}[\sc Proof of Theorem \ref{thm:equiv}]
Suppose $P \in \mathbf P_0$. Then, it follows that $\theta_0 \in \Theta_0(P, \mathbf Q)$, so there exists a $Q \in \mathbf Q_0(P, \mathbf Q)$ such that $\theta_0 = \theta(Q)$ and $Q\{R\in \mathcal R'\} > 0$. Recall $Q \in \mathbf Q_0(P, \mathbf Q)$ if and only if $Q \in \mathbf Q$ and $P = Q T^{-1}$. Since $Q \in \mathbf Q$, it satisfies both Assumptions \ref{as:exog} and \ref{ass:generalizedstrata}. Because $Q$ satisfies Assumption \ref{as:exog} and $P = Q T^{-1}$, we can conclude that
\begin{equation} \label{eq:p-q}
 P \{Y = y, D = d | Z = z\}= Q \{Y(d) = y, D(z) = d | Z = z\} = Q \{Y(d) = y, D(z) = d\}~.
\end{equation}
Let $q$ be the column vector with $|\mathcal R|$ elements indexed by $r$, where the $r$ element is $q(r) = Q \{R = r\}$. It follows from result \eqref{eq:p-q} and the definition of $A$ that the first $|\mathcal M|$ rows of $A q$ equal the first $|\mathcal M|$ rows of $\beta(P)$. 
Next, note that because $Q\in \mathbf Q_0(P,\mathbf Q)$ implies $Q$ satisfies Assumption \ref{ass:generalizedstrata} and hence $Q\{R\in \mathcal R\} = 1$, it follows that
\[ \sum_{r \in \mathcal R} q(r) = 1~, \]
and therefore the $|\mathcal M| + 1$ row of $A q$ equals the $|\mathcal M|+1$ row of $\beta(P)$. 
Finally, since $\theta_0 = \theta(Q)$ by hypothesis, the definition of the conditional expectation yields that
\begin{equation} \label{eq:ratio}
\theta_0 = \frac{E_Q[g(R) I \{R \in \mathcal R'\}]}{Q \{R \in \mathcal R'\}}~,   
\end{equation}
where recall $Q$ satisfies $Q\{R\in \mathcal R'\} > 0$. Rearranging the expression in \eqref{eq:ratio} we obtain that $E_Q [(g(R) - \theta_0) I \{R \in \mathcal R'\}] = 0$ or, equivalently, 
\[ \sum_{r \in \mathcal R'} (g(r) - \theta_0) q(r) = 0~. \]
It follows from the definition of $A$ that the $|\mathcal M| + 2$ row of $A q$ equals the $|\mathcal M|+2$ row of $\beta(P)$. 
Hence, $\mathbf P_0 \subseteq \tilde{\mathbf P}_0$ for $\tilde{\mathbf P}_0 \equiv \{P \in \mathbf P : Ax = \beta(P) \text{ for some } x \geq 0 \}$. 

To conclude, it only remains to show that $\tilde {\mathbf P}_0 \subseteq {\rm cl}_{\mathbf P}(\mathbf P_0)$ whenever $\mathbf P_0 \neq \emptyset$.
To this end, fix any $\tilde P \in \tilde{\mathbf P}_0$.
Since $Aq = \beta(\tilde P)$ for some $q \geq 0$, $\tilde P\{(Y, D, Z) \in \mathcal M\} = 1$ and the definition of $A$ and $\beta(\tilde P)$ imply there exists a $\tilde Q \in \mathbf Q$ such that $\tilde P = \tilde Q T^{-1}$ and $E_{\tilde Q}[g(R) I \{R \in \mathcal R'\}] = \theta_0 \tilde Q \{R \in \mathcal R'\}$. 
By assumption, $\mathbf P_0$ is nonempty and therefore Lemma \ref{lem:P0P} implies $\mathbf P_0(\tilde P) \equiv \{P\in \mathbf P_0 : P_Z=\tilde P_Z\}$ is nonempty as well, where $P_Z$ denotes the marginal distribution of $Z$ under $P$.
Hence, there exists a $P \in \mathbf P_0$ such that $P_Z=\tilde P_Z$, $P = Q T^{-1}$ for some $Q \in \mathbf Q$, $Q \{R \in \mathcal R'\} > 0$, and $\theta(Q) = \theta_0$. Define $P_n = n^{-1} P + (1 - n^{-1}) \tilde P$ and $Q_n = n^{-1} Q + (1 - n^{-1}) \tilde Q$. 
It then follows that $P_n = Q_n T^{-1}$, $Q_n \{R \in \mathcal R'\} > 0$, and $E_{Q_n}[g(R) I \{R \in \mathcal R'\}] = \theta_0 Q_n \{R \in \mathcal R'\}$, which implies that $\theta(Q_n) = \theta_0$. 
Moreover, since $Q = Q_R\times P_Z$ and $\tilde Q = \tilde Q_R\times P_Z$ for some distributions $Q_R$ and $\tilde Q_R$, it follows that $Q_n = (n^{-1} Q_R + (1-n^{-1})\tilde Q_R)\times P_Z$ and therefore that $R\indep Z$ under $Q_n$.
We conclude that $P_n \in \mathbf P_0$ for each $n$ and, since $P_n \to \tilde P$  as $n \to \infty$, that $\tilde P \in \mathrm{cl}_{\mathbf P}(\mathbf P_0)$ implying $\tilde {\mathbf P}_0 \subseteq \text{cl}_{\mathbf P}(\mathbf P_0)$.
\end{proof}

We next show that the conclusion of Theorem \ref{thm:equiv} continues to hold if we employ a subset of the restrictions implied by $A$.
To formalize this result, recall that the first $|\mathcal M|$ rows of $A$ and $\beta(P)$ are indexed by $(y,d,z)\in \mathcal M$. 
For each $z\in \mathcal Z$ we select an arbitrary $(y_z,d_z) \in \mathcal M_z$ and define $A^{\rm s}$ and $\beta^{\rm s}(P)$ to be the matrix and vector obtained by removing the rows $((y_z,d_z,z) : z\in \mathcal Z)$ from $A$ and $\beta(P)$ respectively. Further define 
\begin{align*}
\tilde{\mathbf P}_0 & \equiv \{P\in \mathbf P :Ax   = \beta(P) \text{ \rm for some } x \geq 0\} \\
\tilde{\mathbf P}_0^{\rm s} & \equiv \{P\in \mathbf P :A^{\rm s}x   = \beta^{\rm s}(P) \text{ \rm for some } x \geq 0\},
\end{align*}
and note that Theorem \ref{thm:equiv} shows $\mathbf P_0 \subseteq \tilde{\mathbf P}_0$ and that $\tilde{\mathbf P}_0 \subseteq {\rm cl}_{\mathbf P}(\mathbf P_0)$ whenever $\mathbf P_0 \neq \emptyset$.
Our next lemma establishes that $\tilde {\mathbf P}_0 = \tilde {\mathbf P}_0^{\rm s}$, which implies that the conclusion of Theorem \ref{thm:equiv} continues to hold if we replace $A$ and $\beta(P)$ with $A^{\rm s}$ and $\beta^{\rm s}$ respectively.

\begin{lemma}\label{lm:dropz}
Let $\mathbf P$ denote the set of all distributions of $(Y,D,Z)$ satisfying $P\{(Y,D,Z)\in \mathcal M\} = 1$ and $P\{Z=z\} > 0$ for all $z\in \mathcal Z$. Then, it follows that $\tilde{\mathbf P}_0  = \tilde{\mathbf P}_0^{\rm s}$.
\end{lemma}

\begin{proof}[\sc Proof of Lemma \ref{lm:dropz}]
Since any $x$ satisfying $A x = \beta(P)$ must also satisfy $A^{\rm s} x = \beta^{\rm s}(P)$, it immediately follows that $\tilde{\mathbf P}_0 \subseteq \tilde {\mathbf P}_0^{\rm s}$.
For the reverse inclusion note that, by definition of $\mathcal M_z$ (as in \eqref{eq:Mzdef}), for any $r = (r_o,r_t)\in \mathcal R$ there exists a $(y,d)\in \mathcal M_z$ such that $r_o =(y(d) : d\in \mathcal D)$ and $r_t = (d(z) : z\in \mathcal Z)$ satisfy $I\{y(d) = y, d(z) = d\} = 1$. Since such $(y,d)\in \mathcal M_z$ is necessarily unique, it follows from the definition of $A$ that for any $r\in \mathcal R$ and $z\in \mathcal Z$ we have
\begin{equation}\label{lm:dropz1}
\sum_{(y,d) \in \mathcal M_z} A_{(y,d,z),r} = 1   .
\end{equation}
In particular, letting $A_{(y,d,z),\cdot}$ denote the $(y,d,z)$ row of the matrix $A$ it then follows that 
\begin{equation}\label{lm:dropz2}
\sum_{(y,d)\in \mathcal M_z} A_{(y,d,z),\cdot} = A_{|\mathcal M|+1,\cdot}    
\end{equation}
where the equality follows from the $|\mathcal M|+1$ row of the matrix $A$ having all entries equal to one. 
Next select an arbitrary $P\in \tilde{\mathbf P}_0^{\rm s}$ and note that by definition there exists an $x_0 \geq 0$ satisfying $A^{\rm s} x_0 = \beta^{\rm s}(P)$.
Since the system $A^{\rm s} x_0 = \beta^{\rm s}(P)$ still contains the adding up constraint $A_{|\mathcal M|+1,\cdot} x_0 = 1$, result \eqref{lm:dropz2} allows us to conclude that
\begin{equation}\label{lm:dropz3}
1 = \sum_{(y,d)\in \mathcal M_z} A_{(y,d,z),\cdot}x_0 = \sum_{(y,d)\in \mathcal M_z \setminus \{(y_z,d_z)\}} \beta_{(y,d,z)}(P) + A_{(y_z,d_z,z),\cdot} x_0,
\end{equation}
where in the second equality we used that $A^{\rm s}x_0 = \beta^{\rm s}(P)$ implies that $A_{(y,d,z),\cdot}x_0 = \beta_{(y,d,z)}(P)$ for any $(y,d)\in \mathcal M_z\setminus \{(y_z,d_z)\}$ because those constraints remain in $A^{\rm s}$.
Moreover, since $P\{Z=z\} > 0$ for all $z\in \mathcal Z$ and $P\{(Y,D,Z)\in \mathcal M\} = 1$ due to $P\in \mathbf P$, we also have
\begin{equation}\label{lm:dropz4}
\sum_{(y,d)\in \mathcal M_z} \beta_{(y,d,z)}(P) = \sum_{(y,d)\in \mathcal M_z} P\{Y =y, D = d|Z =z \} = 1.
\end{equation}
Combining \eqref{lm:dropz3} and \eqref{lm:dropz4} we obtain that $A_{(y_z,d_z,z),\cdot} x_0 = \beta_{(y_z,d_z,z)}(P)$ for any $z\in \mathcal Z$.
It follows that $Ax_0 = \beta(P)$ and, since $x_0 \geq 0$, that $P\in \tilde {\mathbf P}_0$ as desired.
\end{proof}

\subsection{Proof of Corollary \ref{cor:validity}}

\begin{proof}[\sc Proof of Corollary \ref{cor:validity}]
As discussed in the paragraph before the statement of Corollary \ref{cor:validity}, setting $g = 1$, $\theta_0 = 1$, and $\mathcal R' = \mathcal R$ implies that
\[
\Theta_0(P,\mathbf Q) =
\begin{cases}
\{1\} & \text{ if } \mathbf Q_0(P,\mathbf Q) \neq \emptyset \\
\emptyset & \text{ otherwise }
\end{cases}
~.
\]
Therefore,
\begin{equation} \label{eq:validity-P0-Q}
\mathbf P_0 = \{P \in \mathbf P : 1 \in \Theta_0(P,\mathbf Q)\} = \{P \in \mathbf P : \mathbf Q_0(P,\mathbf Q) \neq \emptyset \}~.    
\end{equation}
Next, note that for this choice of $g$, $\theta_0$, and $\mathcal R'$, we have
\begin{equation} \label{eq:validity-tilde}
\{P \in \mathbf P: A x = \beta(P) \text{ for some } x \geq 0\} = \{P \in \mathbf P: \tilde A x = \tilde \beta(P) \text{ for some } x \geq 0\}~,   
\end{equation}
because the equality in the last row of the equation $A x = \beta(P)$ is always satisfied. Theorem \ref{thm:equiv} implies that
\begin{equation} \label{eq:validity-P0-LP}
\mathbf P_0 \subseteq \{P \in \mathbf P: A x = \beta(P) \text{ for some } x \geq 0\}~.
\end{equation}
On the other hand, fix some $P \in \mathbf P$ such that $\tilde A q = \tilde \beta(P)$ for some $q \geq 0$. It follows from $P\{(Y, D, Z) \in \mathcal M\} = 1$ and the definition of $\tilde A$ and $\tilde \beta(P)$ that there exists a $Q \in \mathbf Q$ such that $P = Q T^{-1}$, so $\mathbf Q_0(P, \mathbf Q) \neq \emptyset$. As a result,
\begin{equation} \label{eq:validity-P0-LP-reverse}
\{P \in \mathbf P: A x = \beta(P) \text{ for some } x \geq 0\} \subseteq \{P \in \mathbf P : \mathbf Q_0(P,\mathbf Q) \neq \emptyset \}~.
\end{equation}
The claim of the corollary then follows from \eqref{eq:validity-P0-Q}--\eqref{eq:validity-P0-LP-reverse}.
\end{proof}

\subsection{Proof of Theorem \ref{thm:inference}}

In this Appendix, we provide the proof of Theorem \ref{thm:inference} and additionally justify the validity of a test based on employing a subset of the restrictions in $A$, as discussed in Remark \ref{rm:dropzdisc}.
For the latter purpose, we follow the notation of Lemma \ref{lm:dropz} and for each $z\in \mathcal Z$ we select a point $(y_z,d_z)\in \mathcal M_z$.
Recalling that the first $|\mathcal M|$ rows of $A$ and $\beta(P)$ are indexed by $(y,d,z)\in \mathcal M $, we let $A^{\rm s}$ and $\beta^{\rm s}(P)$ denote the matrix and vector obtained by removing the rows $((y_z,d_z,z) : z\in \mathcal Z)$ from $A$ and $\beta(P)$ respectively. 
An ``${\rm s}$" superscript on other quantities, such as $\hat \beta^{\rm s}_n$ in place of $\hat \beta_n$, indicates an analogous modification to the case in which the restrictions indexed by $((y_z,d_z,z) : z\in \mathcal Z)$ are dropped.
 
Throughout, we let $\psi(V,P)$ be a vector of dimension $|\mathcal M| +2$.
Indexing the first $|\mathcal M|$ coordinates of $\psi(V,P)$ by $(y,d,z)\in \mathcal M$, we define $\psi_{ydz}(V,P)$ to be given by
\begin{equation}\label{eq:psidef}
\psi_{ydz}(V,P) \equiv \frac{1}{P_z} I \{Y=y,D=d,Z=z\} - \frac{P_{ydz}}{P_z^2}I\{Z=z\}~,    
\end{equation}
and let the final two coordinates of $\psi(V,P)$ equal zero. The random variable $\psi(V,P)$ is the influence function of $\sqrt n\{\hat \beta_n - \beta(P)\}$ and will play an important role in our analysis.
Similarly note $\psi^{\rm s}(V,P)$ is the influence function of $\sqrt n\{\hat \beta_n^{\rm s} - \beta^{\rm s}(P)\}$.

The proof of Theorem \ref{thm:inference} requires a final assumption that ensures that our consistent approximation for the distribution of our test statistic translates into consistency of the critical values.
To formally state this assumption, we define the set
\begin{equation*}
    \mathcal V(P) \equiv \{s\in \mathbf R^{|\mathcal M|+2} : A^\dagger s\leq 0, \|\Omega(P)(AA^\prime)^\dagger s\|_1\leq 1\},
\end{equation*}
which may be understood as the population analogue to $\hat {\mathcal V}_n$.
We further let $\mathcal E(P)$ denote the extreme points of $(A A')^\dagger \mathcal V(P)$, set $\sigma^2(s,P)\equiv \text{Var}_P\{s^\prime \psi(V,P)\}$ for any $s\in \mathcal E(P)$, and define
\begin{equation*}
\bar \sigma(P) \equiv \max_{s\in \mathcal E(P)} \sigma(s,P), \hspace{0.3 in}\underline{\sigma}(P) \equiv \min_{s\in \mathcal E(P):\sigma(s,P) > 0} \sigma(s,P)~,
\end{equation*}
where we let $\underline{\sigma}(P) = +\infty$ if $\sigma(s,P) = 0$ for all $s\in \mathcal E(P)$.
Finally, for any $P\in \mathbf P$, we let $\Sigma(P) \equiv E_P[\psi(V,P)\psi(V,P)^\prime]$ and $\mathbb G(P) \sim N(0, \Sigma(P))$, and define
\begin{equation*}
\mathrm{m}(P) \equiv \mathrm{med} \left \{ \sup_{s \in \mathcal V(P)} \langle A^\dagger s, A^\dagger \mathbb G(P) \rangle \right \}~,
\end{equation*}
where for any random variable $V$, $\text{med}\{V\}$ denotes its median.
As discussed in the main text, in high-dimensional settings, consistency of the critical values requires a condition termed \emph{anti-concentration} by \cite{chernozhukov2015comparison}. Given the introduced notation, we state a final assumption that ensures there is sufficient anti-concentration in our problem.

\begin{assumption} \label{ass:tuning}
The sequence $\xi_n \equiv \log^{3/2}(n)\sqrt{|\mathcal M|}/\sqrt n \vee \lambda_n \sqrt{\log(|\mathcal M|)}$ satisfies
$$\sup_{P\in \mathbf P} \frac{\mathrm{m}(P) + \bar \sigma(P)}{\underline \sigma^2(P)} = o(\xi_n^{-1}).$$
\end{assumption}

Similarly, establishing the validity of a test based on $A^{\rm s}$ in place of $A$ requires us to impose a condition in which the quantities $\bar {\sigma}^{\rm s}(P)$, $\underline{\sigma}^{\rm s}(P)$, and ${\rm m}^{\rm s}(P)$ are understood as parallel modifications to the parameters $\bar \sigma(P)$, $\underline \sigma(P)$ and ${\rm m}(P)$. 

\begin{assumption} \label{ass:tuning:dropz}
The sequence $\xi_n \equiv \log^{3/2}(n)\sqrt{|\mathcal M|}/\sqrt n \vee \lambda_n \sqrt{\log(|\mathcal M|)}$ satisfies
$$\sup_{P\in \mathbf P} \frac{\mathrm{m}^{\rm s}(P) + \bar \sigma^{\rm s}(P)}{(\underline \sigma^{\rm s}(P))^2} = o(\xi_n^{-1}).$$
\end{assumption}

By dropping redundant rows, the asymptotic distribution of the test based on $A^{\rm s}$ (in place of $A$) depends on the asymptotic distribution of $\sqrt n(\hat \beta_n^{\rm s} - \beta^{\rm s}(P))$ (in place of $\sqrt n(\hat \beta_n - \beta(P))$).
Unlike the asymptotic variance of $\sqrt n(\hat \beta_n - \beta(P))$, the asymptotic variance of $\sqrt n(\hat \beta_n^{\rm s} - \beta^{\rm s}(P))$ is invertible once we exclude the final two coordinates.
This property is useful for obtaining lower level, but potentially conservative, sufficient conditions for the anti-concentration condition in Assumption \ref{ass:tuning:dropz}.
Our next lemma formalizes this observation.

\begin{lemma}\label{lm:anti}
If Assumptions \ref{ass:test}(i) and \ref{ass:test}(iii) hold, then it follows that Assumption \ref{ass:tuning:dropz} is satisfied provided  $\log^2(n)|\mathcal M|^{5/2} = o(|\mathcal Z|\sqrt n)$ and $\lambda_n \log(|\mathcal M|)|\mathcal M|^{2} = o(|\mathcal Z|)$.
\end{lemma}

The conclusion of Lemma \ref{lm:anti} implies that anti-concentration holds for the test based on $A^{\rm s}$ provided that $|\mathcal M|^{5/2} = o(|\mathcal Z| \sqrt n)$ (up to logs) and $\lambda_n$ tends to zero sufficiently fast.
An important implication of Lemma \ref{lm:anti} is that Assumption \ref{ass:tuning:dropz} is automatically satisfied for the test based on $A^{\rm s}$ in low dimensional approximations in which $|\mathcal M|$ is fixed provided $\lambda_n$ tends to zero.
Moreover, we note that the universal (in $A^{\rm s}$) bounds used by Lemma \ref{lm:anti} may be quite conservative for specific $A^{\rm s}$, leading to the condition $|\mathcal M|^{5/2} = o(|\mathcal Z| \sqrt n)$ (up to logs) being overly stringent.
Alternatively, Assumptions \ref{ass:tuning} and \ref{ass:tuning:dropz} can be dispensed with by modifying our critical value to include a ``correction factor'' introduced by \cite{andrews2013inference}. We opt to introduce Assumptions \ref{ass:tuning} and \ref{ass:tuning:dropz} instead, however, to avoid introducing an additional tuning parameter.

For completeness, we state the validity of a test based on $A^{\rm s}$, which we denote by $\phi_n^{\rm s}(\theta_0)$. 

\begin{theorem} \label{thm:inference:dropz}
Suppose Assumptions \ref{ass:test} and \ref{ass:tuning:dropz} hold, $\alpha \in (0, 0.5)$, $\log^3(n)|\mathcal M|/n = o(1)$, and $0 \leq \lambda_n \leq 1$ satisfies $\lambda_n\sqrt{\log(|\mathcal M|)} = o(1)$. Then it follows that  $$\limsup_{n \rightarrow \infty} \sup_{P \in \mathbf P_0} E_P[\phi_n^{\rm s}(\theta_0)] \leq \alpha.$$
\end{theorem}

We next provide the proofs of Theorems \ref{thm:inference}, \ref{thm:inference:dropz} and  Lemma \ref{lm:anti}. The proofs rely on numerous auxiliary lemmas that are stated and proven in Online Appendix Section \ref{sec:auxlemma}.

\begin{proof}[\sc Proof of Theorem \ref{thm:inference}] The proof follows the same arguments as those in the proof of Theorem 4.2 in \cite{fang2023inference} but with their coupling results (Lemmas A.4 and A.5) replaced by our improved coupling rates from Lemma \ref{lm:coup}.
Specifically, we note that Lemma \ref{lm:check} implies Assumptions 4.1, 4.2, 4.3, 4.4(v) and A.1 of \cite{fang2023inference} hold. 
Assumptions 4.4(i)--(iv) are not verified directly, because their only role in the proof of Theorem 4.2 in \cite{fang2023inference} is to ensure a coupling of the bootstrap statistic --- a condition that we instead establish directly in Lemma \ref{lm:coup}.
The arguments in \cite{fang2023inference} can therefore be applied with $r_n = \log^{3/2}(n)\sqrt{|\mathcal M|}/\sqrt n$ (the coupling rate of Lemma \ref{lm:coup}(i)) and $b_n = r_n$ (the coupling rate of Lemma \ref{lm:coup}(ii)).
\end{proof}

\begin{proof}[\sc Proof of Theorem \ref{thm:inference:dropz}] 
The proof follows from identical arguments to those in the proof of Theorem \ref{thm:inference} after noting that the assumptions in \cite{fang2023inference} holding for $\phi_n(\theta_0)$ imply they hold for $\phi^{\rm s}_n(\theta_0)$ as well.
In particular, note that the coupling results in Lemma \ref{lm:coup} holding for $\sqrt n(\hat \beta_n - \beta(P))$ and $\sqrt n(\hat \beta_n^* - \hat \beta_n)$ imply that they continue to hold when applied to the subvectors $\sqrt n(\hat \beta_n^{\rm s} - \beta^{\rm s}(P))$ and $\sqrt n(\hat \beta_n^{\rm s*} - \hat \beta_n^{\rm s})$.
Similarly, Assumption 4.1 in \cite{fang2023inference} is satisfied with $a_n = \log(|\mathcal M|)\sqrt{|\mathcal M|/n}$ by Assumption \ref{ass:test}(i), Assumption \ref{ass:test}(ii) implying $(\hat \beta_n^{\rm s} - \beta^{\rm s}(P))\in \text{range}(A^{\rm s})$ with probability tending to one uniformly in $P\in \mathbf P$, and Lemma \ref{aux:lin}.
Furthermore, Assumptions 4.2(i)(iii)(iv), 4.3,  4.4(v), and A.1 in \cite{fang2023inference} holding for $\phi_n(\theta_0)$, as shown in Lemma \ref{lm:check}, immediately imply they hold for $\phi_n^{\rm s}(\theta_0)$ as well.
Finally, Assumption 4.2(ii) in \cite{fang2023inference} holds by Lemma \ref{aux:smallC}(i).
Hence, all the assumptions of \cite{fang2023inference} verified to hold for $\phi_n(\theta_0)$ in the proof of Theorem \ref{thm:inference} continue to hold for $\phi_n^{\rm s}(\theta_0)$  as well, and the theorem can be shown using the same arguments in the proof of Theorem \ref{thm:inference}.
\end{proof}

\begin{proof}[\sc Proof of Lemma \ref{lm:anti}]
First select an arbitrary $s\in \mathcal E^{\rm s}(P)$, which we note by definition must satisfy $s = (A^{\rm s}(A^{\rm s})')^\dagger \tilde s$ for some $\tilde s \in \mathcal V^{\rm s}(P)$.
Note that if $\Omega^{\rm s}(P)(A^{\rm s} (A^{\rm s})')^\dagger \tilde s = 0$, we have
\begin{equation}\label{lm:anti1}
(\sigma^{\rm s}(s,P))^2 = \text{Var}_P\{\langle (A^{\rm s}(A^{\rm s})')^\dagger \tilde s, \mathbb G^{\rm s}(P)\rangle\} = \text{Var}_P\{ \langle \Omega^{\rm s}(P)(A^{\rm s}(A^{\rm s})')^\dagger \tilde s, \Omega^{\rm s}(P)^\dagger \mathbb G^{\rm s}(P)\rangle\}    =0
\end{equation}
where the second equality follows from $\Omega^{\rm s}(P)\Omega^{\rm s}(P)^\dagger \mathbb G^{\rm s}(P) = \mathbb G^{\rm s}(P)$ and the third equality from $\Omega^{\rm s}(P)(A^{\rm s}(A^{\rm s})')^\dagger \tilde s = 0$.
On the other hand, if $\Omega^{\rm s}(P)(A^{\rm s}(A^{\rm s})')^\dagger \tilde s \neq 0$, then 
\begin{equation}\label{lm:anti2}
\|\Omega^{\rm s}(P)(A^{\rm s}(A^{\rm s})')^\dagger \tilde s\|_1 = 1,    
\end{equation}
because otherwise it would be possible to write $s$ as a convex combination of zero and $(A^{\rm s}(A^{\rm s})')^\dagger \tilde s /\|\Omega^{\rm s}(P)(A^{\rm s}(A^{\rm s})')^\dagger \tilde s\|_1$, which contradicts $s$ being an extreme point of the set $(A^{\rm s}(A^{\rm s})')^\dagger \mathcal V^{\rm s}(P)$.
Therefore, again using that $\Omega^{\rm s}(P)\Omega^{\rm s}(P)^\dagger \mathbb G^{\rm s}(P) = \mathbb G^{\rm s}(P)$ we obtain
\begin{align}\label{lm:anti3}
(\underline{\sigma}^{\rm s}(P))^2 & \geq \inf_{\tilde s\in \mathcal V^{\rm s}(P) : \|\Omega^{\rm s}(P)(A^{\rm s}(A^{\rm s})')^\dagger \tilde s\|_1 =1 } \text{Var}_P\{\langle (A^{\rm s}(A^{\rm s})')^\dagger \tilde s, \mathbb G^{\rm s}(P)\rangle\} \notag \\
& =  \inf_{\tilde s\in \mathcal V^{\rm s}(P) : \|\Omega^{\rm s}(P)(A^{\rm s}(A^{\rm s})')^\dagger \tilde s\|_1 =1 } \text{Var}_P\{\langle\Omega^{\rm s}(P) (A^{\rm s}(A^{\rm s})')^\dagger \tilde s, \Omega^{\rm s}(P)^\dagger \mathbb G^{\rm s}(P)\rangle\} .
\end{align}
Next note that the range of $\Omega^{\rm s}(P)^\dagger \Sigma^{\rm s}(P)\Omega^{\rm s}(P)^\dagger$ includes all vectors whose last two coordinates equal zero by Lemma \ref{aux:smallC}(iii).
Hence, it follows that
\begin{equation}\label{lm:anti5}
\Omega^{\rm s}(P)(A^{\rm s}(A^{\rm s})')^\dagger \tilde s \in \text{range}(\Omega^{\rm s}(P)^\dagger \Sigma^{\rm s}(P)\Omega^{\rm s}(P)^\dagger)    
\end{equation}
for all $\tilde s$ and $P\in \mathbf P$ because the final two diagonal entries of $\Omega^{\rm s}(P)$ equal zero.
Therefore, Theorem 6.6.3(ii) in \cite{luenberger1997optimization} and \eqref{lm:anti5} imply $\Omega^{\rm s}(P)(A^{\rm s}(A^{\rm s})')^\dagger \tilde s$ is in the span of the eigenvectors corresponding to the nonzero eigenvalues of $\Omega^{\rm s}(P)^\dagger \Sigma^{\rm s}(P)\Omega^{\rm s}(P)^\dagger$. 
Hence, result \eqref{lm:anti3}, the inequality $\|a\|_1 \leq \sqrt{d}\|a\|_2$ for any $a \in \mathbf R^d$, and Lemma \ref{aux:smallC}(ii) yield
\begin{equation}\label{lm:anti6}
(\underline{\sigma}^{\rm s}(P))^2 \geq  \pi \times \frac{|\mathcal Z|}{|\mathcal M|^{2}}
\end{equation}
for some $\pi > 0$.
Moreover, the definition of $\mathcal V^{\rm s}(P)$, $\Omega^{\rm s}(P)\Omega^{\rm s}(P)^\dagger \mathbb G^{\rm s}(P) = \mathbb G^{\rm s}(P)$, Lemma \ref{aux:smallC}(i) and $\|a\|_2 \leq \|a\|_1$ for any $a\in \mathbf R^d$ imply uniformly in $P\in \mathbf P$ that
\begin{multline}\label{lm:anti7}
(\bar{\sigma}^{\rm s}(P))^2 \leq \sup_{\tilde s \in \mathcal V^{\rm s}(P)} \text{Var}_P\{\langle (A^{\rm s}(A^{\rm s})')^\dagger \tilde s, \mathbb G^{\rm s}(P)\rangle\} \\ \leq \sup_{\|\Omega^{\rm s}(P)(A^{\rm s}(A^{\rm s })')^\dagger \tilde s\|_1 \leq 1} \text{Var}_P\{\langle \Omega^{\rm s}(P)(A^{\rm s}(A^{\rm s})')^\dagger \tilde s, \Omega^{\rm s}(P)^\dagger \mathbb G^{\rm s}(P)\rangle\}  = O(1).
\end{multline}
Finally, note similar manipulations,  $((A^{\rm s})^\dagger)'(A^{\rm s})^\dagger = (A^{\rm s}(A^{\rm s})')^\dagger $, and H\"older's inequality yield
\begin{equation}\label{lm:anti8}
    \mathrm{m}^{\rm s}(P) \leq \text{med}\left\{\|\Omega^{\rm s}(P)^\dagger \mathbb G^{\rm s}(P)\|_\infty\right\} = O(\sqrt{\log(|\mathcal M|)}),
\end{equation}
where the equality holds uniformly in $P\in \mathbf P$ by a standard maximal inequality (see, e.g., Lemma 2.2.2 in \cite{van_der_vaart1996weak}). The claim of the lemma therefore follows from combining results \eqref{lm:anti6}, \eqref{lm:anti7}, and \eqref{lm:anti8} with the definition of $\xi_n$.
\end{proof}

\subsection{Proof of Theorem \ref{thm:W}}

\begin{proof}[\sc Proof of Theorem \ref{thm:W}]
Suppose $P \in \mathbf P_0^{\rm w}$. Then, it follows that $\theta_0 \in \Theta_0^{\rm w}(P, \mathbf Q^{\rm w})$, so there exists a $Q \in \mathbf Q_0^{\rm w}(P, \mathbf Q^{\rm w})$ such that $\theta_0 = \theta^{\rm w}(Q)$ and $Q\{R \in \mathcal R', W \in \mathcal W'\} > 0$. Recall $Q \in \mathbf Q_0^{\rm w}(P, \mathbf Q^{\rm w})$ if and only if $Q \in \mathbf Q^{\rm w}$ and $P = Q(T^{\rm w})^{-1}$. Since $Q \in \mathbf Q^{\rm w}$, it satisfies both Assumptions \ref{as:exog_w} and \ref{ass:generalizedstrata}. Because $P = Q (T^{\rm w})^{-1}$, we can conclude that
\begin{equation} \label{eq:zw}
Q \{Z = z, W = w\} = P \{Z = z, W = w\}
\end{equation}
for all $(z, w) \in \mathcal Z \times \mathcal W$. 
Moreover, since $P\{Z = z, W = w\} > 0$ for all $(z, w) \in \mathcal Z \times \mathcal W$ for any $P\in \mathbf P_{0}^{\rm w}$, $Q$ satisfying Assumption \ref{as:exog_w} and $P = Q (T^{\rm w})^{-1}$ imply that
\begin{align}
\nonumber P\{Y = y, D = d | Z = z, W = w\} & = Q\{Y(d) = y, D(z) = d | Z = z, W = w\} \\
\label{eq:p-q-w} & = Q\{Y(d) = y, D(z) = d | W = w\}~.
\end{align}
Hence, results \eqref{eq:zw} and \eqref{eq:p-q-w} together with $Q\{W=w\} = P\{W=w\}$ imply that
\begin{equation} \label{eq:p-q*w}
P\{Y = y, D = d | Z = z, W = w\} \times P\{W = w\} = Q\{Y(d) = y, D(z) = d, W = w\}~.
\end{equation}
Next, let $q$ be a column vector with $|\mathcal R| \times |\mathcal W|$ elements indexed by $(r, w) \in \mathcal R \times \mathcal W$, where the $(r, w)$ coordinate of $q$ equals $q(r, w) = Q\{R = r, W = w\}$. 
Then note  result \eqref{eq:p-q*w} and the definitions of $A^{\rm w}$ and $\beta^{\rm w}(P)$ imply that the first $|\mathcal M^{\rm w}|$ rows of $A^{\rm w} q$ equal the first $|\mathcal M^{\rm w}|$ rows of $\beta^{\rm w}(P)$. 
Moreover, result \eqref{eq:zw} and $Q\{R \in \mathcal R\} = 1$ by definition of $\mathbf Q_0^{\rm w}(P,\mathbf Q^{\rm w})$  imply
\begin{equation} \label{eq:sumq-w}
\sum_{r \in \mathcal R} q(r, w) = \sum_{r \in \mathcal R} Q\{R = r, W = w\} = Q\{W = w\} =  P\{W = w\}~,  
\end{equation}
and therefore it follows that the $|\mathcal M^{\rm w}| + 1$ through $|\mathcal M^{\rm w}| + |\mathcal W|$ rows of $A^{\rm w} q$ equal the $|\mathcal M^{\rm w}| + 1$ through $|\mathcal M^{\rm w}| + |\mathcal W|$ rows of $\beta^{\rm w}(P)$. 
Finally, because $\theta_0 = \theta^{\rm w}(Q)$ by hypothesis, the definition of the conditional expectation yields that
\begin{equation} \label{eq:ratio-w}
\theta_0 = \frac{E_Q[g(R, W) I\{R \in \mathcal R', W \in \mathcal W'\}]}{Q\{R \in \mathcal R', W \in \mathcal W'\}}~,
\end{equation}
where we recall $Q\{R \in \mathcal R', W \in \mathcal W'\}$ satisfies $Q\{R \in \mathcal R', W \in \mathcal W'\} > 0$. Rearranging the expression in \eqref{eq:ratio-w} we obtain that $E_Q[(g(R, W) - \theta_0) I\{R \in \mathcal R', W \in \mathcal W'\}] = 0$ or, equivalently,
\[ \sum_{r \in \mathcal R', w \in \mathcal W'} (g(r, w) - \theta_0) q(r, w) = 0~. \]
Thus, it follows from the definition of $A^{\rm w}$ that the $|\mathcal M^{\rm w}| + |\mathcal W| + 1$ row of $A^{\rm w} q$ equals the $|\mathcal M^{\rm w}| + |\mathcal W| + 1$ row of $\beta^{\rm w}(P)$. Hence, we can conclude that $\mathbf P_0^{\rm w} \subseteq \tilde{\mathbf P}_0^{\rm w} \equiv \{P \in \mathbf P^{\rm w}: A^{\rm w} x = \beta^{\rm w}(P) \text{ for some } x \geq 0\}$, which establishes the first claim of the theorem.

Next, we establish the second claim by showing $\tilde{\mathbf P}_0^{\rm w} \subseteq \mathrm{cl}_{\mathbf P^{\rm w}}(\mathbf P_0^{\rm w})$ whenever $\mathbf P_0^{\rm w} \neq \emptyset$. To this end, fix any $\tilde P \in \tilde{\mathbf P}_0^{\rm w}$ and note $A^{\rm w} q = \beta^{\rm w}(\tilde P)$ for some $q \geq 0$ due to $\tilde P \in \tilde{\mathbf P}_0^{\rm w}$.
Since $\tilde P \in \mathbf P^{\rm w}$ implies that $\tilde P\{Z = z, W = w\} > 0$ for all $(z, w) \in \mathcal Z \times \mathcal W$, it follows that we can define a finite measure $\tilde Q$ on the sample space $\mathcal R \times \mathcal Z \times \mathcal W$ according to
\begin{equation} \label{eq:tildeQ-w}
\tilde Q\{R = r, Z = z, W = w\} = q(r, w) \tilde P\{Z = z | W = w\}    ~.
\end{equation}
Next, note that the $|\mathcal M^{\rm w}|+1$ through $|\mathcal M^{\rm w}|+|\mathcal W|$ rows of $A^{\rm w} q$ being equal to the corresponding rows of $\beta^{\rm w}(\tilde P)$ further allow us to conclude that
\begin{equation*}
\sum_{z\in \mathcal Z} \sum_{w\in \mathcal W} \sum_{r\in \mathcal R} \tilde Q\{R=r,Z=z,W = w\} = \sum_{z\in \mathcal Z} \sum_{w\in \mathcal W} \tilde P\{W = w\}\tilde P\{Z=z|W=w\} = 1,
\end{equation*}
which implies that $\tilde Q$ is a probability measure.
Moreover, by construction, $R\indep Z|W$ under $\tilde Q$ and $\tilde Q\{R\in \mathcal R\} = 1$, which implies that $\tilde Q \in \mathbf Q^{\rm w}$.
Similar arguments yield that $\tilde Q\{Z = z,W = w\} = \tilde P\{Z = z, W=w\}$ for any $(z,w)\in \mathcal Z \times \mathcal W$, while the first $|\mathcal M^{\rm w}|$ rows of $A^{\rm w}q$ being equal to the first $|\mathcal M^{\rm w}|$ rows of $\beta^{\rm w}(\tilde P)$ imply that
\begin{align}\label{eq:wadd1}
\tilde P\{Y=y,D=d|Z=z,W=w\} & =  \sum_{r\in \mathcal R} I\{y(d)=y, d(z) =d \} \frac{q(r,w)}{\tilde P\{W =w \}} \notag \\
&  =  \sum_{r\in \mathcal R} I\{y(d)=y, d(z) =d \} \frac{\tilde Q\{R =r, Z =z , W =w \}}{\tilde P\{Z= z, W =w \}} \notag \\
&  = \tilde Q\{Y(d) = y, D(z) = d|Z=z, W = w\} ~,
\end{align}
where in the second equality we used the definition of $\tilde Q$, and in the third that $\tilde P\{W=w,Z=z\} = \tilde Q\{W =w,Z = z\}$.
Since \eqref{eq:wadd1} holds for any $(y,d,z,w) \in \mathcal M^{\rm w}$ and $\tilde P\{(Y,D,Z,W)\in \mathcal M^{\rm w}\} = 1$ due to $\tilde P\in \mathbf P^{\rm w}$, it follows that $\tilde P = \tilde Q(T^{\rm w})^{-1}$.
Furthermore, we also have
\begin{multline}\label{eq:wadd2}
    E_{\tilde Q}[(g(R,W) -\theta_0)I\{R\in \mathcal R^\prime, W\in \mathcal W^\prime\}] \\= \sum_{r\in \mathcal R}\sum_{w\in \mathcal W}(g(r,w)-\theta_0)I\{r\in \mathcal R^\prime, w\in \mathcal W^\prime\} q(r,w) = 0,
\end{multline}
where the first equality follows by definition of $\tilde Q$, and the second from the final row of $A^{\rm w} q$ being equal to zero.

On the other hand, because $\mathbf P_0^{\rm w} \neq \emptyset$, there exists $Q^\star \in \mathbf Q^{\rm w}$, hence satisfying Assumptions \ref{as:exog_w} and \ref{ass:generalizedstrata}, such that $Q^\star \{R \in \mathcal R', W \in \mathcal W'\} > 0$ and $\theta^{\rm w}(Q^\star) = \theta_0$. Moreover, we have
\[ Q^\star\{R = r, Z = z, W = w\} = Q^\star\{R = r, W = w\} Q^\star\{Z = z | W = w\} \]
because $Q^\star$ satisfies Assumption \ref{as:exog_w} due to $Q^\star \in \mathbf Q^{\rm w}$.
Finally, we define a measure $Q$ by
\[ Q\{R = r, Z = z, W = w\} = Q^\star\{R = r, W = w\} \tilde P\{Z = z | W = w\}~, \]
and note that $Q \in \mathbf Q^{\rm w}$ because $Q^\star \in \mathbf Q^{\rm w}$. 
Furthermore, $Q\{R = r, W = w\} = Q^\star\{R = r, W = w\}$ for all $r \in \mathcal R$ and $w \in \mathcal W$, which implies $Q \{R \in \mathcal R', W \in \mathcal W'\} > 0$ and $\theta^{\rm w}(Q) = \theta_0$. 
To conclude, we let $P = Q (T^{\rm w})^{-1}$, $P_n = n^{-1} P + (1-n^{-1})\tilde P$, and $Q_n = n^{-1} Q + (1-n^{-1})\tilde Q$.
Since $P\{Z=z|W=w\} = \tilde P\{Z = z|W=w\}$ for any $(z,w)\in \mathcal Z \times \mathcal W$, it is straightforward to show that  $P_n = Q_n(T^{\rm w})^{-1}$ and that $Q_n$ satisfies Assumptions \ref{ass:generalizedstrata} and \ref{as:exog_w}.
Moreover, $Q\{R\in \mathcal R^\prime, W\in \mathcal W^\prime\} > 0$, $\theta^{\rm w}(Q) = \theta_0$, and result \eqref{eq:wadd2} yield that $\theta^{\rm w}(Q_n) = \theta_0$.
We thus obtain that $\theta_0 \in \Theta_0^{\rm w}(P_n,\mathbf Q^{\rm w})$ for all $n$.
Since $P_n \in \mathbf P^{\rm w}$, it follows that $P_n \in \mathbf P^{\rm w}_0$ and, because $P_n\to \tilde P$ as $n\to \infty$, that $\tilde P  \in \mathrm{cl}_{\mathbf P^{\rm w}}(\mathbf P_0^{\rm w})$ as desired.
\end{proof}

\subsection{Proof of Theorem \ref{thm:equiv_cont}} \label{app:equiv_cont}

In this appendix, we first establish Theorem \ref{thm:equiv_cont} and subsequently show that the result continues to hold if we remove linearly dependent restrictions from $A_L$ and $\beta_L(P)$.

Before proceeding with the proof, we first introduce some additional notation. 
Given the sets $\mathcal R \subseteq \mathcal Y^{|\mathcal D|} \times \mathcal D^{|\mathcal Z|}$ and $\mathcal R^\prime \subseteq \mathcal Y^{|\mathcal D|} \times \mathcal D^{|\mathcal Z|}$ we define, for any $r_t\in \mathcal D^{ | \mathcal Z|}$, the sets
\begin{align*}
\mathcal R_o(r_t) & \equiv \{r_o \in \mathcal Y^{|\mathcal D|} : (r_o,r_t)\in \mathcal R\}\\
\mathcal R_o^\prime(r_t) & \equiv \{r_o \in \mathcal Y^{|\mathcal D|} : (r_o,r_t)\in \mathcal R^\prime\}~.
\end{align*}
Recall $B_{c_o} \equiv \prod_{d \in \mathcal D} B_{c_o(d), L} \subseteq \mathcal Y^{|\mathcal D|}$ for any $c_o \equiv (c_o(d): d \in \mathcal D) \in \mathcal L^{|\mathcal D|}$ and $\mathcal M_L$ is defined as the set of realizations for $(\mathbb L, D,Z)$ that is compatible with the response types $R$ belonging to $\mathcal R$ (see \eqref{eq:MLdef}).
To define the matrix $A_L$ we introduce a binary variable $\chi \in \{0,1\}$ and for any $c_o = (c_o(d) : d\in \mathcal D)$ and $r_t = (d(z) : z\in \mathcal Z)$ we set
\[ \mathcal Y(c_o, r_t, \chi) = \begin{cases}
    B_{c_o} \cap \mathcal R_o'(r_t) & \text{ if } \chi = 1 \\
    \mathcal R_o(r_t) \cap (B_{c_o} \backslash \mathcal R_o'(r_t)) & \text{ if } \chi = 0~.
\end{cases} \]
Note that because $\mathcal R_o^\prime (r_t) \subseteq \mathcal R_o(r_t)$, the set $\mathcal Y(c_o,r_t,0)\cup \mathcal Y(c_o,r_t,1)$ is non-empty if and only if $B_{c_o} \cap \mathcal R_o(r_t) \neq \emptyset$. Using this notation we define the set of indices
\begin{equation}\label{eq:defI}
\mathcal I \equiv \{(c_o,r_t,\chi) \in \mathcal L^{|\mathcal D|} \times \mathcal D^{|\mathcal Z|} \times \{0,1\} : \mathcal Y(c_o,r_t,\chi) \neq \emptyset\}    .
\end{equation}

We can finally define $A_L$ as a matrix with $|\mathcal M_L|+3$ rows and $|\mathcal I| +2$ columns.
We index the first $|\mathcal M_L|$ rows of $A_L$ by $(\ell,d,z)\in \mathcal M_L$, the first $|\mathcal I|$ columns by $(c_o,r_t,\chi) \in \mathcal I$, and set
\begin{align*}
A_{L,(\ell,d,z), (c_o,r_t,\chi)} & = I\{c_o(d) = \ell, ~ d(z) = d\} \\
A_{L,|\mathcal M_L|+1, (c_o,r_t,\chi)} & = 1 ~.
\end{align*}
We also set the $(|\mathcal M_L| + 2) \times (c_o,r_t,\chi)$ and $(|\mathcal M_L|+3)\times (c_o,r_t,\chi)$ elements of $A_L$ to be
\begin{align*}
A_{L,|\mathcal M_L|+2, (c_o,r_t,\chi)} & = \left( \sup_{r_o \in \mathcal Y(c_o, r_t, \chi)} g(r_o, r_t) - \theta_0 \right)  \cdot \chi  \\
A_{L,|\mathcal M_L|+3, (c_o,r_t,\chi)} & = \left( \inf_{r_o \in \mathcal Y(c_o, r_t, \chi)} g(r_o, r_t) - \theta_0 \right)  \cdot \chi ~.
\end{align*}
Finally, the $(|\mathcal M_L| + 2) \times (|\mathcal I|+1)$ element of $A_L$ is given by $-1$ and the $(|{\mathcal M}_L| + 3) \times (|\mathcal I|+2)$ element of $A_L$ is given by $1$. 
All other elements of $A_L$ are set to equal zero.

Given the introduced definitions, we next prove Theorem \ref{thm:equiv_cont}.

\begin{proof}[\sc Proof of Theorem \ref{thm:equiv_cont}]
First note that $\mathcal R_o^\prime(r_t)\subseteq \mathcal R_o(r_t)$ implies that $B_{c_o}\cap \mathcal R_o(r_t) \neq \emptyset$ if and only if $\mathcal Y(c_o,r_t,0)\cup \mathcal Y(c_o,r_t,1) \neq \emptyset$. Therefore, the definition of $\mathcal I$, $\mathcal R \subseteq \mathcal Y^{|\mathcal D|} \times \mathcal D^{|\mathcal Z|}$, and $\{B_{c_o}\}_{c_o \in \mathcal L^{|\mathcal D|}}$ being a partition of $\mathcal Y^{|\mathcal D|}$ by assumption allows us to conclude that
\begin{equation}\label{thm:equiv_cont0}
\bigcup_{(c_o,r_t,\chi) \in \mathcal I} (\mathcal Y(c_o,r_t,\chi) \times \{r_t\}) = \mathcal R~.  
\end{equation}
Similarly, we also note that by the same arguments and the definition of $\mathcal Y(c_o,r_t,\chi)$ we have
\begin{equation}\label{thm:equiv_cont0p5}
\bigcup_{(c_o,r_t) : (c_o,r_t,1) \in \mathcal I} (\mathcal Y(c_o,r_t,1) \times \{r_t\}) = \mathcal R^\prime~.  
\end{equation}

To show $\mathbf P_0 \subseteq \mathbf P_{0,L}$ we first suppose $P \in \mathbf P_0$. 
Then, there exists a $Q \in \mathbf Q_0(P, \mathbf Q)$ such that $\theta_0 = \theta(Q)$ and $Q\{R\in \mathcal R'\} > 0$. 
Letting $ q(c_o, r_t, \chi) \equiv Q\{R_o \in \mathcal Y(c_o, r_t, \chi), R_t = r_t\}$ for any $(c_o,r_t,\chi)\in \mathcal I$ we then obtain from $Q \in \mathbf Q_0(P,\mathbf Q)$, result \eqref{thm:equiv_cont0}, and $Q\{R\in \mathcal R\} =1$ that
\begin{align} \label{eq:p-q_cont}
P\{Y\in B_{\ell,L}, D= d|Z=z\} & = Q \{Y(d) \in B_{\ell,L}, D(z) = d, ~ R \in \mathcal R \} \notag \\
& = \sum_{(c_o,r_t,\chi)\in \mathcal I}  q(c_o, r_t, \chi) I\{c_o(d) = \ell, d(z) = d\} ~.
\end{align}
Collect the probabilities $q(c_o, r_t, \chi)$ into a vector $q = (q(c_o, r_t, \chi) : (c_o,r_t,\chi)\in \mathcal I)$ and define
\begin{align*}
\xi_1 & \equiv \sum_{(c_o,r_t): (c_o,r_t,1)\in \mathcal I} \left( \sup_{r_o \in \mathcal Y(c_o,r_t,1)} g(r_o, r_t) - \theta_0 \right) q(c_o, r_t, 1) \\
\xi_2 &\equiv - \sum_{(c_o,r_t): (c_o,r_t,1)\in \mathcal I} \left( \inf_{r_o \in \mathcal Y(c_o,r_t,1)} g(r_o, r_t) - \theta_0 \right)  q(c_o, r_t, 1)~.
\end{align*}
Next, note that since $\theta(Q) = \theta_0$ by hypothesis and $Q \{R \in \mathcal R'\} > 0$, it follows that $Q$ satisfies $E_Q[(g(R)-\theta_0)I\{R\in \mathcal R^\prime\}] = 0$.
In particular, by result \eqref{thm:equiv_cont0p5} we obtain that
\begin{equation*}
-\xi_2 \leq E_Q[(g(R)-\theta_0)I\{R\in \mathcal R^\prime\}] = 0 \leq \xi_1~.
\end{equation*}
Setting $x = (q^\prime, \xi_1,\xi_2)^\prime$ we then note that $x\geq 0$.
Moreover, it follows from \eqref{eq:p-q_cont} and the definition of $A_L$ that the first $|{\mathcal M}_L|$ rows of $A_L x$ equal the first $|{\mathcal M}_L|$ rows of $\beta_L(P)$. 
Also note that result \eqref{thm:equiv_cont0}, the definition of $q$, and $Q\{R \in \mathcal R\} =1$ allow us to conclude that
\begin{equation*}
\sum_{(c_o,r_t,\chi)\in \mathcal I} q(c_o, r_t, \chi)  = \sum_{(c_o,r_t,\chi)\in \mathcal I} Q\{R_o \in \mathcal Y(c_o,r_t,\chi), R_t = r_t\}  = Q\{R \in \mathcal R\} = 1~.
\end{equation*}
In particular, it follows that the $|{\mathcal M}_L| + 1$ row of $A_L x$ equals the $|{\mathcal M}_L|+1$ row of $\beta_L(P)$ by definition of $A_L$ and $\beta_L(P)$. 
Finally, we note the $|{\mathcal M}_L|+2$ and $|{\mathcal M}_L|+3$ rows of $A_L x$ equal the $|\mathcal M_L|+2$ and $|\mathcal M_L|+3$ rows of $\beta_L(P)$ by definition of $\xi_1$, $\xi_2$, $A_L$ and $\beta_L(P)$.
Therefore, we have shown that $\mathbf P_0 \subseteq {\mathbf P}_{0,L} \equiv \{P \in \mathbf P : A_Lx = \beta_L(P) \text{ for some } x \geq 0 \}$.

To establish the second claim, we proceed by contradiction and assume there is a $P\in \limsup_{L \to \infty} \mathbf P_{0,L}$ with $P\notin \mathrm{cl}_{\mathbf P}(\mathbf P_0)$.
Since $P\notin \mathrm{cl}_{\mathbf P}(\mathbf P_0)$, Lemma \ref{lem:sep} implies that
\begin{equation}\label{th:cont1}
\int f dP = 1~, \hspace{0.5 in}\int f dP^\prime \leq 0 \text{ for all } P^\prime \in \mathbf P_0 \text{ satisfying } P^\prime_Z = P_Z  
\end{equation}
for some continuous and bounded function $f: \mathcal F \to \mathbf R$ for $\mathcal F \equiv \mathcal Y \times \mathcal D \times \mathcal Z$ and $P_Z$ denoting the marginal distribution of $Z$ under $P$.
Since the function $f:\mathcal F \to \mathbf R $ is continuous and $\mathcal Y$ is compact by Assumption \ref{ass:cont}(i), it follows that the function $f$ is uniformly continuous on $\mathcal F$.
For any $\eta > 0$ there therefore exists a $\delta > 0$ such that $f$ satisfies
\begin{equation}\label{th:cont2}
\max_{d\in \mathcal D,z\in \mathcal Z} \sup_{|y-y^\prime| \leq 2\delta} |f(y,d,z)-f(y^\prime,d,z)| \leq \eta~.    
\end{equation}
Moreover, by Assumption \ref{ass:cont}(iii) there is a $L_0 < \infty$ such that $\max_{1\leq \ell \leq L} \text{diam}\{B_{\ell,L}\} < \delta$ for all $L\geq L_0$.
Since in addition $P\in \limsup_{L \to \infty} \mathbf P_{0,L}$ implies that $P\in \mathbf P_{0,L}$ for some $L\geq L_0$, we can conclude that there is some $L < \infty$ such that
\begin{equation}\label{th:cont3}
P\in \mathbf P_{0,L} ~ ~ ~ \text{ and } ~ ~  ~ \max_{1\leq \ell \leq L} \text{diam}\{B_{\ell,L}\} < \delta~.    
\end{equation}
Next define $\lambda$ to be a vector of dimension $|\mathcal M_L| +3$ whose last three entries are zero, and each of its first $|\mathcal M_L|$ entries, which we index by $(\ell,d,z)\in \mathcal M_L$, is given by
\begin{equation*}
\lambda(\ell, d, z) \equiv \sup_{y \in B_{\ell, L}} f(y, d, z) P \{Z = z\}~.	
\end{equation*}
Also observe that $P\in \mathbf P_{0,L}$ implies $\beta_L(P) = A_Lx$ for some $x\geq 0$ belonging to $\mathbf R^{|\mathcal I|+2}$.
Index the first $|\mathcal I|$ entries of $x$ by $(c_o,r_t, \chi) \in \mathcal I$.
For each $(c_o,r_t,\chi)\in \mathcal I$ we select an arbitrary $r_o^*(c_o,r_t,\chi) \in \mathcal Y(c_o,r_t,\chi)$ and define a finite measure $Q^\prime_R$ with support $\{(r_o^*(c_o,r_t,\chi), r_t)\}_{(c_o,r_t,\chi)\in \mathcal I}$ by setting $Q^\prime_R\{(R_o,R_t) = (r_o^*(c_o,r_t,\chi),r_t)\} = x(c_o,r_t,\chi)$.
Also note that by definition of $Q_R^\prime$ and the $|\mathcal M_L|+1$ row of $A_L x$ being equal to one we have that
\begin{equation*}
\sum_{(c_o,r_t,\chi)\in \mathcal I} Q_R^\prime \{(R_o,R_t) = (r_o^*(c_o,r_t,\chi),r_t)\} = \sum_{(c_o,r_t,\chi)\in \mathcal I} x(c_o,r_t,\chi)  = 1 ,
\end{equation*}
which implies $Q_{R}^\prime$ is a probability measure and, since $r_o^*(c_o,r_t,\chi)\in \mathcal Y(c_o,r_t,\chi)$, that
\begin{equation} \label{th:cont4}
\sum_{(c_o,r_t,\chi)\in \mathcal I} Q_R'\{(R_o,R_t) \in (\mathcal Y(c_o, r_t, \chi)\times \{r_t\})\} = 1~.
\end{equation}
Moreover, since the last two  entries of $A_Lx$ equal $(0,0)$ because $A_Lx=\beta_L(P)$, and in addition $Q^\prime_R\{(R_o,R_t) \in (\mathcal Y(c_o,r_t,1)\times \{r_t\}) \} = x(c_o,r_t,1)$ by definition of $Q_R^\prime$, we obtain
\begin{align}\label{lm:cont5}
\sum_{(c_o,r_t) : (c_o,r_t,1)\in \mathcal I} \left(\sup_{r_o\in \mathcal Y(c_o,r_t,1)} g(r_o,r_t) - \theta_0\right)\times Q_R^\prime\{(R_o,R_t) \in (\mathcal Y(c_o,r_t,1)\times \{r_t\})\} & \geq 0 \notag \\
\sum_{(c_o,r_t) : (c_o,r_t,1)\in \mathcal I} \left(\inf_{r_o\in \mathcal Y(c_o,r_t,1)} g(r_o,r_t) - \theta_0\right)\times Q_R^\prime\{(R_o,R_t) \in (\mathcal Y(c_o,r_t,1)\times \{r_t\}) \} & \leq 0 ~.
\end{align}
Setting $Q^\prime \equiv Q_R^\prime \times P_Z$ and letting $P^\prime = Q^\prime T^{-1}$, then note $A_Lx = \beta_L(P)$, the definition of $Q_R^\prime$, and the manipulations in \eqref{eq:p-q_cont} (applied with $P^\prime$ and $Q^\prime$ in place of $P$ and $Q$) imply  $\beta_L(P) = \beta_L(P^\prime)$.
Therefore, using that $f$ satisfies condition \eqref{th:cont2}, $L$ was chosen to satisfy $\max_{1\leq \ell\leq L} \text{diam}\{B_{\ell,L}\} < \delta$ in \eqref{th:cont3}, the definition of $\lambda$, and $P\{(\mathbb L,D,Z)\in \mathcal M_L\} = 1$ gives
\begin{multline}\label{th:cont6}
\int f dP \leq \sum_{(\ell,d,z)\in \mathcal M_L} P\{Y\in B_{\ell,L},D=d,Z=z\}\times \sup_{y\in B_{\ell,L}}f(y,d,z) \\ 
= \lambda^\prime \beta_L(P) = \lambda^\prime \beta_L(P^\prime)  \leq \int f dP^\prime + \eta ~.
\end{multline}
On the other hand, $Q^\prime_R$ satisfying the restrictions in \eqref{th:cont4} and \eqref{lm:cont5} implies by Lemma \ref{lm:fix} that there is a distribution $\tilde Q_R$ for $R$ such that $\tilde Q \equiv \tilde Q_R\times P_Z$ satisfies $E_{\tilde Q}[(g(R)-\theta_0)I\{R\in \mathcal R^\prime\}] = 0$, $\tilde Q\{R\in \mathcal R\} = 1$, and for which the induced distribution $\tilde P \equiv \tilde QT^{-1}$ satisfies
\begin{equation}\label{th:cont7}
\int f(dP^\prime - d\tilde P) \leq 2 \eta~.    
\end{equation}
Finally, let $\mathbf P_0(P) \equiv \{\hat P \in \mathbf P_0 : \hat P_Z = P_Z\}$ and note that Lemma \ref{lem:P0P} implies $\mathbf P_0(P)$ is not empty. 
Pick any $P_0\in \mathbf P_0(P)$ and note that $P_0 = Q_0 T^{-1}$ for some $Q_0\in \mathbf Q$ satisfying $Q_0 = Q_{0,R}\times P_Z$ for some distribution $Q_{0,R}$ for $R$.
For any $\tau \in (0, 1]$, it then follows that $Q_\tau \equiv (1-\tau)\tilde Q + \tau Q_0 = \{(1-\tau)\tilde Q_R + \tau Q_{0,R}\} \times P_Z$ satisfies $E_{Q_\tau}[(g(R)-\theta_0)I\{R\in \mathcal R^\prime\}] = 0$, $Q_\tau\{R\in \mathcal R\} =1$, and $Q_\tau\{R\in \mathcal R^\prime\} > 0$ because $Q_{0}\{R\in \mathcal R^\prime\} > 0$.
In particular, it follows that $P_\tau \equiv Q_\tau T^{-1} \in \mathbf P_0(P)$.
Moreover, since $f$ is bounded, we also have that $|\int f(d\tilde P- dP_\tau)|\leq 2\tau \|f\|_\infty$.
Therefore, setting $\tau \leq \eta/(2\|f\|_\infty)$ allows us to conclude that
\begin{equation}\label{th:cont8}
\int f d\tilde P \leq \int fdP_\tau + \eta \leq \eta~,
\end{equation}
where the final inequality holds by \eqref{th:cont1} and $P_\tau \in \mathbf P_0(P)$. Combining results \eqref{th:cont6}, \eqref{th:cont7}, and \eqref{th:cont8} finally yields $\int f dP \leq 4\eta$, which contradicts \eqref{th:cont1} for $\eta$ sufficiently small. 
Hence, we conclude that there is no $P\in \limsup_{L \to \infty} \mathbf P_{0,L}$ satisfying $P\notin \text{cl}_{\mathbf P}(\mathbf P_0)$.
\end{proof}

We next show that the conclusion of Theorem \ref{thm:equiv_cont} also holds if we remove linearly redundant restrictions from $A_L$.
To this end, recall that the first $|\mathcal M_L|$ rows of $A_L$ and $\beta_L(P)$ are indexed by $(\ell,d,z)\in \mathcal M_L$.
For each $z\in \mathcal Z$ we select an arbitrary $(\ell_z,d_z)\in \mathcal M_{L,z}$ (as defined in \eqref{eq:MzLdef}), let $A_L^{\rm s}$ and $\beta^{\rm s}_L(P)$ denote the matrix and vector obtained by removing the rows $((\ell_z,d_z,z) : z\in \mathcal Z)$ from $A_L$ and $\beta_L(P)$, and define the set of probabilities
\begin{equation*}
\mathbf P_{0,L}^{\rm s}  \equiv \{P\in \mathbf P :A_L^{\rm s}x   = \beta_L^{\rm s}(P) \text{ \rm for some } x \geq 0\}~.
\end{equation*}
Our next lemma establishes that $ {\mathbf P}_{0,L} = {\mathbf P}_{0,L}^{\rm s}$, which implies that the conclusion of Theorem \ref{thm:equiv_cont} continues to hold if we replace $A_L$ and $\beta_L(P)$ with $A_L^{\rm s}$ and $\beta_L^{\rm s}$ respectively.

\begin{lemma}\label{lm:dropz:cont}
Let $\mathbf P$ denote the set of all distributions of $(Y,D,Z)$ satisfying $P\{(\mathbb L,D,Z)\in \mathcal M_L\} = 1$ for all $L$ and $P\{Z=z\} > 0$ for all $z\in \mathcal Z$. Then, it follows that ${\mathbf P}_{0,L}  = {\mathbf P}_{0,L}^{\rm s}$.
\end{lemma}

\begin{proof}[\sc Proof of Lemma \ref{lm:dropz:cont}]
The proof follows similar arguments to those used in Lemma \ref{lm:dropz}.
First note that any $x$ satisfying $A_L x = \beta_L(P)$ also satisfies $A^{\rm s}_L x = \beta^{\rm s}_L(P)$, and therefore ${\mathbf P}_{0,L} \subseteq {\mathbf P}_{0,L}^{\rm s}$.
For the reverse inclusion select an arbitrary $(c_o,r_t,\chi) \in \mathcal I$ and $z\in \mathcal Z$.
Since $\mathcal Y(c_o,r_t,\chi) \neq \emptyset$ for any $(c_o,r_t,\chi) \in \mathcal I$ and $\mathcal R^\prime \subseteq \mathcal R$, it follows from the definition of $\mathcal Y(c_o,r_t,\chi)$ that there is a $r = (r_o,r_t)\in \mathcal R$ such that $r_o = (y(d) : d\in \mathcal D)$ and $r_t = (d(z) : z\in \mathcal Z)$ satisfy $r_o \in B_{c_o}$.
By definition of $\mathcal M_{L,z}$ it follows that there exists an $(\ell, d)\in \mathcal M_{L,z}$ such that
\begin{equation}\label{lm:dropz:cont0}
1 = I\{y(d) \in B_{\ell,L}, ~ d(z) = d\} = I\{c_o(d) = \ell, ~d(z) = d\}~,
\end{equation}
where the second equality follows from $r_o \in B_{c_o}$.
Noting that, for any $(c_o,r_t)$, the indicator $I\{c_o(d) = \ell, ~ d(z) = d\}$ can only equal one for a single $(\ell,d)$, we obtain from \eqref{lm:dropz:cont0} that 
\begin{equation}\label{lm:dropz:cont1}
\sum_{(\ell,d) \in \mathcal M_{L,z}} A_{L,(\ell,d,z),(c_o,r_t,\chi)} = 1   
\end{equation}
for any $(c_o,r_t,\chi) \in \mathcal I$ and $z\in \mathcal Z$.
Hence, letting $A_{L,(\ell,d,z),\cdot}$ denote the $(\ell,d,z)$ row of the matrix $A_L$ and using the definition of $A_{L,|\mathcal M_L|+1,\cdot}$ and that the final two columns of $A_{L,(\ell,d,z),\cdot}$ and $A_{L,|\mathcal M_L|+1,\cdot}$ equal zero, we obtain from result \eqref{lm:dropz:cont1} that for any $z\in \mathcal Z$ 
\begin{equation}\label{lm:dropz:cont2}
\sum_{(\ell,d)\in \mathcal M_{L,z}} A_{L,(\ell,d,z),\cdot} = A_{L,|\mathcal M_L|+1,\cdot}   ~.  
\end{equation}
Next select an arbitrary $P\in {\mathbf P}_{0,L}^{\rm s}$ and note that by definition there exists an $x_0 \geq 0$ satisfying $A^{\rm s}_L x_0 = \beta^{\rm s}_L(P)$.
Since the system $A^{\rm s}_L x_0 = \beta^{\rm s}_L(P)$ still contains the adding up constraint $A_{L,|\mathcal M_L|+1,\cdot} x_0 = 1$, result \eqref{lm:dropz:cont2} allows us to conclude for any $z\in \mathcal Z$ that
\begin{equation}\label{lm:dropz:cont3}
1 = \sum_{(\ell,d)\in \mathcal M_{L,z}} A_{L,(\ell,d,z),\cdot}x_0 = \sum_{(\ell,d)\in \mathcal M_{L,z} \setminus \{(\ell_z,d_z)\}} \beta_{L,(\ell,d,z)}(P) + A_{L,(\ell_z,d_z,z),\cdot} x_0,
\end{equation}
where in the second equality we used that $A^{\rm s}_Lx_0 = \beta^{\rm s}_L(P)$ implies $A_{L,(\ell,d,z),\cdot}x_0 = \beta_{L,(\ell,d,z)}(P)$ for any $(\ell,d)\in \mathcal M_{L,z}\setminus \{(\ell_z,d_z)\}$ because those constraints remain in $A^{\rm s}_L$.
Moreover, since $P\{Z=z\} > 0$ for all $z\in \mathcal Z$ and $P\{(\mathbb L,D,Z)\in \mathcal M_L\} = 1$ due to $P\in \mathbf P$, we have
\begin{equation}\label{lm:dropz:cont4}
\sum_{(\ell,d)\in \mathcal M_{L,z}} \beta_{L,(\ell,d,z)}(P) = \sum_{(\ell,d)\in \mathcal M_{L,z}} P\{Y \in B_{\ell,L}, D = d|Z =z \} = 1.
\end{equation}
Combining \eqref{lm:dropz:cont3} and \eqref{lm:dropz:cont4} we obtain that $A_{L,(\ell_z,d_z,z),\cdot} x_0 = \beta_{L,(\ell_z,d_z,z)}(P)$ for any $z\in \mathcal Z$.
It follows that $A_Lx_0 = \beta_L(P)$ and, since $x_0 \geq 0$, that $P\in {\mathbf P}_{0,L}$ as desired.
\end{proof}

\subsection{Proof of Corollary \ref{cor:validity_cont}}
\begin{proof}[\sc Proof of Corollary \ref{cor:validity_cont}]
Setting $g = 1$, $\theta_0 = 1$, and $\mathcal R' = \mathcal R$, we have
\[
\Theta_0(P,\mathbf Q) =
\begin{cases}
\{1\} & \text{ if } \mathbf Q_0(P,\mathbf Q) \neq \emptyset \\
\emptyset & \text{ otherwise }
\end{cases}
~.
\]
Therefore, $\mathbf P_0 = \{P \in \mathbf P : 1 \in \Theta_0(P,\mathbf Q)\} = \{P \in \mathbf P : \mathbf Q_0(P,\mathbf Q) \neq \emptyset\}$. For this choice of $g$, $\theta_0$, and $\mathcal R'$, because the last two rows of the equation $A_L x = \beta_L(P)$ are always satisfied, we have $\tilde{\mathbf P}_{0, L} = \mathbf P_{0, L}$. Theorem \ref{thm:equiv_cont} then implies $\{P \in \mathbf P : \mathbf Q_0(P,\mathbf Q) \neq \emptyset\} \subseteq \tilde{\mathbf P}_{0, L}$ as well as the last statement of the corollary.    
\end{proof}

\section{Additional Examples} \label{app:examples}

This section includes additional examples that were briefly mentioned in Section \ref{sec:examples}.

\begin{example}\label{eg:chengsmall}
\cite{cheng2006bounds} study the special case of Example \ref{eg:onesided} in which $\mathcal{D} = \mathcal{Z} = \{0,1,2\}$. Their ``Monotonicity I'' assumption corresponds to Assumption \ref{ass:generalizedstrata} with $\mathcal R$ defined as in Example \ref{eg:onesided}.  They further consider imposing the restriction that $Q\{ D(1)=1 \mid D(2)=2\}=1$, i.e., that subjects who would comply with assignment to treatment $2$ would also comply with assignment to treatment $1$.  They argue that this assumption is plausible in contexts where the ``cost'' of compliance with treatment $1$ is lower than that with treatment $2$.  The combination of these two restrictions can be formulated in terms of Assumption \ref{ass:generalizedstrata} with $\mathcal R = \{(r_o,r_t): (d(0), d(1), d(2)) \in \{(0, 0, 0), (0, 1, 0), (0, 1, 2)\}\}$.  In their application, \cite{cheng2006bounds} focus on the following parameters: (i) $E_Q[Y(j) - Y(0) \mid (D(0),D(1),D(2)) = (0,1,2)]$ for $j \in \{1,2\}$; and (ii) $Q\{(D(0),D(1),D(2)) = r_t\}$ for different  values of $r_t$.  Each of these parameters can be expressed in the form of \eqref{eq:parameter} for appropriate choices of $g$ and $\mathcal R'$.  For example, the parameter in (i) equals $E_Q[g(R) \mid R \in \mathcal R']$ for $g(R) = Y(j) - Y(0)$ and $\mathcal R' = \{(r_o,r_t): (d(0), d(1), d(2)) = (0, 1, 2)\}$.  Our framework nests their parameters within their context, while also allowing for inference on the corresponding parameters for trials with more than three treatment arms, for which their approach becomes computationally infeasible.
\end{example}

\begin{example} \label{ex:klm}
\cite{kirkeboen2016field} study the effects of college field of study on earnings. In their setting, $\mathcal D = \{0, 1, 2\}$ represents three fields of study, ordered by their (soft) admission cutoffs from  lowest to  highest. The instrument is $Z \in \{0,1,2\}$, where $Z = 1$ when the student crosses the (soft) admission cutoff for field 1 but not for field 2, $Z = 2$ when the student crosses the (soft) admission cutoff for field 2, and $Z = 0$ otherwise. The authors assume that $Z$ is exogenous in the sense that $Q$ satisfies Assumption \ref{as:exog}, and they impose: 
\begin{align} 
\label{eq:monotonicity1} Q \{D(1) = 1 \mid D(0) = 1\} & = 1~,  \\
\label{eq:monotonicity2} Q \{D(2) = 2 \mid D(0) = 2\} & = 1~.
\end{align}
These monotonicity conditions require that crossing the cutoff for field 1 or 2 weakly encourages students toward that field. They further impose the following ``irrelevance'' conditions:
\begin{align}  
\label{eq:irrelevance1} Q \{D(0) = 0, D(1) = 2\} & = 0 = Q \{D(0) = 2, D(1) = 0\} \\
\label{eq:irrelevance2} Q \{D(0) = 0, D(2) = 1\} & = 0 = Q \{D(0) = 1, D(2) = 0\}~.
\end{align}
The condition in \eqref{eq:irrelevance1} states that crossing the cutoff for field 1 does not cause a student to switch between field 0 and field 2. A similar interpretation applies to  \eqref{eq:irrelevance2}. 
The restrictions in \eqref{eq:monotonicity1}--\eqref{eq:irrelevance2} can be formulated in terms of Assumption \ref{ass:generalizedstrata} by setting $\mathcal R = \{(r_o,r_t): (d(0), d(1) , d(2))  \in \{(0, 0, 0), (0, 0, 2), (0,1,0), (0,1,2), (1, 1, 1), (1, 1, 2), (2, 1, 2),  (2,2,2)\} \}$.  \allowbreak \cite{kirkeboen2016field} acknowledge the irrelevance conditions \eqref{eq:irrelevance1}–\eqref{eq:irrelevance2} are strong. Indeed, these conditions exclude, for example, the type $(d(0), d(1), d(2)) = (0, 1, 1)$, which corresponds to a student that has a preference for field 1 over field 2 and would therefore select field 1 when above the (soft) cutoff for field 2. Nevertheless, under these conditions, \cite{kirkeboen2016field} show that the TSLS estimands identify $E_Q[Y(1)-Y(0) \mid D(1)=1, D(0)=0]$ and $E_Q[Y(2)-Y(0) \mid D(2)=2, D(0)=0]$.  Under their assumptions \eqref{eq:monotonicity1}–\eqref{eq:irrelevance2}, our analysis  allows one to analyze additional parameters that do not correspond to the TSLS estimands, including the average relative effects of field 1 versus field 2 for any given generalized principal stratum. Our framework further allows one to investigate what can be learned about these parameters without imposing \eqref{eq:irrelevance1}–\eqref{eq:irrelevance2}. For further discussion of this example and Example \ref{eg:close}, see \cite{lee2026treatment} and \cite{bai2025sharp}. 
\end{example}

\begin{example}[\textit{Revealed-Preference Restrictions}] \label{eg:unordered}
\cite{heckman2018unordered} consider a setting in which a voucher assignment  $Z\in\mathcal Z=\{0,1,2\}$ subsidizes the purchase of one of three  cars,  denoted by $A$, $B$, and $C$. They assume that the voucher assignment is randomized, so Assumption \ref{as:exog} holds. The treatment $D\in\mathcal D=\{A,B,C\}$ corresponds to the car purchased.
In this setting, \cite{heckman2018unordered} consider a series of examples in which they apply the Weak Axiom of Revealed Preference (WARP) to the voucher-induced budget-set relationships, in the first two examples below together with an additional assumption of neutral income effects on car choice, to restrict treatment response types. Throughout the examples below, we write treatment response types in the order $(d(0),d(1),d(2))$.

\begin{enumerate}[\rm (i)]
\item In their leading example, $Z=0$ corresponds to no voucher, $Z=1$ corresponds to a voucher that subsidizes the purchase of car $A$, and $Z=2$ corresponds to a voucher that subsidizes the purchase of either $B$ or $C$.  Applying WARP to the corresponding budget-set relationships, together with the additional assumption that a consumer who chooses $B$ or $C$ when neither is subsidized continues to choose the same car when both are subsidized, yields the seven admissible response types in Table III of \cite{heckman2018unordered}.  These restrictions can be formulated in terms of Assumption \ref{ass:generalizedstrata} with $\mathcal R=\mathcal Y^3\times\mathcal R_t$, where
\[\mathcal R_t=\{(A,A,A),(A,A,B),(A,A,C),(B,A,B),(B,B,B),(C,A,C),(C,C,C)\}.\]

\item In a second example, $Z=0$ corresponds to no voucher, $Z=1$ corresponds to a voucher that subsidizes the purchase of $B$, and $Z=2$ corresponds to a voucher that subsidizes the purchase of $B$ or $C$.  The same combination of WARP, the corresponding budget-set relationships, and the neutral-income-effect assumption yields the seven admissible response types in Table V of \cite{heckman2018unordered}.  These restrictions can be formulated in terms of Assumption \ref{ass:generalizedstrata} with $\mathcal R=\mathcal Y^3\times\mathcal R_t$, where 
\[
 \mathcal R_t=\{(A,A,A),(A,A,C),(A,B,B),(A,B,C), (B,B,B),(C,C,C),(C,B,C)\}. 
\]

\item Finally, in a third example, $Z=0$ corresponds to a voucher that subsidizes the purchase of $C$, $Z=1$ corresponds to a voucher that subsidizes the purchase of $B$, and $Z=2$ corresponds to a voucher that subsidizes the purchase of $B$ or $C$.  In this case, applying WARP to the corresponding budget-set relationships, without imposing the neutral-income-effect assumption used in the first two examples, yields the seven admissible response types in Table VI of \cite{heckman2018unordered}.   These restrictions can be formulated in terms of Assumption \ref{ass:generalizedstrata} with $\mathcal R=\mathcal Y^3\times\mathcal R_t$, where 
\[ \mathcal R_t=\{(A,A,A),(A,B,B),(B,B,B),(C,A,C),(C,B,B),(C,B,C),(C,C,C)\}.\]
\end{enumerate}
\cite{heckman2018unordered} provide several additional examples that fit within our framework.  
They establish conditions for identification of response-type probabilities and of certain counterfactual means and average treatment effects conditional on particular unions of principal strata; by contrast, our analysis allows inference on a more general class of parameters, including partially identified parameters. 
\end{example}

\begin{example}[\textit{Monotone treatment response}]\label{example:mtr}
Consider a setting in which the treatment admits a natural ordering and write $\mathcal D = \{0, \ldots, K_d\}$. Suppose further that $Z$ is randomly assigned, so Assumption \ref{as:exog} holds. Following \cite{manski1997monotone}, it is often reasonable to assume that $Q\{Y(j) \geq Y(k)\} = 1$ for any pair of treatments $j,k \in \mathcal D$ with $j \geq k$. For example, in \cite{manski2000monotone}, the treatment is years of schooling and the outcome is log(wage), and the restriction is that additional schooling leads to (weakly) higher wages. This restriction can be formulated in terms of Assumption \ref{ass:generalizedstrata} with $\mathcal R = \{(r_o,r_t): y(K_d) \geq \ldots \geq y(0)\}$. While \cite{manski1997monotone} develops partial-identification results for unconditional treatment parameters in this context without restrictions on treatment-response types, our analysis additionally allows such restrictions to be incorporated and enables inference on parameters that condition on generalized principal strata.
\end{example}

\clearpage
\bibliography{multivalued}

\appendix
\newpage

\begin{center}
\LARGE Online Appendix to ``Inference for Treatment Effects Conditional on Generalized Principal Strata using Instrumental Variables''    
\end{center}

\setcounter{page}{1}
\setcounter{section}{2}

\section{Auxiliary Lemmas}\label{sec:auxlemma}



\begin{lemma}\label{lm:check}
If Assumptions \ref{ass:test}(i)-(iii) hold and $\log^2(|\mathcal M|)|\mathcal M|/n = o(1)$ then Assumptions 4.1, 4.2, 4.3, 4.4(v), and A.1 of \cite{fang2023inference} are satisfied for the test defined in \eqref{eq:test} with $a_n = \log(|\mathcal M|)\sqrt{|\mathcal M| /n}$ and $M_{3,\Psi} = \sqrt{|\mathcal M|}$.  
\end{lemma}

\begin{proof}
Assumption 4.1(i) in \cite{fang2023inference} is equivalent to our Assumption \ref{ass:test}(i), while their Assumption 4.1(ii) holds with $a_n = \log(|\mathcal M|)\sqrt{|\mathcal Z|/ n}$ by Lemma \ref{aux:lin} and Assumption \ref{ass:test}(ii).
Next note that Assumptions 4.2(i)(ii) in \cite{fang2023inference} hold by Lemma \ref{aux:cov}(i)(ii), their Assumption 4.2(iii) holds with $M_{3,\Psi} = \sqrt{|\mathcal M|}$ by Lemma \ref{aux:cov}(iv), and their Assumption 4.2(iv) holds due to $\Omega(P)$ being diagonal.
Further observe their Assumption 4.3(i) holds due to the first $|\mathcal M|$ diagonal entries of $\Omega(P)$ being non-zero and the final two rows of $\psi(V,P)$ being zero. 
In turn, Assumptions 4.3(ii) and 4.4(v) in \cite{fang2023inference} hold by Lemma \ref{aux:cov}(v) and Assumption \ref{ass:test}(ii).
Finally, we note that Assumption A.1 in \cite{fang2023inference} holds with $a_n = \log(|\mathcal M|)\sqrt{|\mathcal M|/n}$ by Lemma \ref{aux:var}.
\end{proof}

\begin{lemma} \label{lm:coup}
Suppose Assumptions \ref{ass:test}(i) and \ref{ass:test}(iii) hold, $\log(|\mathcal M|)|\mathcal M|/n = o(1)$, and for $\psi(V,P)$ as in \eqref{eq:psidef} let $\Sigma(P) \equiv E_P[\psi(V,P)\psi(V,P)^\prime]$. Then, it follows that
\vspace{-0.1 in}
\begin{packed_enum}
    \item[(i)] There exists a Gaussian vector $\mathbb G(P)\sim N(0,\Sigma(P))$ satisfying uniformly in $P\in \mathbf P$ 
\begin{equation} \label{eq:coupling}
\left \| \Omega(P)^\dagger \left ( \sqrt n(\hat \beta_n - \beta(P)) - \mathbb G(P) \right ) \right \|_\infty = O_P\left(\frac{\log^{3/2}(n)\sqrt{|\mathcal M|}}{\sqrt n}\right)~.  \end{equation} 

    \item[(ii)] There exists a Gaussian vector $\mathbb G^\star(P)\sim N(0,\Sigma(P))$ that is independent of the sample $\{V_i\}_{i=1}^n$ and satisfying uniformly in $P\in \mathbf P$
\begin{equation} \label{eq:bootcoupling}
\left \| \Omega(P)^\dagger \left ( \sqrt n(\hat \beta_n^* - \hat \beta_n) - \mathbb G^\star(P) \right ) \right \|_\infty = O_P\left(\frac{\log^{3/2}(n)\sqrt{|\mathcal M|}}{\sqrt n}\right) ~.   
\end{equation}
\end{packed_enum}
\end{lemma}

\noindent \emph{Proof.} First note that Lemma \ref{aux:lin} allows us to conclude, uniformly in $P\in \mathbf P$, that we have
\begin{equation}\label{lm:coup1}
\|\Omega(P)^\dagger \{\sqrt n\{\hat \beta_n - \beta(P)\} - \frac{1}{\sqrt n}\sum_{i=1}^n \psi(V_i,P)\}\|_\infty = O_P\left(\frac{\log(|\mathcal M|)\sqrt{|\mathcal Z|}}{\sqrt n}\right) ~.
\end{equation}
To establish the lemma, we next  apply Theorem 1 in \cite{massart1989strong} to couple the vector $\sum_i \psi(V_i,P)/\sqrt n$ to a Gaussian variable $\mathbb G(P)$.
To this end, we let $\{U_i\}_{i=1}^n$ be an i.i.d.\ sample with $U_i \sim \text{Unif}[0,1]$. 
We also divide $[0,1]$ into $|\mathcal Z|$ disjoint intervals, denoted $C_{z}(P)$, satisfying $P(U_i \in C_{z}(P)) = P_{z}$.
We further subdivide each interval $C_z$ into disjoint subintervals, denoted $C_{ydz}(P)$, satisfying $P(U_i \in C_{ydz}(P)) = P_{ydz}$ for any $(y,d)$ such that $(y,d,z)\in \mathcal M$.
In what follows, we will assume that $(Y_i,D_i,Z_i)$ are generated from $U_i$ according to  
\begin{equation}\label{lm:coup2}
I\{Y_i = y, D_i = d, Z_i = z\} = I\{U_i \in C_{ydz}(P)\}~,    
\end{equation}
which we note is without loss of generality in that the distribution of $(Y_i,D_i,Z_i)$ is still $P$ in this probability space.
Next, let $\mathcal S$ denote the class of all intervals contained in $[0,1]$. 
Because all intervals are convex, $\mathcal S$ satisfies the uniform Minkowski condition in Definition 2 in \cite{massart1989strong}. 
Moreover, since $\mathcal S$ is a VC-class, it follows that $\mathcal S$ satisfies assumption $H(\zeta)$ in \cite{massart1989strong} with $\zeta = 0$.
Defining $\mathbb G_n \in \ell^\infty(\mathcal S)$ by
\begin{equation*}
\mathbb G_n(S) \equiv \frac{1}{\sqrt n}\sum_{i=1}^n (I\{U_i \in S\} - P(U \in S))    
\end{equation*}
for any $S\in \mathcal S$, we then obtain from Theorem 1 in \cite{massart1989strong} that there is a Brownian bridge $\mathbb W$ satisfying
\begin{equation}\label{lm:coup4}
\text{Cov}[\mathbb G_n(S),\mathbb G_n(S^\prime)] = \text{Cov}[\mathbb W(S),\mathbb W(S^\prime)]
\end{equation}
for all $S,S^\prime \in \mathcal S$ and such that for some constants $K$, $\Lambda$, and $\theta$ depending only on $\mathcal S$ we have
\begin{equation}\label{lm:coup5}
P\left\{\sup_{S\in \mathcal S} |\mathbb G_n(S) - \mathbb W(S)| > \frac{\sqrt{\log(n)}}{\sqrt n}(t + K\log(n))\right\} \leq \Lambda \exp(-\theta t)   
\end{equation}
 for all $t > 0$. In particular,  by setting $t$ to be large enough in \eqref{lm:coup5} we obtain that
\begin{equation}\label{lm:coup6}
\sup_{S \in \mathcal S}|\mathbb G_n(S) - \mathbb W(S)| = O_P\left(\frac{\log^{3/2}(n)}{\sqrt n}\right) ~.\end{equation}
To conclude, define a Gaussian vector $\mathbb G(P)$ of dimension $|\mathcal M| + 2$ by letting its last two coordinates equal zero and its first $|\mathcal M|$ coordinates, indexed by $(y,d,z)\in \mathcal M$, be 
\begin{equation}\label{lm:coup6p5}
\mathbb G_{ydz}(P) \equiv \frac{1}{P_z}\mathbb W(C_{ydz}(P)) - \frac{P_{ydz}}{P_z^2}\mathbb W(C_z(P))~.     
\end{equation}
Similarly, by \eqref{lm:coup2}, definition of $\psi(V,P)$,  and the construction of $C_{ydz}(P)$ and $C_z(P)$ we have
\begin{equation}\label{lm:coup7}
\frac{1}{\sqrt n}\sum_{i=1}^n \psi_{ydz}(V_i,P) = \frac{1}{P_z}\mathbb G_n(C_{ydz}(P)) - \frac{P_{ydz}}{P_z^2}\mathbb G_n(C_z(P))~.    
\end{equation}
In particular, results \eqref{lm:coup4} and \eqref{lm:coup7}  imply that $E[\mathbb G(P)\mathbb G(P)^\prime] = \Sigma(P)$ as desired.
Moreover, by definition of $\Omega(P)$, Assumption \ref{ass:test}(iii), and \eqref{lm:coup6p5} and \eqref{lm:coup7} we also have that
\begin{align}\label{lm:coup8}
\Bigg\|(\Omega(P))^\dagger & \left\{\frac{1}{\sqrt n}\sum_{i=1}^n\psi(V_i,P) - \mathbb G(P)\right\}\Bigg\|_\infty \notag \\
& = \max_{(y,d,z)\in \mathcal M} \frac{P_z^{1/2}}{P_{yd|z}^{1/2}(1-P_{yd|z})^{1/2}}\left|\frac{1}{\sqrt n}\sum_{i=1}^n \psi_{ydz}(V_i,P) - \mathbb G_{ydz}(P)\right| \notag \\ 
& \lesssim \max_{(y,d,z)\in \mathcal M}\frac{P_z^{1/2}}{P_{yd|z}^{1/2}}\left(\frac{1}{P_z} + \frac{P_{ydz}}{P_z^2}\right)  \times \sup_{S\in \mathcal S} |\mathbb G_n(S) - \mathbb W(S)| \notag \\
& = O_P\left(\frac{\log^{3/2}(n)\sqrt{|\mathcal M|}}{\sqrt n}\right)~,
\end{align}
where the final result holds uniformly in $P\in \mathbf P$ by Assumption \ref{ass:test}(iii) and result \eqref{lm:coup6}. 
The first claim of the lemma therefore follows from \eqref{lm:coup1} and \eqref{lm:coup8}. 

In order to establish the second claim of the lemma, we first define the vector $\hat{\mathbb G}_n$ by
\begin{equation}\label{lm:coup9}
\hat {\mathbb G}_n \equiv \frac{1}{\sqrt n} \sum_{i=1}^n   (\psi(V_i^*, P)-\frac{1}{n}\sum_{j=1}^n \psi(V_j,P)) 
\end{equation}
Given this notation, then note that Lemma \ref{aux:bootlin} implies, uniformly in $P\in \mathbf P$, that
\begin{equation}\label{lm:coup10}
\|(\Omega(P))^\dagger \{\sqrt n\{\hat \beta_n^* - \hat \beta_n\} - \hat {\mathbb G}_n\}\|_\infty = O_P\left(\frac{\log(|\mathcal M|) \sqrt{|\mathcal M|}}{\sqrt n}\right)~.
\end{equation}
Applying the same construction based on \cite{massart1989strong} conditionally on $\{V_i\}_{i=1}^n$ implies that there is $\mathbb G(\hat P)$ satisfying $\mathbb G(\hat P)\sim N(0,\text{Var}_{\hat P}\{\psi(V,P)\})$ conditionally on $\{V_i\}_{i=1}^n$ and
\begin{equation}\label{lm:coup11}
P\left\{ \|\hat {\mathbb G}_n - {\mathbb G(\hat P)}\|_\infty > C_1\frac{|\mathcal Z|\sqrt{\log(n)}}{\sqrt n}(t+ K\log(n))\Big| \{V_i\}_{i=1}^n \right\} \leq \Lambda \exp\{-\theta t\}    
\end{equation}
for some $C_1 < \infty$ not depending on $\{V_i\}_{i=1}^n$.
Further let $\mathcal B$ denote the Borel $\sigma$-field, and for any $B\in \mathcal B$ denote its $\epsilon$-enlargement under $\|\cdot\|_\infty$ by $B^\epsilon \equiv \{\tilde b : \inf_{b\in B}\|b-\tilde b\|_\infty \leq \epsilon\}$.
Defining $\delta_n = 2C_1K\log^{3/2}(n)|\mathcal Z|/\sqrt n$ and setting $t = C_2$, we then obtain from result \eqref{lm:coup11} and Strassen's Theorem (see, e.g., Theorem 10.3 in \cite{pollard2002user}), that
\begin{equation}\label{lm:coup12}
\sup_{P\in \mathbf P}E_P\left[\sup_{B\in \mathcal B} \left\{P\left\{\hat {\mathbb G}_n\in B\Big|\{V_i\}_{i=1}^n\right\} - P\left\{\mathbb G(\hat P)\in B^{K\delta_n}\Big| \{V_i\}_{i=1}^n\right\} \right\}\right] \leq \Lambda \exp\{-\theta C_2\}
\end{equation}
for $n$ sufficiently large.
Since the right hand side of \eqref{lm:coup12} can be made arbitrarily small by setting $C_2$ sufficiently large, we obtain from Theorem 4 in \cite{monrad1991nearby} that there exists a random variable $\bar{\mathbb G}(P)$ with distribution $N(0,\text{Var}_{\hat P}\{\psi(V,P)\})$ conditionally on $\{V_i\}_{i=1}^n$ and such that uniformly in $P\in \mathbf P$ it satisfies
\begin{equation}\label{lm:coup13}
\|\hat {\mathbb G}_n - \bar {\mathbb G}(P)\|_\infty = O_P\left(\frac{\log^{3/2}(n)|\mathcal Z|}{\sqrt n}\right)~.    
\end{equation}
Moreover, result \eqref{lm:coup13}, the definition of $\Omega(P)$, and Assumption \ref{ass:test}(iii) yield that 
\begin{multline}\label{lm:coup14}
\|(\Omega(P))^\dagger\{\hat {\mathbb G}_n - \bar{\mathbb G}(P)\}\|_\infty \\ \leq \sup_{(y,d,z)\in \mathcal M} \left(\frac{P_z}{P_{yd|z}(1-P_{yd|z})}\right)^{1/2} \|\hat{\mathbb G}_n - \bar {\mathbb G}(P)\|_\infty = O_P\left(\frac{\log^{3/2}(n)\sqrt{|\mathcal M|}}{\sqrt n} \right)
\end{multline}
uniformly in $P\in \mathbf P$. 
However, since $(\Omega(P))^\dagger \bar {\mathbb G}(P)\sim N(0,(\Omega(P))^\dagger \text{Var}_{\hat P}\{\psi(V,P)\}(\Omega(P))^\dagger)$ conditionally on the data, we may apply the arguments in Lemma A.7 in \cite{fang2023inference} (see in particular their equation (A.57)) and Lemma \ref{lm:auxop} to conclude that there is a $\bar {\mathbb G}^\star(P) \sim N(0,(\Omega(P))^\dagger \text{Var}_P\{\psi(V,P)\}(\Omega(P))^\dagger)$ that is independent of the data, and satisfies 
\begin{equation}\label{lm:coup15}
\|(\Omega(P))^\dagger \bar {\mathbb G}(P) - \bar {\mathbb G}^\star(P)\|_\infty = O_P\left(\frac{\log(|\mathcal M|)\sqrt{|\mathcal M|}}{\sqrt{n}}\right)    
\end{equation}
uniformly in $P\in \mathbf P$. 
Finally, let $\mathbb G^\star(P) = \Omega(P) \bar{\mathbb G}^\star(P)$ and note that $\mathbb G^\star(P)$ is independent of the data, because $\bar{\mathbb G}^\star(P)$ is, and $\mathbb G^\star(P)\sim N(0,\text{Var}_P\{\psi(V,P)\})$ because the columns of $\text{Var}_P\{\psi(V,P)\}$ are in the range of $\Omega(P)$. 
The second claim of the lemma then follows from results \eqref{lm:coup10}, \eqref{lm:coup14}, and \eqref{lm:coup15}. \qed

\begin{lemma}\label{aux:lin}
Let Assumptions \ref{ass:test}(i) and \ref{ass:test}(iii) hold and $\psi(V,P)$ be as defined in \eqref{eq:psidef}. If $\log(|\mathcal M|)|\mathcal M|/n = o(1)$, then it follows that, uniformly in $P\in \mathbf P$, we have
$$\|(\Omega(P))^\dagger \{\sqrt n\{\hat \beta_n - \beta(P)\} - \frac{1}{\sqrt n}\sum_{i=1}^n\psi(V_i,P)\}\|_\infty = O_P\left(\frac{\log(|\mathcal M|) \sqrt{|\mathcal Z|}}{\sqrt n}\right)~.$$  
\end{lemma}

\begin{proof}
We first note that by Lemma \ref{aux:prob}(iv), $\min_{z\in \mathcal Z} \hat P_z > 0$ with probability tending to one uniformly in $P\in \mathbf P$.
Therefore, the first $|\mathcal M|$ coordinates of $\hat \beta_n -\beta(P)$ have the following structure for some $(y,d,z)\in \mathcal M$ with probability tending to one
\begin{equation*}
\frac{\hat P_{ydz}}{\hat P_z} - \frac{P_{ydz}}{P_z} = \frac{1}{P_z}(\hat P_{ydz} - P_{ydz}) - \frac{P_{ydz}}{P_z^2}(\hat P_z - P_z) + \hat \gamma_{ydz}~,
\end{equation*}
where
\begin{equation*}
\hat \gamma_{ydz} = (\hat P_{ydz} - P_{ydz}) \left ( \frac{1}{\hat P_z} - \frac{1}{P_z} \right ) - \frac{P_{ydz}}{P_z} \left ( \frac{1}{\hat P_z} - \frac{1}{P_z} \right )(\hat P_z - P_z)~.
\end{equation*}
Since the final two rows of $\{\hat \beta_n - \beta(P)\}$ are identically zero, the definition of $\Omega(P)$ implies
\begin{multline}\label{eq:lin3}
\left\|(\Omega(P))^\dagger \left\{\sqrt n\{\hat \beta_n - \beta(P)\} - \frac{1}{\sqrt n}\sum_{i=1}^n\psi(V_i,P)\right\}\right\|_\infty \\ = \max_{(y, d, z) \in \mathcal M} \left | \left ( \frac{P_z}{P_{yd|z} (1 - P_{yd|z})} \right )^{1/2} \sqrt n \hat \gamma_{ydz} \right |  
\lesssim  \left ( \frac{|\mathcal M|}{|\mathcal Z|^2} \right )^{1/2} \times
\max_{(y, d, z) \in \mathcal M}  \sqrt n \left |\hat \gamma_{ydz} \right |~,  
\end{multline}
where the inequality follows from Assumption \ref{ass:test}(iii).
Next note Lemma \ref{aux:prob} implies that
\begin{multline}\label{eq:lin4}
\max_{(y,d,z)\in \mathcal M} \left|(\hat P_{ydz} - P_{ydz})\left(\frac{1}{\hat P_z} - \frac{1}{P_z}\right)\right| \\ \leq    \max_{(y,d,z)\in \mathcal M} |\hat P_{ydz} - P_{ydz}| \times \max_{z\in \mathcal Z} |\hat P_z - P_z| \times \max_{z\in \mathcal Z} \frac{1}{P_z^2}\times O_P(1) = O_P\left(\frac{\log(|\mathcal M|)|\mathcal Z|^2}{n \sqrt{|\mathcal M| |\mathcal Z|}}\right)
\end{multline}
uniformly in $P\in \mathbf P$.
Similarly, Lemma \ref{aux:prob} and Assumption \ref{ass:test}(iii) also yield that
\begin{multline}\label{eq:lin5}
\max_{(y,d,z)\in \mathcal M} \left|\frac{P_{ydz}}{P_z}\left(\frac{1}{\hat P_z} - \frac{1}{P_z}\right)(\hat P_z - P_z)\right|\\
\leq \max_{(y,d,z)\in \mathcal M} \frac{P_{ydz}}{P_z^3} \times \max_{z\in \mathcal Z} |\hat P_z - P_z|^2 \times O_P(1) = O_P\left( \frac{|\mathcal Z|^2\log(|\mathcal Z|)}{|\mathcal M| n}\right)
\end{multline}
uniformly in $P\in \mathbf P$.
The claim of the lemma then follows from combining results  \eqref{eq:lin3}, \eqref{eq:lin4}, and \eqref{eq:lin5}.
\end{proof}

\begin{lemma}\label{aux:bootlin}
Let Assumptions \ref{ass:test}(i) and \ref{ass:test}(iii) hold and $\psi(V,P)$ be as defined in \eqref{eq:psidef}. 
If $\log(|\mathcal M|)|\mathcal M|/n = o(1)$, then it follows that, uniformly in $P\in \mathbf P$, we have
$$\|(\Omega(P))^\dagger \{\sqrt n\{\hat \beta_n^* - \hat \beta_n\} - \frac{1}{\sqrt n}\sum_{i=1}^n(\psi(V_i^*, P)-\frac{1}{n}\sum_{j=1}^n \psi(V_j,P))\}\|_\infty = O_P\left(\frac{\log(|\mathcal M|) \sqrt{|\mathcal M|}}{\sqrt n}\right)~.$$  
\end{lemma}

\begin{proof}
The proof follows similar arguments to those employed in the proof of Lemma \ref{aux:lin}.
We first note that by Lemmas \ref{aux:prob}(iv) and \ref{aux:bootprob}(iv), $\min_{z\in \mathcal Z} \hat P^*_z \wedge \hat P_z > 0$ with probability tending to one uniformly in $P\in \mathbf P$.
Therefore, the first $|\mathcal M|$ coordinates of $\hat \beta_n^* - \hat \beta_n$ have the following structure for some $(y,d,z)\in \mathcal M$ with probability tending to one
\begin{equation} \label{eq:bootlin1}
\frac{\hat P_{ydz}^*}{\hat P^*_z} - \frac{\hat P_{ydz}}{\hat P_z} = \frac{1}{\hat P_z}(\hat P^*_{ydz} - \hat P_{ydz}) - \frac{\hat P_{ydz}}{\hat P_z^2}(\hat P^*_z - \hat P_z) + \hat \gamma^*_{ydz}~,
\end{equation}
where
\begin{equation} \label{eq:bootlin2}
\hat \gamma^*_{ydz} = (\hat P^*_{ydz} - \hat P_{ydz}) \left ( \frac{1}{\hat P^*_z} - \frac{1}{\hat P_z} \right ) - \frac{\hat P_{ydz}}{\hat P_z} \left ( \frac{1}{\hat P^*_z} - \frac{1}{\hat P_z} \right )(\hat P^*_z - \hat P_z)~.
\end{equation}
Since the final two rows of $\{\hat \beta_n^* - \hat \beta_n\}$ are identically zero, the definition of $\Omega(P)$ implies
\begin{multline}\label{eq:bootlin3}
\|(\Omega(P))^\dagger \{\sqrt n\{\hat \beta^*_n - \hat \beta_n\} - \frac{1}{\sqrt n}\sum_{i=1}^n\psi(V_i^*,\hat P)\}\|_\infty \\ = \max_{(y, d, z) \in \mathcal M} \left | \left ( \frac{P_z}{P_{yd|z} (1 - P_{yd|z})} \right )^{1/2} \sqrt n \hat \gamma^*_{ydz} \right |  \lesssim  \left ( \frac{|\mathcal M|}{|\mathcal Z|^2} \right )^{1/2} \times
\max_{(y, d, z) \in \mathcal M}  \sqrt n \left |\hat \gamma_{ydz}^* \right |~,  
\end{multline}
where the inequality follows from Assumption \ref{ass:test}(iii).
Next note that Lemma \ref{aux:bootprob}(iv) yields
\begin{multline}\label{eq:bootlin4}
\max_{(y,d,z)\in \mathcal M} \left|(\hat P^*_{ydz} - \hat P_{ydz})\left(\frac{1}{\hat P^*_z} - \frac{1}{\hat P_z}\right)\right| \\ \leq \max_{(y,d,z)\in \mathcal M} |\hat P^*_{ydz} - \hat P_{ydz}| \times \max_{z\in \mathcal Z} |\hat P^*_z - \hat P_z| \times \max_{z\in \mathcal Z} \frac{1}{\hat P_z^2}\times O_P(1) = O_P\left(\frac{\log(|\mathcal M|)|\mathcal Z|^2}{n \sqrt{|\mathcal M| |\mathcal Z|}}\right)
\end{multline}
uniformly in $P\in \mathbf P$, where in the final result we used Lemmas \ref{aux:bootprob}(i)(ii), Lemma \ref{aux:prob}(iv), and Assumption \ref{ass:test}(iii).
Similarly, applying Lemmas \ref{aux:prob}(iii)(iv) and Lemma \ref{aux:bootprob}(iv) yields, uniformly in $P\in \mathbf P$, that
\begin{multline}\label{eq:bootlin5}
\max_{(y,d,z)\in \mathcal M} \left|\frac{\hat P_{ydz}}{\hat P_z}\left(\frac{1}{\hat P^*_z} - \frac{1}{\hat P_z}\right)(\hat P^*_z - \hat P_z)\right|\\
\leq \max_{(y,d,z)\in \mathcal M} \frac{P_{ydz}}{P_z^3} \times \max_{z\in \mathcal Z} |\hat P_z^* - \hat P_z|^2 \times O_P(1) = O_P\left( \frac{|\mathcal Z|^2\log(|\mathcal Z|)}{|\mathcal M| n}\right)
\end{multline}
where the final equality used Lemma \ref{aux:bootprob}(ii) and Assumption \ref{ass:test}(iii). 
Hence, we obtain
\begin{equation}\label{eq:bootlin6}
\|(\Omega(P))^\dagger\{\sqrt n\{\hat \beta_n^*-\hat \beta_n\} - \frac{1}{\sqrt n}\sum_{i=1}^n \psi(V_i^*,\hat P)\}\|_\infty = O_P\left(\frac{\log(|\mathcal M|) \sqrt{|\mathcal Z|}}{\sqrt n}\right)    
\end{equation}
uniformly in $P\in \mathbf P$, by results \eqref{eq:bootlin3}, \eqref{eq:bootlin4}, and \eqref{eq:bootlin5}.
Next, note that for any $(y,d,z)\in \mathcal M$ 
\begin{equation}\label{eq:bootlin7}
 \psi_{ydz}(V_i^*,\hat P) 
=  \frac{1}{\hat P_z}(I\{Y_i^*=y,D_i^*=d,Z_i^*=z\} - \hat P_{ydz}) - \frac{\hat P_{ydz}}{\hat P_z^2}(I\{Z_i^*=z\}-\hat P_z)   ~,
\end{equation}
and therefore
\begin{multline*}
\psi_{ydz}(V_i^*,\hat P) - (\psi_{ydz}(V_i^*,P)-\frac{1}{n}\sum_{j=1}^n\psi_{ydz}(V_j,P)) \\ = \left(\frac{1}{\hat P_z} - \frac{1}{P_z}\right)\left(I\{Y_i^*=y,D_i^*=d,Z_i^*=z\} - \hat P_{ydz}\right) 
- \left(\frac{\hat P_{ydz}}{\hat P_z^2} - \frac{P_{ydz}}{P_z^2}\right)\left(I\{Z_i^*=z\}-\hat P_z\right) .
\end{multline*}
Therefore, the definition of $\Omega(P)$ and Assumption \ref{ass:test}(iii) together allow us to conclude that
\begin{multline}\label{eq:bootlin9}
\|(\Omega(P))^\dagger\{\frac{1}{\sqrt n}\sum_{i=1}^n\psi(V_i^*,\hat P) -  \frac{1}{\sqrt n}\sum_{i=1}^n(\psi(V_i^*, P)-\frac{1}{n}\sum_{j=1}^n \psi(V_j,P))\}\|_\infty \\
\lesssim \max_{(y,d,z)\in \mathcal M} \left|\left(\frac{P_z}{P_{yd|z}}\right)^{1/2} \sqrt n\left\{\left(\frac{1}{\hat P_z} - \frac{1}{P_z}\right)\left(\hat P_{ydz}^* - \hat P_{ydz}\right) - \left(\frac{\hat P_{ydz}}{\hat P_z^2}-\frac{P_{ydz}}{P_z^2}\right)\left(\hat P_z^*-\hat P_z\right)\right\}\right|.
\end{multline}
Next, observe that Assumption \ref{ass:test}(iii) and Lemmas \ref{aux:prob}(ii)(iv) and \ref{aux:bootprob}(i) imply that
\begin{multline}\label{eq:bootlin10}
\max_{(y,d,z)\in \mathcal M}  \left|\left(\frac{P_z}{P_{yd|z}}\right)^{1/2}\sqrt n \left(\frac{1}{\hat P_z} - \frac{1}{P_z}\right)\left(\hat P_{ydz}^* - \hat P_{ydz}\right)\right| \\
\lesssim \frac{\sqrt{n|\mathcal M|}}{|\mathcal Z|} \times \max_{z\in \mathcal Z}\left|\frac{\hat P_z-P_z}{\hat P_z P_z}\right|\times \max_{(y,d,z)\in \mathcal M}|\hat P_{ydz}^*- \hat P_{ydz}| 
= O_P\left(\frac{\log(|\mathcal M|)\sqrt{|\mathcal Z|}}{\sqrt n}\right)
\end{multline}
uniformly in $P\in \mathbf P$.
Moreover, Lemmas \ref{aux:prob}(ii)(v) and Assumption \ref{ass:test}(iii) also imply that
\begin{multline}\label{eq:bootlin11}
\max_{(y,d,z)\in \mathcal M}|\hat P_{yd|z}P_z - P_{yd|z}\hat P_z| \\
\leq \max_{(y,d,z)\in \mathcal M}\left\{|\hat P_{yd|z}-P_{yd|z}|P_z + P_{yd|z}|\hat P_z - P_z|\right\} = O_P\left(\frac{\sqrt{\log(|\mathcal M|)}}{\sqrt{n|\mathcal Z|}}\right)
\end{multline}
uniformly in $P\in \mathbf P$.
Therefore, result \eqref{eq:bootlin11}, Assumption \ref{ass:test}(iii), and Lemmas \ref{aux:prob}(iv) and \ref{aux:bootprob}(ii) yield
\begin{multline}\label{eq:bootlin12}
\max_{(y,d,z)\in \mathcal M}  \left|\left(\frac{P_z}{P_{yd|z}}\right)^{1/2}\sqrt n\left(\frac{\hat P_{ydz}}{\hat P_z^2}-\frac{P_{ydz}}{P_z^2}\right)\left(\hat P_z^*-\hat P_z\right)\right| \\
\lesssim \frac{\sqrt{n|\mathcal M|}}{|\mathcal Z|} \times  \max_{(y,d,z)\in \mathcal M}\left|\frac{\hat P_{yd|z}P_z - P_{yd|z}\hat P_z}{\hat P_z P_z}\right|\times \max_{z\in \mathcal Z}|\hat P_z^*- \hat P_z|  
= O_P\left(\frac{\log(|\mathcal M|)\sqrt{|\mathcal M|}}{\sqrt n}\right) 
\end{multline}
uniformly in $P\in \mathbf P$.
The claim of the lemma then follows from \eqref{eq:bootlin6}, \eqref{eq:bootlin9}, \eqref{eq:bootlin10}, and \eqref{eq:bootlin12}.
\end{proof}

\begin{lemma}\label{aux:var}
Let Assumptions \ref{ass:test}(i) and \ref{ass:test}(iii) hold, and let $\|C\|_{o,\infty} = \sup_{\|a\|_\infty \leq 1} \|Ca\|_\infty$ for any $p\times p$ matrix $C$.
If $\log(|\mathcal M|) |\mathcal M|/n = o(1)$, then it follows that  $\|\Omega(P)^\dagger(\hat \Omega_n - \Omega(P))\|_{o,\infty} = O_P(\sqrt{\log(|\mathcal M|)|\mathcal M|/n})$ uniformly in $P\in \mathbf P$.
\end{lemma}

\begin{proof}
For any constant $c \geq 0$, we have $|c-1| = |\sqrt c - 1| |\sqrt c + 1|$ and therefore that $|\sqrt c - 1|\leq |c-1|$.
Therefore, employing the definitions of $\Omega(P)$ and $\hat \Omega_n$ yields that
\begin{align}\label{aux:var1}
\|\Omega(P)^\dagger(\hat \Omega_n - \Omega(P))\|_{o,\infty}&  = \max_{(y,d,z)\in \mathcal M}\left|\frac{\hat P_{yd|z}^{1/2}(1-\hat P_{yd|z})^{1/2}}{\hat P_z^{1/2}}\times  \frac{P_z^{1/2}}{ P_{yd|z}^{1/2}(1-P_{yd|z})^{1/2}} -1\right|\notag  \\
& \leq \max_{(y,d,z)\in \mathcal M}\left|\frac{\hat P_{ydz}(1-\hat P_{yd|z})}{\hat P_z^2}\times  \frac{P_z^2}{ P_{ydz}(1-P_{yd|z})} -1\right|~.
\end{align}
Next, note that $(a^2-b^2) = (a-b)(a+b)$, Lemmas \ref{aux:prob}(ii)(iv) and Assumption \ref{ass:test}(iii) imply
\begin{equation}\label{aux:var2}
\max_{z\in \mathcal Z} \left | \frac{P_z^2}{\hat P_z^2}-1 \right | \leq \max_{z\in \mathcal Z} \frac{1}{\hat P_z^2} \times \max_{z\in \mathcal Z} |\hat P_z-P_z|\times \max_{z\in \mathcal Z} |\hat P_z + P_z| = O_P\left(\frac{\sqrt{\log(|\mathcal Z|)|\mathcal Z|}}{\sqrt n}\right)
\end{equation}
uniformly in $P\in \mathbf P$.
Similarly relying on Assumption \ref{ass:test}(iii) and Lemma \ref{aux:prob}(v) yields
\begin{equation}\label{aux:var3}
\max_{(y,d,z)\in \mathcal M} \left|\frac{(1-\hat P_{yd|z})}{(1-P_{yd|z})} -1\right| = \max_{(y,d,z)\in \mathcal M} \left|\frac{(P_{yd|z}-\hat P_{yd|z})}{(1-P_{yd|z})} \right|=O_P\left( \frac{\sqrt{\log(|\mathcal M|)|\mathcal Z|}}{\sqrt n}\right)
\end{equation}
uniformly in $P\in \mathbf P$. The lemma then follows from \eqref{aux:var1}, Lemma \ref{aux:prob}(iii), and results \eqref{aux:var2} and \eqref{aux:var3}.
\end{proof}

\begin{lemma}\label{lm:auxop}
Let Assumptions \ref{ass:test}(i) and \ref{ass:test}(iii) hold, $\psi(V,P)$ be as defined in \eqref{eq:psidef} and for any $k\times k$ symmetric matrix $M$ let $\|M\|_o$ denote its largest absolute eigenvalue. If $\delta_n^2 \equiv \log(|\mathcal M|)|\mathcal M|/n$ satisfies $\delta_n = o(1)$, then uniformly in $P\in \mathbf P$ it follows that 
\begin{align}
\|(\Omega(P))^\dagger({\rm Var}_{\hat P}\{\psi(V,P)\} - {\rm Var}_P\{\psi(V,P)\})(\Omega(P))^\dagger\|_o & = O_P\left(\delta_n  \right)    \label{lm:auxopdisp1}\\
\|(\Omega(P)^\dagger{\rm Var}_{\hat P}\{\psi(V,P)\}\Omega(P)^\dagger)^{1/2} - (\Omega(P)^\dagger {\rm Var}_P\{\psi(V,P)\}\Omega(P)^\dagger)^{1/2}\|_o & = O_P\left(\delta_n  \right)  .  \label{lm:auxopdisp2} 
\end{align}
\end{lemma}

\begin{proof}
First define the vector $\bar \psi_n(P) \equiv \sum_i \psi(V_i,P)/n$ for notational simplicity.
Next, note that for any unit-length vector $v$ (under $\|\cdot\|_2$), the Cauchy-Schwarz inequality implies
\begin{equation*}
\|\bar\psi_n(P) \bar \psi_n(P)^\prime v\|_2 = \|\bar \psi_n(P)\|_2 |\bar \psi_n(P)^\prime v| \leq \|\bar \psi_n(P)\|_2^2 = \sum_{(y,d,z)\in \mathcal M} \left(\frac{1}{n}\sum_{i=1}^n \psi_{ydz}(V_i,P)\right)^2,
\end{equation*}
where the final equality follows from the definition of $\psi(V,P)$.
Therefore, since Lemma \ref{aux:cov}(i) implies that $E_P[\psi(V,P)] = 0$ for all $P\in \mathbf P$ and the sample $\{V_i\}_{i=1}^n$ is i.i.d., Assumption \ref{ass:test}(iii) and the definitions of $\Omega(P)$ and $\psi(V,P)$ yield the bound
\begin{multline}\label{lm:auxop2}
\sup_{P\in \mathbf P} E_P\left[\|(\Omega(P))^\dagger\bar \psi_n \bar \psi_n^\prime(\Omega(P))^\dagger\|_o\right] \leq \sup_{P\in \mathbf P}\left\{\|(\Omega(P))^\dagger\|_o^2 \times \frac{1}{n} \sum_{(y,d,z)\in \mathcal M} E_P\left[\psi^2_{ydz}(V,P)\right]\right\}  \\ \leq  \sup_{P\in \mathbf P}\left\{ \max_{(y,d,z)\in \mathcal M} \frac{P_z}{P_{yd|z}(1-P_{yd|z})}\times  \frac{1}{n} \sum_{(y,d,z)\in \mathcal M} \left(\frac{P_{ydz}}{P_z^2} + \frac{P_{ydz}^2}{P_z^3}\right)\right\} \lesssim \frac{|\mathcal M|}{n}~.
\end{multline}
Next, set $M_i(P) \equiv (\Omega(P))^\dagger(\psi(V_i,P)\psi(V_i,P)^\prime - E_P[\psi(V,P)\psi(V,P)^\prime])(\Omega(P))^\dagger/n$, and note result \eqref{lm:auxop2}, the triangle inequality, and Markov's inequality yield, uniformly in $P\in \mathbf P$, that
\begin{equation}\label{lm:auxop3}
\|(\Omega(P))^\dagger\{{\rm Var}_{\hat P}\{\psi(V,P)\} - {\rm Var}_{P}\{\psi(V,P)\}\}(\Omega(P))^\dagger\|_o \leq \|\sum_{i=1}^n M_i(P)\|_o + O_P\left(\frac{|\mathcal M|}{n}\right)~.    
\end{equation}
Moreover, for any unit length vector $v$ (under $\|\cdot\|_2)$ the Cauchy-Schwarz inequality implies
\begin{align}\label{lm:auxop4}
\|\psi(V,P)\psi(V,P)^\prime v\|_2 & \leq \|\psi(V,P)\|_2^2  \notag \\
& =\sum_{(y,d,z)\in \mathcal M}\left(\frac{1}{P_z}I\{Y=y,D=d,Z=z\}-\frac{P_{ydz}}{P_z^2} I\{Z=z\}\right)^2 \notag \\
& \leq 2 \max_{z\in \mathcal Z}\sum_{(y,d):(y,d,z)\in \mathcal M}\left(\frac{1}{P_z^2} I\{Y=y,D=d\} + \frac{P^2_{ydz}}{P_z^4}\right) \notag \\
& \leq 2 \max_{z\in \mathcal Z} \frac{1}{P_z^2} + 2 \max_{z\in \mathcal Z} \sum_{(y,d): (y,d,z)\in \mathcal M} \frac{P_{ydz}^2}{P_z^4}\notag \\
& \lesssim |\mathcal Z|^2~,
\end{align}
where the first equality follows by definition of $\psi(V,P)$ and the final inequality from Assumption \ref{ass:test}(iii) and $\sum_{(y,d):(y,d,z)\in \mathcal M} P_{yd|z}^2 \leq 1$.
Hence, the definitions of $\Omega(P)$ and $M_i(P)$, the triangle inequality, result \eqref{lm:auxop4}, and Lemma \ref{aux:cov}(ii) imply 
\begin{multline}\label{lm:auxop5}
\|M_i(P)\|_o \leq \frac{1}{n}\{\|(\Omega(P))^\dagger\|_o^2 \|\psi(V,P)\psi(V,P)^\prime\|_o + \|(\Omega(P))^\dagger E_P[\psi(V,P)\psi(V,P)^\prime](\Omega(P))^\dagger \|_o\} \\ \lesssim \frac{1}{n}\left(\sup_{(y,d,z)\in \mathcal M} \frac{P_z}{P_{yd|z}(1-P_{yd|z})} \times |\mathcal Z|^2 + O(1)\right) = O\left(\frac{|\mathcal M|}{n}\right)~.
\end{multline}
Next, observe that the definition of $\Omega(P)$, result \eqref{lm:auxop4}, and Assumption \ref{ass:test}(iii) yield that
\begin{equation}\label{lm:auxop6}
\|(\Omega(P))^\dagger \psi(V,P) \psi(V,P)^\prime (\Omega(P))^\dagger\|_o \leq \max_{(y,d,z)\in \mathcal M} \frac{P_z}{P_{yd|z}(1-P_{yd|z})} \times \|\psi(V,P)\|_2^2 \lesssim |\mathcal M|  ~. 
\end{equation}
Hence, the triangle inequality, $\{M_i(P)\}_{i=1}^n$ being i.i.d., Lemma \ref{aux:cov}(ii), and result \eqref{lm:auxop6} imply 
\begin{equation}\label{lm:auxop7}
\|\sum_{i=1}^n E_P[M_i^2(P)]\|_o 
\leq \frac{1}{n}\left\{\|E_P[((\Omega(P))^\dagger\psi(V,P)\psi(V,P)^\prime(\Omega(P))^\dagger)^2]\|_o +O(1)\right\} \lesssim \frac{|\mathcal M|}{n}~.
\end{equation}
Results \eqref{lm:auxop5} and \eqref{lm:auxop7} allow us to apply Bernstein's inequality for matrices \citep[see, e.g., Theorem 1.4 in][]{tropp2012user} to conclude there is a  $C < \infty$ such that for all $t \geq 0$ we have
\begin{equation}\label{lm:auxop8}
P\left\{\|\sum_{i=1}^n M_i(P)\|_o > t \right\} \leq |\mathcal M| \exp\left\{-C\frac{n t^2}{|\mathcal M|(1+t)}\right\}~.  
\end{equation}
Evaluating the bound in \eqref{lm:auxop8} at $t = K\sqrt{|\mathcal M| \log(|\mathcal M|)}/\sqrt n$ for $K$ sufficiently large implies 
\begin{equation}\label{lm:auxop9}
\|\sum_{i=1}^n M_i(P)\|_o = O_P\left(\frac{\sqrt{\log(|\mathcal M|)|\mathcal M|}}{\sqrt n}\right) 
\end{equation}
uniformly in $P\in \mathbf P$. Result \eqref{lm:auxopdisp1} then follows from \eqref{lm:auxop3} and \eqref{lm:auxop9}.

In order to establish \eqref{lm:auxopdisp2}, let $N(P)$ and $N(\hat P)$ denote the null spaces 
of the matrices $\Omega(P)^\dagger \text{Var}_{P}\{\psi(V,P)\}\Omega(P)^\dagger$ and $\Omega(P)^\dagger \text{Var}_{\hat P}\{\psi(V,P)\}\Omega(P)^\dagger$ respectively, and note that the null space of $\text{Var}_P\{\psi(V,P)\}$ being a subset of the null space of $\text{Var}_{\hat P}\{\psi(V,P)\}$ with probability one implies $N(P)\subseteq N(\hat P)$ with probability one. 
On the other hand, Weyl's Perturbation Theorem (see, e.g., Corollary III.2.6 in \cite{bhatia2013matrix}), the positive eigenvalues of $\Omega(P)^\dagger \text{Var}_{P}\{\psi(V,P)\}\Omega(P)^\dagger$ being (weakly) larger than one for all $P\in \mathbf P$ by Lemma \ref{aux:cov}(iii), and result \eqref{lm:auxopdisp1} imply that the dimension of $N(\hat P)$ is (weakly) smaller than the dimension of $N(P)$ with probability tending to one uniformly in $P\in \mathbf P$.
Since, as argued, $N(P)\subseteq N(\hat P)$ with probability one, we can therefore conclude that
\begin{equation}\label{lm:auxop10}
\liminf_{n\rightarrow \infty} \inf_{P\in \mathbf P} P\{N(P) = N(\hat P)\} =1.    
\end{equation}
Next set $\underline{\lambda}^+(P)$ and $\underline{\lambda}^+(\hat P)$ to be the smallest strictly positive eigenvalues of the matrices $\Omega(P)^\dagger \text{Var}_{P}\{\psi(V,P)\}\Omega(P)^\dagger$ and $\Omega(P)^\dagger \text{Var}_{\hat P}\{\psi(V,P)\}\Omega(P)^\dagger$ respectively.
A second application of Weyl's Perturbation Theorem, Lemma \ref{aux:cov}(iii), and results \eqref{lm:auxopdisp1} and \eqref{lm:auxop10} then imply
\begin{equation}\label{lm:auxop11}
\liminf_{n\rightarrow \infty} \inf_{P\in \mathbf P} P\{\underline{\lambda}^+(\hat P) > 1/2\} =1.        
\end{equation}
In particular, results \eqref{lm:auxop10} and \eqref{lm:auxop11} imply the event $E_0(P) \equiv \{N(P) = N(\hat P) \text{ and } \underline{\lambda}^+(\hat P) > 1/2\}$ occurs with probability tending to one uniformly in $P\in \mathbf P$.
Under $E_0(P)$ we can apply Theorem X.3.8 in \cite{bhatia2013matrix} (on the subspace orthogonal to $N(P)$) to conclude
\begin{multline}\label{lm:auxop12}
    \|(\Omega(P)^\dagger{\rm Var}_{\hat P}\{\psi(V,P)\}\Omega(P)^\dagger)^{1/2} - (\Omega(P)^\dagger {\rm Var}_P\{\psi(V,P)\}\Omega(P)^\dagger)^{1/2}\|_o \\ \leq \frac{1}{\sqrt 2} \|(\Omega(P))^\dagger({\rm Var}_{\hat P}\{\psi(V,P)\} - {\rm Var}_P\{\psi(V,P)\})(\Omega(P))^\dagger\|_o,
\end{multline}
and therefore \eqref{lm:auxopdisp2} holds due to \eqref{lm:auxop12} holding with probability tending to one and \eqref{lm:auxopdisp1}.
\end{proof}

\begin{lemma}\label{aux:cov}
Let Assumptions \ref{ass:test}(i) and \ref{ass:test}(iii) hold, $\log(|\mathcal M|)|\mathcal M|/n = o(1)$, $\psi(V,P)$ be as  in \eqref{eq:psidef}, and $\Sigma(P) \equiv E_P[\psi(V,P)\psi(V,P)^\prime]$. Then, it follows that
\vspace{-0.1 in}
\begin{packed_enum}
    \item[(i)] $\psi(V,P)$ satisfies $E_P[\psi(V,P)] = 0$; 
    \item[(ii)] The eigenvalues of $\Omega(P)^\dagger \Sigma(P) \Omega(P)^\dagger$ are uniformly bounded in $P\in \mathbf P$; 
    \item[(iii)] The  nonzero eigenvalues of $\Omega(P)^\dagger \Sigma(P) \Omega(P)^\dagger$ are weakly bigger than one for all  $P\in \mathbf P$;
    \item[(iv)] $\Psi(V,P) \equiv \|\Omega(P)^\dagger \psi(V,P)\|_\infty$ satisfies $\sup_{P\in \mathbf P} E_P[|\Psi(V,P)|^3] \lesssim |\mathcal M|^{3/2}$; 
    \item[(v)] $(\hat \beta_n - \beta(P)),(\hat \beta_n^*-\hat \beta_n)\in {\rm range}(\Sigma(P))$ with probability tending to one uniformly in $P\in \mathbf P$.
    \item[(vi)] The null space of $\Sigma(P)$ has dimension $|\mathcal Z|+2$ for all $P\in \mathbf P$.
\end{packed_enum}
\end{lemma}

\begin{proof}
In order to establish the first claim, recall that the first $|\mathcal M|$ entries of $\psi(V,P)$ are indexed by $(y,d,z) \in \mathcal M$ and have the following structure
\begin{equation*}
\psi_{ydz}(V,P) = \frac{1}{P_z} I\{Y = y, D=d, Z=z\} - \frac{P_{ydz}}{P_z^2} I\{Z=z\}~.
\end{equation*}
By direct calculation, $E_P[\psi_{ydz}(V,P)] = P_{yd|z} - P_{yd|z} = 0$ and, since the final two coordinates of $\psi(V,P)$ are identically equal to zero, the first claim of the lemma follows.

To establish the second claim, we first characterize the covariance matrix $\Sigma(P)$ by noting
\begin{equation} \label{aux:cov2}
\text{Cov}[\psi_{ydz}(V,P), \psi_{\tilde y\tilde d\tilde z}(V,P)] = \begin{cases} 0 & \text{ if } z \neq \tilde z \\
-\frac{P_{yd|z}P_{\tilde y \tilde d|z}}{P_z} & \text{ if } z = \tilde z \text{ and } (y,d)\neq (\tilde y, \tilde d) \\
\frac{P_{yd|z}(1-P_{yd|z})}{P_z} &\text{ if } (y,d,z) = (\tilde y, \tilde d, \tilde z)
\end{cases}~.
\end{equation}
Also observe that the last two rows and columns of $\Sigma(P)$ are zero because the final two rows of $\psi(V,P)$ equal zero.
Since $\Omega(P)$ is a diagonal matrix in which the first $|\mathcal M|$ entries have the structure $(P_{yd|z}(1-P_{yd|z})/P_z)^{1/2}$ for some $(y,d,z) \in \mathcal M$ and the final two entries equal to zero, result \eqref{aux:cov2} implies that $\Omega(P)^\dagger \Sigma(P) \Omega(P)^\dagger$ is block diagonal and the final two columns and rows equal zero.
In particular, there are $|\mathcal Z|$ blocks, which we denote by $\Gamma_z(P)$ for $z \in \mathcal Z$.
Letting $a_z(P)$ be a vector with coordinates $\{P_{yd|z}^{1/2} : (y,d,z)\in \mathcal M\}$, $b_z(P)$ a vector with coordinates $\{(1-P_{yd|z})^{-1/2} : (y,d,z)\in \mathcal M\}$ and $\text{diag}\{b_z(P)\}$ denote a diagonal matrix with diagonal $b_z(P)$, it then follows by direct calculation that each $\Gamma_z(P)$ satisfies
\begin{equation}\label{aux:cov2p1}
\Gamma_z(P) = \text{diag}\{b_z(P)\}(I - a_z(P)a_z(P)^\prime) \text{diag}\{b_z(P)\},    
\end{equation}
where $I$ is the identity matrix. Since $\|b_z(P)\|_\infty$ is uniformly bounded due to $P_{yd|z}$ being uniformly bounded away from one by Assumption \ref{ass:test}(iii), the second claim of the lemma follows from \eqref{aux:cov2p1} and  $\Omega(P)^\dagger \Sigma(P) \Omega(P)^\dagger$ being block diagonal with blocks $\Gamma_z(P)$.

For the third claim, set $M_z(P) \equiv (I - a_z(P)a_z(P)^\prime)$ and let $\lambda_0 > 0$ be any strictly positive eigenvalue of $\Gamma_z(P)$ and $x_0$ denote its corresponding unit norm eigenvector. By result \eqref{aux:cov2p1} and $M_z(P)M_z(P) = M_z(P)$ due to $M_z(P)$ being a projection matrix, we then obtain
\begin{equation}\label{aux:cov2p2}
    \text{diag}\{b_z(P)\}M_z(P)M_z(P)\text{diag}\{b_z(P)\} x_0 = \lambda_0 x_0 .
\end{equation}
Next set $v_0 \equiv M_z(P)\text{diag}\{b_z(P)\} x_0$ and note that result \eqref{aux:cov2p2} implies that $\|v_0\|_2^2 = \lambda_0$ because $\|x_0\|_2= 1$.   
Moreover, by direct calculation we can also conclude that
\begin{equation}
    M_z(P)(\text{diag}\{b_z(P)\})^2 M_z(P) \frac{v_0}{\sqrt{\lambda_0}} = M_z(P)\text{diag}\{b_z(P)\}\Gamma_z(P)\frac{x_0}{\sqrt{\lambda_0}} = \lambda_0 \frac{v_0}{\sqrt{\lambda_0}},
\end{equation}
which implies $\lambda_0$ is an eigenvalue of the matrix $M_z(P)(\text{diag}\{b_z(P)\})^2 M_z(P)$. 
Also note
\begin{multline}\label{aux:cov2p3}
\min_{\|x\|_2 = 1, ~ \langle x,a_z(P)\rangle = 0} x^\prime M_z(P)(\text{diag}\{b_z(P)\})^2 M_z(P) x \\ =      \min_{\|x\|_2 = 1, ~ \langle x,a_z(P)\rangle = 0} x^\prime (\text{diag}\{b_z(P)\})^2 x \geq \min_{(y,d,z)\in \mathcal M} \frac{1}{1-P_{yd|z}} \geq 1
\end{multline}
where the equality follows from the definition of $M_z(P)$ and the inequality from the definition of $b_z(P)$. In particular, since $a_z(P)$ spans the null space of $M_z(P)(\text{diag}\{b_z(P)\})^2M_z(P)$, result \eqref{aux:cov2p3} implies that the positive eigenvalues of $M_z(P)(\text{diag}\{b_z(P)\})^2M_z(P)$, and hence of $\Gamma_z(P)$, are bounded below by one.
Since $\Omega(P)^\dagger \Sigma(P) \Omega(P)^\dagger$ is block diagonal with blocks $\Gamma_z(P)$, the third claim of the lemma follows. 

For the fourth claim of the lemma, use the definitions of $\psi(V,P)$ and $\Omega(P)$ to obtain
\begin{align*}
\|\Omega(P)^\dagger & \psi(V,P)\|_\infty \\
& = \max_{(y,d,z)\in \mathcal M} \left(\frac{\sqrt{P_z}}{\sqrt{P_{yd|z}(1-P_{yd|z})}} \times \left|\frac{1}{P_z}I\{Y =y, D=d, Z=z\} - \frac{P_{ydz}}{P_z^2}I\{Z=z\}\right|\right) \\ 
&\leq \max_{(y,d,z)\in \mathcal M} \frac{\sqrt{P_z}}{\sqrt{P_{yd|z}(1-P_{yd|z})}} \times \max_{(y,d,z)\in \mathcal M}\left|\frac{1}{P_z} + \frac{P_{yd|z}}{P_z}\right| \\
&\lesssim \sqrt{|\mathcal M|}~,
\end{align*}
where the final inequality holds uniformly in $P\in \mathbf P$ by Assumption \ref{ass:test}(iii).

Turning to the fifth and sixth claims,  note that by Lemma \ref{aux:prob}(iv), the event $E_0 \equiv \{\min_{z\in \mathcal Z} \hat P_z > 0\}$ has probability tending to one uniformly in $P\in \mathbf P$.
We next argue that $(\hat \beta_n - \beta(P))\in \text{range}(\Sigma(P))$ whenever the event $E_0$ occurs.
To this end, note that by result \eqref{aux:cov2}, $\Sigma(P)$ is block diagonal with $|\mathcal Z|$ blocks we denote by $\Lambda_z(P)$, and has the final two rows and columns equal to zero. 
To examine the null space of $\Lambda_z(P)$ note that a vector $a_z \equiv (a_{yd} : (y,d,z) \in \mathcal M)$ satisfies $\Lambda_z(P) a_z = 0$ if and only if it satisfies the equality
\begin{equation}\label{aux:cov6}
0 = \text{Var}_P\left(\sum_{(y,d): (y,d,z)\in \mathcal M} a_{yd}(\frac{1}{P_z}I\{Y=y,D=d,Z=z\} - \frac{P_{ydz}}{P_z^2}I\{Z=z\})\right) ~.
\end{equation}
Defining a function $f$ of $(Y,D)$ by $f(y,d) = a_{yd}$, condition \eqref{aux:cov6} may be equivalently expressed as $f$ satisfying $(f(Y,D) - E_P[f(Y,D)|Z=z])I\{Z=z\} = 0$.
However, by Assumption \ref{ass:test}(iii) such an equality can only hold if $f$, and hence $a$, is constant. 
Given the structure of $\Sigma(P)$, we thus obtain that the null space of $\Sigma(P)$ has dimension $|\mathcal Z| +2$ and its basis is given by: (a) $|\mathcal Z|$ vectors whose final two coordinates equal zero and first $|\mathcal M|$ coordinates equal a vector $a \equiv (a_{ydz} : (y,d,z)\in \mathcal M)$ satisfying $a_{ydz} = I\{z=z_0\}$ for some $z_0$; and (b) The coordinate vectors for the last two coordinates. 
Because the final two coordinates of $(\hat \beta_n - \beta(P))$ equal zero and in addition $\sum_{(y,d)} \hat P_{yd|z} = \sum_{(y,d)} P_{yd|z}=1$ whenever the event $E_0$ occurs, it follows that $(\hat \beta_n - \beta(P))$ is orthogonal to the null space of $\Sigma(P)$ and hence belongs to its range because $\Sigma(P)$ is symmetric. Identical arguments further imply that $(\hat \beta_n^*-\hat \beta_n) \in \text{range}(\Sigma(P))$ with probability tending to one uniformly in $P\in \mathbf P$.
\end{proof}

\begin{lemma}\label{lm:equiv}
Let Assumptions  \ref{ass:test}(i) and \ref{ass:test}(iii) hold, $e_j \in \mathbf R^{|\mathcal M|+2}$ has its $j^{th}$ coordinate equal to one and all other coordinates equal to zero, and $v_z \in \mathbf R^{|\mathcal M|+2}$ has its $(y,d,z)$ coordinate equal to one for all $(y,d,z)\in \mathcal M$ and all other coordinates equal to zero. 
Then: 
\vspace{-0.1 in}
\begin{packed_enum}
    \item[(i)] If $\log(|\mathcal M|)|\mathcal M|/n = o(1)$ and the null space of $A'$ is a subset of the span of the vectors $\{v_z : z\in \mathcal Z\}\cup \{e_{|\mathcal M|+1}, e_{|\mathcal M|+2}\}$, then it follows that Assumption \ref{ass:test}(ii) holds.
    \item[(ii)] If $|\mathcal M|$ is fixed with $n$ and Assumption \ref{ass:test}(ii) holds, then it follows that the null space of $A'$ is a subset of the span of the vectors $\{v_z : z\in \mathcal Z\}\cup \{e_{|\mathcal M|+1}, e_{|\mathcal M|+2}\}$.
\end{packed_enum}
\end{lemma}

\begin{proof}
For notational simplicity we set $V$ to equal the span of $\{v_z : z\in \mathcal Z\}\cup \{e_{|\mathcal M|+1}, e_{|\mathcal M|+2}\}$, ${\rm null}(A')$ to denote the null space of $A'$, and for any vector space $H$ we let $H^\perp$ denote its orthogonal complement (i.e., $H^\perp \equiv \{h^\prime : \langle h,h^\prime \rangle = 0 \text{ for all } h\in H\}$).

For the first claim, first note that set inclusion implies that $V^\perp \subseteq [{\rm null}(A')]^\perp$, and Theorem 6.6.3 in \cite{luenberger1997optimization} therefore allows us to conclude that
\begin{equation}\label{lm:equiv1}
V^\perp \subseteq {\rm range}(A).    
\end{equation}
Next define the event $E_0 \equiv \{\hat P_z > 0 \text{ for all } z\in \mathcal Z\}$ and note that under $E_0$ we have $\langle v_z , \hat \beta_n\rangle = 1$ for any $z\in \mathcal Z$.
Since $\langle v_z, \beta(P)\rangle = 1$ for any $z\in \mathcal Z$ and $\langle \hat \beta_n - \beta(P), e_{j}\rangle  = 0$ for $j \in \{|\mathcal M|+1,|\mathcal M|+2\}$ by construction, it follows that under the event $E_0$ we have $(\hat \beta_n - \beta(P)) \in V^\perp$.
Therefore, result \eqref{lm:equiv1} and $E_0$ occurring with probability tending to one uniformly in $P\in \mathbf P$ by Lemma \ref{aux:prob}(iv) imply that Assumption \ref{ass:test}(ii) holds when ${\rm null}(A')\subseteq V$.

For the second claim, fix any $P\in \mathbf P$ and let $\psi(V,P)$ be defined in \eqref{eq:psidef}.
Then note that $\Omega(P)\Omega(P)^\dagger(\hat \beta_n - \beta(P)) = (\hat \beta_n - \beta(P))$ and $\Omega(P)\Omega(P)^\dagger \psi(V,P) = \psi(V,P)$ due to the last two coordinates of $(\hat \beta_n - \beta(P))$ and $\psi(V,P)$ being zero, which imply that
\begin{multline}\label{lm:equiv2}
\|\sqrt n\{\hat \beta_n - \beta(P)\} - \frac{1}{\sqrt n}    \sum_{i=1}^n \psi(V_i,P)\|_2 \\ \leq  \|\Omega(P)\|_o \|\Omega(P)^\dagger\{\sqrt n\{\hat\beta_n - \beta(P)\} - \frac{1}{\sqrt n}\sum_{i=1}^n \psi(V_i,P)\}\|_2 = o_P(1)
\end{multline}
where the final result follows from Lemma \ref{aux:lin} and the dimension $|\mathcal M|$ being fixed. 
Setting $\Sigma(P) \equiv E_P[\psi(V,P)\psi(V,P)^\prime]$ and $\mathbb G(P) \sim N(0, \Sigma(P))$, it then follows that $\sqrt n\{\hat \beta_n - \beta(P)\} \stackrel{d}{\rightarrow} \mathbb G(P)$ by the central limit theorem and result \eqref{lm:equiv2}.
We therefore obtain that
\begin{equation}\label{lm:equiv3}
1 = \limsup_{n\to \infty} P\left\{ \sqrt n(\hat \beta_n - \beta(P)) \in \text{range}( A)\right\}
\leq P\left\{ \mathbb G(P) \in  \text{range}(A)\right\}
\end{equation}
where the equality follows from Assumption \ref{ass:test}(ii) being satisfied by hypothesis, and the inequality by the Portmanteau Theorem (see, e.g., Theorem 1.3.4(iii) in \cite{van_der_vaart1996weak}).
Next note that since $\mathbb G(P) \sim N(0, \Sigma(P))$, Theorem 3.6.1 in \cite{bogachev1998gaussian} and result \eqref{lm:equiv3} imply that $\text{range}(\Sigma(P) ) \subseteq \text{range}( A)$ or, equivalently, that
$[{\rm null}(\Sigma(P))]^\perp \subseteq [{\rm null}(A')]^\perp$ by Theorem 6.6.3 in \cite{luenberger1997optimization}.
Therefore, by set inclusion we can conclude that if Assumption \ref{ass:test}(ii) holds, then we must have 
\begin{equation}\label{lm:equiv4}
{\rm null}(A')\subseteq {\rm null}(\Sigma(P)).    
\end{equation}
Finally, note that $\psi(V,P)'v = 0$ for any $v\in V$ by definition of $\psi(V,P)$. Hence, $\Sigma(P)v = 0$ for all $v \in V$ and, since the dimension of $V$ is $|\mathcal Z|+2$, Lemma \ref{aux:cov}(vi) implies that ${\rm null}(\Sigma(P)) = V$. 
The second claim of the lemma therefore follows from result \eqref{lm:equiv4}.
\end{proof}

\begin{lemma}\label{aux:smallC}
For each $z\in \mathcal Z$ select a $(y_z,d_z)\in \mathcal M_z$, let $\psi(V,P)$ be as defined in \eqref{eq:psidef}, $\psi^{\rm s}(V,P)$ equal the vector obtained from removing the rows $((y_z,d_z,z) : z\in \mathcal Z)$ from $\psi(V,P)$, $\Sigma^{\rm s}(P) \equiv E_P[\psi^{\rm s}(V,P)\psi^{\rm s}(V,P)^\prime]$, and $\Omega^{\rm s}(P)$ be a diagonal matrix with $(j,j)$ entry $\Omega_{j,j}^{\rm s}(P)$ equal to $(\Sigma_{j,j}(P))^{1/2}$. 
If Assumptions \ref{ass:test}(i) and \ref{ass:test}(iii) hold, then it follows that
\vspace{-0.1 in}
\begin{packed_enum}
    \item[(i)] The eigenvalues of $\Omega^{\rm s}(P)^\dagger\Sigma^{\rm s}(P) \Omega^{\rm s}(P)^\dagger$ are uniformly bounded in $P\in \mathbf P$.
    \item[(ii)] The nonzero eigenvalues of $\Omega^{\rm s}(P)^\dagger \Sigma^{\rm s}(P)\Omega^{\rm s}(P)^\dagger$ are bounded from below by $\pi (|\mathcal Z|/ |\mathcal M|)$ for all $P\in \mathbf P$ and some constant $\pi > 0$ not depending on $|\mathcal M|$ or $|\mathcal Z|$.
    \item[(iii)] The range of $\Omega^{\rm s}(P)^\dagger\Sigma^{\rm s}(P) \Omega^{\rm s}(P)^\dagger$ equals all vectors with the final two coordinates equal to zero for any $P\in \mathbf P$.
\end{packed_enum}
\end{lemma}

\begin{proof}
For the first part of the lemma note that $\Omega(P)^\dagger \Sigma(P)\Omega(P)^\dagger$ is a $(|\mathcal M|+2)\times (|\mathcal M|+2)$ matrix whose first $|\mathcal M|$ rows and columns can be indexed by $(y,d,z)\in \mathcal M$. Since $\Omega^{\rm s}(P)^\dagger \Sigma^{\rm s}(P) \Omega^{\rm s}(P)^\dagger$ equals the matrix obtained from removing every row and column indexed by $((y_z,d_z,z): z \in \mathcal Z)$, part (i) of the lemma follows from Lemma \ref{aux:cov}(ii) and Cauchy's Interlacing Theorem (see, e.g., Corollary III.1.5 in \cite{bhatia2013matrix}).

For the second part, note that $\Omega^{\rm s}(P)^\dagger \Sigma^{\rm s}(P)\Omega^{\rm s}(P)^\dagger$ is a block diagonal matrix containing $|\mathcal Z|$ blocks and whose final two rows and columns equal zero.
Letting $b_z^{\rm s}(P)$ denote a vector with coordinates $((1-P_{yd|z})^{-1} : (y,d)\in \mathcal M_z\setminus \{(y_z,d_z)\})$, $a_z^{\rm s}(P)$ denote a vector with coordinates $(P_{yd|z}^{1/2}/(1-P_{yd|z})^{1/2} : (y,d)\in \mathcal M_z\setminus \{(y_z,d_z)\})$, and $\text{diag}\{b_z^{\rm s}(P)\}$ denote the diagonal matrix with entries $b_z^{\rm s}(P)$, it then follows from direct calculation that 
\begin{equation}\label{aux:smallC1}
\Gamma_z^{\rm s}(P) = \text{diag}\{b_z^{\rm s}(P)\} - a_z^{\rm s}(P) a_z^{\rm s}(P)^\prime    .
\end{equation}
Let $\pi \equiv c_1/c_2$ for $c_1$ and $c_2$ the constants in Assumption \ref{ass:test}(iii) and define $c^* \equiv \pi(|\mathcal Z|/ |\mathcal M|)$.
Then note that, because of the block diagonal structure of $\Omega^{\rm s}(P)^\dagger \Sigma^{\rm s}(P) \Omega^{\rm s}(P)^\dagger$, in order to establish part (ii) of the lemma it suffices to show that for any $x$ we have
\begin{equation}\label{aux:smallC2}
x^\prime ( \text{diag}\{b_z^{\rm s}(P)\} - c^* I - a_z^{\rm s}(P) a_z^{\rm s}(P)^\prime ) x \geq 0 ,
\end{equation}
where $I$ denotes the $(|\mathcal M_z|-1)\times (|\mathcal M_z|-1)$ identity matrix.
Also note that all coordinates of $b_z^{\rm s}(P)$ are larger than one, which implies that the matrix $\text{diag}\{b_z^{\rm s}(P)\} - c^*I$ is invertible because $c^* < 1$.
Hence, using the Cauchy-Schwarz inequality allows us to conclude that
\begin{align}\label{aux:smallC3}
(x^\prime a_z^{\rm s}(P))^2 & = (x^\prime ( \text{diag}\{b_z^{\rm s}(P)\} - c^* I)^{1/2} ( \text{diag}\{b_z^{\rm s}(P)\} - c^* I)^{-1/2} a_z^{\rm s}(P))^2 \notag \\
& \leq (x^\prime ( \text{diag}\{b_z^{\rm s}(P)\} - c^* I) x) (a_z^{\rm s}(P)^\prime ( \text{diag}\{b_z^{\rm s}(P)\} - c^* I)^{-1} a_z^{\rm s}(P)) ~.
\end{align}
Moreover, by direct calculation and using that $(1-P_{yd|z})\leq 1$ and $c^* < 1$ we obtain
\begin{multline}\label{aux:smallC4}
a_z^{\rm s}(P)^\prime ( \text{diag}\{b_z^{\rm s}(P)\} - c^* I)^{-1} a_z^{\rm s}(P) =  \sum_{(y,d)\in \mathcal M_z \setminus \{(y_z,d_z)\}} \frac{P_{yd|z}}{1 - c^*(1-P_{yd|z})}   \\
\leq \sum_{(y,d)\in \mathcal M_z \setminus \{(y_z,d_z)\}} \frac{P_{yd|z}}{1 - c^*} = \frac{1-P_{y_z d_z |z}}{1-c^*} \leq 1 ,
\end{multline}
where the final inequality follows from $c^* \leq P_{yd|z}$ for all $(y,d,z)\in \mathcal M$ and $P\in \mathbf P$ by Assumption \ref{ass:test}(iii).
The second part of the lemma then follows by noting that results \eqref{aux:smallC3} and \eqref{aux:smallC4} imply that \eqref{aux:smallC2} is satisfied. 

The third claim of the lemma is immediate from the block diagonal structure of the matrix $\Omega^{\rm s}(P)^\dagger \Sigma^{\rm s}(P)\Omega^{\rm s}(P)^\dagger$ and each block $\Gamma_z^{\rm s}(P)$ being full rank by result \eqref{aux:smallC2} and $c^* > 0$.
\end{proof}

\begin{lemma}\label{aux:prob}
Let Assumptions \ref{ass:test}(i) and \ref{ass:test}(iii) be satisfied. If in addition the condition $\log(|\mathcal M|)|\mathcal M|/n = o(1)$ holds, then it follows that uniformly in $P\in \mathbf P$: 
\vspace{-0.1 in}
\begin{packed_enum}
    \item[(i)] $\max_{(y,d,z)\in \mathcal M} |\hat P_{ydz}-P_{ydz}| = O_P(\sqrt{\log(|\mathcal M|)}/\sqrt{|\mathcal M| n})$;
    \item[(ii)] $\max_{z\in \mathcal Z} |\hat P_{z}-P_{z}| = O_P(\sqrt{\log(|\mathcal Z|)}/\sqrt{|\mathcal Z| n})$;
    \item[(iii)] $\max_{(y,d,z)\in \mathcal M}|\hat P_{ydz}/P_{ydz} - 1| = O_P(\sqrt{\log(|\mathcal M|)|\mathcal M|}/\sqrt n)$;
    \item[(iv)] $\max_{z\in \mathcal Z}|\hat P_z/P_z - 1| = O_P(\sqrt{\log(|\mathcal Z|)|\mathcal Z|}/\sqrt n)$;
    \item[(v)] $\max_{(y,d,z)\in \mathcal M} |\hat P_{yd|z}-P_{yd|z}|=O_P(\sqrt{\log(|\mathcal M|)|\mathcal Z|/n})$.
\end{packed_enum}
\end{lemma}

\begin{proof}
First note Bernstein's inequality \citep[see, for example, Lemma 2.2.9 in][] {van_der_vaart1996weak} allows us to conclude that
\begin{equation*}
P \left \{ |\hat P_{ydz} - P_{ydz}| > x \right \} \leq 2 \exp \left ( - \frac{1}{2} \frac{x^2}{\frac{P_{ydz}(1 - P_{ydz})}{n} + \frac{x}{3n}} \right )~, 
\end{equation*}
where the inequality holds for all $(y,d,z)\in \mathcal M$ and $P\in \mathbf P$.
Therefore, Lemma 2.2.10 in \cite{van_der_vaart1996weak} and the norm inequality $\|\cdot\|_1 \leq \|\cdot\|_{\psi_1}$ yield
\begin{equation} \label{eq:prob2}
E\left[\max_{(y, d, z) \in \mathcal M} |\hat P_{ydz} - P_{ydz}|\right] 
\lesssim \frac{\log(|\mathcal M|)}{n} + \max_{(y,d,z)\in \mathcal M} (P_{ydz}(1-P_{ydz}))^{1/2} \frac{\sqrt{\log(|\mathcal M|)}}{\sqrt n}
~.
\end{equation}
Next observe $P_{ydz}(1-P_{ydz}) \lesssim  1/|\mathcal M|$ uniformly in $(y,d,z)\in \mathcal M$ and $P\in \mathbf P$ by Assumption \ref{ass:test}(iii).
Therefore, result \eqref{eq:prob2},  Markov's inequality, and $\log(|\mathcal M|)|\mathcal M|/n = o(1)$ imply
\begin{equation}\label{eq:prob3}
\max_{(y,d,z)\in \mathcal M} |\hat P_{ydz}-P_{ydz}| = O_P\left(\frac{\sqrt{\log(|\mathcal M|)}}{\sqrt{|\mathcal M| n}}\right)    
\end{equation}
uniformly in $P\in \mathbf P$, which establishes the first claim.
Moreover, by identical arguments 
\begin{equation} \label{eq:prob4}
\max_{z\in \mathcal Z} |\hat P_{z}-P_{z}| = O_P\left(\frac{\sqrt{\log(|\mathcal Z|)}}{\sqrt{|\mathcal Z| n}}\right)    
\end{equation}
uniformly in $P\in \mathbf P$, which establishes the second claim.
For the third claim, note that
\begin{equation} \label{eq:prob5}
\max_{(y,d,z)\in \mathcal M}\left|\frac{\hat P_{ydz}}{P_{ydz}}-1\right| \leq \max_{(y,d,z)\in \mathcal M} \frac{1}{P_{ydz}}\times \max_{(y,d,z)\in \mathcal M} |\hat P_{ydz}-P_{ydz}| = O_P\left(\frac{\sqrt{\log(|\mathcal M|)|\mathcal M|}}{\sqrt{n}}\right)~,
\end{equation}
where the equality holds uniformly in $P\in \mathbf P$ by Assumption \ref{ass:test}(iii) and result \eqref{eq:prob3}. The fourth claim of the lemma follows by identical arguments to those employed in \eqref{eq:prob5} by relying on \eqref{eq:prob4} instead of \eqref{eq:prob3}. 
To establish the final claim of the Lemma, note that
\begin{multline}\label{eq:prob6}
\max_{(y,d,z)\in \mathcal M}\left|\frac{\hat P_{ydz}}{\hat P_z}-\frac{P_{ydz}}{P_z}\right| \leq  \max_{z\in \mathcal Z} \frac{1}{\hat P_z} \times \max_{(y,d,z)\in \mathcal M} |\hat P_{ydz}-P_{ydz}| + \max_{(y,d,z)\in \mathcal M}P_{ydz}\times \max_{z\in \mathcal Z} \frac{|\hat P_z - P_z|}{\hat P_zP_z}  \\
= O_P\left(|\mathcal Z|\frac{\sqrt{\log(|\mathcal M|)}}{\sqrt{|\mathcal M|n}}\right) + O_P\left(\frac{|\mathcal Z|^2\sqrt{\log(|\mathcal Z|)}}{|\mathcal M|\sqrt{|\mathcal Z| n} }\right) = O_P\left(\frac{\sqrt{\log(|\mathcal M|) |\mathcal Z|}}{\sqrt n}\right)
\end{multline}
where the equalities hold uniformly in $P\in \mathbf P$ by parts (i), (ii), and (iv) of this Lemma and Assumption \ref{ass:test}(iii).
\end{proof}

\begin{lemma}\label{aux:bootprob}
Let Assumptions \ref{ass:test}(i) and \ref{ass:test}(iii) be satisfied. If in addition the condition $\log(|\mathcal M|)|\mathcal M|/n = o(1)$ holds, then it follows that uniformly in $P\in \mathbf P$: 
\vspace{-0.1 in}
\begin{packed_enum}
    \item[(i)] $\max_{(y,d,z)\in \mathcal M} |\hat P^*_{ydz}-\hat P_{ydz}| = O_P(\sqrt{\log(|\mathcal M|)}/\sqrt{|\mathcal M| n})$;
    \item[(ii)] $\max_{z\in \mathcal Z} |\hat P^*_{z}-\hat P_{z}| = O_P(\sqrt{\log(|\mathcal Z|)}/\sqrt{|\mathcal Z| n})$;
    \item[(iii)] $\max_{(y,d,z)\in \mathcal M}|\hat P^*_{ydz}/\hat P_{ydz} - 1| = O_P(\sqrt{\log(|\mathcal M|)|\mathcal M|}/\sqrt n)$;
    \item[(iv)] $\max_{z\in \mathcal Z}|\hat P^*_z/\hat P_z - 1| = O_P(\sqrt{\log(|\mathcal Z|)|\mathcal Z|}/\sqrt n)$.
\end{packed_enum}
\end{lemma}

\begin{proof}
We follow similar arguments to those used in the proof of Lemma \ref{aux:prob}. 
First define 
\begin{equation*}
A_n \equiv \left\{ \max_{(y,d,z)\in \mathcal M} \left| \frac{\hat P_{ydz}}{P_{ydz}} -1 \right| \leq \frac{1}{2} \text{ and } \max_{z\in \mathcal Z} \left| \frac{\hat P_{z}}{P_{z}} -1 \right| \leq \frac{1}{2}\right\}~,
\end{equation*}
which we note is a function of $\{V_i\}_{i=1}^n$.
Therefore, the law of iterated expectations yields
\begin{multline}\label{aux:bootprob2}
P\left\{ \max_{(y,d,z)\in \mathcal M} |\hat P^*_{ydz}-\hat P_{ydz}| > x \right\} \\ \leq E_P\left[P\left\{\max_{(y,d,z)\in \mathcal M} |\hat P^*_{ydz}-\hat P_{ydz}| > x \Big| \{V_i\}_{i=1}^n \right\} I\{A_n\}\right] + P\{A_n^c\}
\end{multline}
for any $x > 0$.
Next, note that Bernstein's inequality (see, for example, Lemma 2.2.9 in \cite{van_der_vaart1996weak}) allows us to conclude that
\begin{equation*}
P \left \{ |\hat P^*_{ydz} - \hat P_{ydz}| > x \Big| \{V_i\}_{i=1}^n \right \} \leq 2 \exp \left ( - \frac{1}{2} \frac{x^2}{\frac{\hat P_{ydz}(1 - \hat P_{ydz})}{n} + \frac{x}{3n}} \right )~, 
\end{equation*}
where the inequality holds for all $(y,d,z)\in \mathcal M$.
Therefore, Lemma 2.2.10 in \cite{van_der_vaart1996weak} and the norm inequality $\|\cdot\|_1 \leq \|\cdot\|_{\psi_1}$ yield that
\begin{equation} \label{eq:bootprob4}
E_P\left[\max_{(y, d, z) \in \mathcal M} |\hat P^*_{ydz} - \hat P_{ydz}|\Big| \{V_i\}_{i=1}^n\right] I\{A_n\}
\lesssim \frac{\log(|\mathcal M|)}{n} + \max_{(y,d,z)\in \mathcal M} P^{1/2}_{ydz} \frac{\sqrt{\log(|\mathcal M|)}}{\sqrt n}
~.
\end{equation}
Next observe that $P_{ydz}\lesssim  1/|\mathcal M|$ uniformly in $(y,d,z)\in \mathcal M$ and $P\in \mathbf P$ by Assumption \ref{ass:test}(iii).
Therefore, result \eqref{eq:bootprob4} together with Markov's inequality and the condition $\log(|\mathcal M|)|\mathcal M|/n = o(1)$ yield for any constant $K > 0$ that
\begin{equation}\label{eq:bootprob5}
E_P\left[P\left\{\max_{(y,d,z)\in \mathcal M} |\hat P^*_{ydz}-\hat P_{ydz}| > K \frac{\sqrt{\log(|\mathcal M|)}}{\sqrt{|\mathcal M| n}} \bigg | \{V_i\}_{i=1}^n \right\} I\{A_n\}\right] \lesssim \frac{1}{K}~.
\end{equation}
Hence, since $\sup_{P\in \mathbf P}P\{A_n^c\} = o(1)$ by Lemmas \ref{aux:prob}(iii),(iv), results \eqref{aux:bootprob2} and \eqref{eq:bootprob5} yield
\begin{equation}\label{eq:bootprob6}
\max_{(y,d,z)\in \mathcal M} |\hat P^*_{ydz}-\hat P_{ydz}| = O_P\left(\frac{\sqrt{\log(|\mathcal M|)}}{\sqrt{|\mathcal M| n}}\right)    
\end{equation}
uniformly in $P\in \mathbf P$, which establishes the first claim.
Moreover, by identical arguments,
\begin{equation} \label{eq:bootprob7}
\max_{z\in \mathcal Z} |\hat P^*_{z}-\hat P_{z}| = O_P\left(\frac{\sqrt{\log(|\mathcal Z|)}}{\sqrt{|\mathcal Z| n}}\right)    
\end{equation}
uniformly in $P\in \mathbf P$, which establishes the second claim.
For the third claim, note that
\begin{equation} \label{eq:bootprob8}
\max_{(y,d,z)\in \mathcal M}\left|\frac{\hat P^*_{ydz}}{\hat P_{ydz}}-1\right| \leq \max_{(y,d,z)\in \mathcal M} \frac{1}{\hat P_{ydz}}\times \max_{(y,d,z)\in \mathcal M} |\hat  P^*_{ydz}- \hat P_{ydz}| = O_P\left(\frac{\sqrt{\log(|\mathcal M|)|\mathcal M|}}{\sqrt{n}}\right)~,
\end{equation}
where the equality holds uniformly in $P\in \mathbf P$ by Lemma \ref{aux:prob}(iii), Assumption \ref{ass:test}(iii), and \eqref{eq:bootprob6}. The fourth claim follows by identical arguments to those employed in \eqref{eq:bootprob8} but relying on \eqref{eq:bootprob7} instead of \eqref{eq:bootprob6} and Lemma \ref{aux:prob}(iv) instead of Lemma \ref{aux:prob}(iii).
\end{proof}

\begin{lemma}\label{lm:fix}
Let Assumption \ref{ass:cont} hold, ${\rm diam}\{B_{\ell,L}\} < \delta$ for some $\delta > 0$ and all $1\leq \ell \leq L$, $P\{(Y,D,Z)\in \mathcal Y\times \mathcal D \times \mathcal Z\} = 1$, and $f :\mathcal Y \times \mathcal D \times \mathcal Z \to \mathbf R$ be a function satisfying
\begin{equation}\label{lm:fixdisp0}
\max_{d\in \mathcal D, z\in \mathcal Z} \sup_{|y-y^\prime|  \leq 2\delta} |f(y,d,z) - f(y^\prime,d,z)| \leq \eta
\end{equation}
for some $\eta > 0$. Define $\mathcal I \equiv \{(c_o,r_t,\chi)\in \mathcal L^{|\mathcal D|}\times \mathcal R_t \times \{0, 1\}: \mathcal Y(c_o, r_t, \chi) \neq \emptyset\}$
and suppose $P = QT^{-1}$ where $Q = Q_R\times P_Z$ for $P_Z$ the marginal distribution of $Z$ and $Q_R$ satisfying
\begin{align}
\sum_{(c_o,r_t,\chi)\in \mathcal I} Q_R\{(R_o,R_t) \in (\mathcal Y(c_o, r_t, \chi)\times \{r_t\})\} & = 1 \label{lm:fixdisp1}\\
\sum_{(c_o,r_t):(c_o,r_t,1)\in \mathcal I} \left(\sup_{r_o\in \mathcal Y(c_o,r_t,1)} g(r_o,r_t) - \theta_0\right)\times Q_R\{(R_o,R_t) \in (\mathcal Y(c_o,r_t,1) \times \{r_t\}) \} & \geq 0 \label{lm:fixdisp2}\\
\sum_{(c_o,r_t):(c_o,r_t,1)\in \mathcal I} \left(\inf_{r_o\in \mathcal Y(c_o,r_t,1)} g(r_o,r_t) - \theta_0\right)\times Q_R\{(R_o,R_t) \in (\mathcal Y(c_o,r_t,1) \times \{r_t\})  \} & \leq 0 \label{lm:fixdisp3}~.
\end{align}
Then, there is a distribution $\tilde Q_R$ of $R$ satisfying $\tilde Q_R\{R\in \mathcal R\} = 1$, $E_{\tilde Q_R}[(g(R)-\theta_0)I\{R\in \mathcal R^\prime\}] = 0$, and such that $\tilde P = \tilde Q T^{-1}$ with $\tilde Q = \tilde Q_R\times P_Z$ satisfies $\int f(dP-d\tilde P)\leq 2\eta$.
\end{lemma}

\begin{proof}
The proof proceeds by constructing a discrete measure $\tilde Q_R$ satisfying the desired properties. 
We distinguish multiple cases depending on whether the inequalities in \eqref{lm:fixdisp2} and \eqref{lm:fixdisp3} hold strictly or with equality.
In what follows, it will be helpful to note that $P=QT^{-1}$ and $R$ being independent of $Z$ under $Q$ allow us to conclude that
\begin{align}\label{lm:fix1}
P\{& Y\in B_{\ell,L}, D=d,Z=z\} \notag \\
& = \sum_{c_o\in \mathcal L^{|\mathcal D|}, r_t\in \mathcal R_t} Q_R\{R_o\in B_{c_o},R_t=r_t\}I\{c_o(d)=\ell,d(z)=d\}P\{Z=z\} \notag \\
& = \sum_{(c_o,r_t,\chi)\in \mathcal I} Q_R\{R_o\in \mathcal Y(c_o,r_t,\chi), ~ R_t=r_t\}I\{c_o(d)=\ell,d(z)=d\}P\{Z=z\} ,
\end{align}
where the final equality follows from the definition of $\mathcal Y(c_o,r_t,\chi)$ and $Q_R$ satisfying \eqref{lm:fixdisp1}.
Moreover, the same equality holds if we replace $P$ and $Q$ with $\tilde P$ and $\tilde Q$ because of $\tilde Q_R\{R\in \mathcal R\} = 1$ and result \eqref{thm:equiv_cont0}. 

\noindent \underline{Case I}. First, suppose that the inequalities in \eqref{lm:fixdisp2} and \eqref{lm:fixdisp3} both hold strictly and define
\begin{equation*}
\mathbf V \equiv  \bigotimes_{(c_o,r_t) : (c_o,r_t,1)\in \mathcal I} \mathcal Y(c_o,r_t,1)~.
\end{equation*}
Note that the set $\mathbf V$ is connected because $\mathcal Y(c_o,r_t,1) = B_{c_o}\cap \mathcal R^\prime_o(r_t)$ is connected for all $(c_o,r_t)$ such that $(c_o,r_t,1)\in \mathcal I$ by Assumption \ref{ass:cont}(iv).
Further write  any element $v\in \mathbf V$ as $v=(v(c_o,r_t) \in \mathcal Y(c_o,r_t,1): (c_o,r_t,1)\in \mathcal I)$ and define the function
\begin{equation*}
F(v) \equiv \sum_{(c_o,r_t): (c_o,r_t,1)\in \mathcal I} \left(g(v(c_o,r_t),r_t) - \theta_0\right)\times Q\{(R_o,R_t) \in (\mathcal Y(c_o,r_t,1)\times \{r_t\})\}    ~.
\end{equation*}
Note that $F : \mathbf V\to \mathbf R$ and that inequalities \eqref{lm:fixdisp2} and \eqref{lm:fixdisp3} holding strictly imply there are $\underline v, \bar v \in \mathbf V$ satisfying $F(\underline v) < 0 < F(\bar v)$.
Since $F : \mathbf V \to \mathbf R$ is continuous and $\mathbf V$ is connected, it follows that $F(\mathbf V)$ is connected as well; see, e.g.,  Theorem 23.5 in \cite{munkres2000topology}.
Therefore, we conclude that there is a $v^\star\in \mathbf V$ such that $F(v^\star) = 0$.
Finally, define a discrete measure $\tilde Q_R$ to have support points $(s(c_o,r_t,\chi) : (c_o,r_t,\chi)\in \mathcal I)$, where
\begin{equation*}
s(c_o,r_t,\chi) = \begin{cases} (v^\star(c_o,r_t),r_t) & \text{ if } \chi = 1\\ \text{ any } (r_o,r_t) \in (\mathcal Y(c_o, r_t, \chi) \times \{r_t\}) & \text{ if } \chi = 0~.
\end{cases}
\end{equation*}
Note that such a definition is possible because $\mathcal Y(c_o,r_t,\chi)\neq \emptyset$ for any $(c_o,r_t,\chi)\in \mathcal I$.
We assign probabilities $\tilde Q_R\{(R_o,R_t)=s(c_o,r_t,\chi)\} = Q_R\{(R_o,R_t)\in \mathcal Y(c_o, r_t, \chi) \times \{r_t\}\}$ for any $(c_o,r_t,\chi)\in \mathcal I$. 
Since $s(c_o,r_t,\chi) \in (\mathcal Y(c_o,r_t,\chi)\times \{r_t\})$ and $\mathcal Y(c_o,r_t,\chi)\times \{r_t\} \subseteq \mathcal R$ for all $(c_o,r_t,\chi)\in \mathcal I$ by definition of $\mathcal Y(c_o,r_t,\chi)$, it follows from result \eqref{lm:fixdisp1} that $\tilde Q_R\{R\in \mathcal R\} =1$. 
Further note that $F(v^\star) = 0$ implies that $E_{\tilde Q_R}[(g(R)-\theta_0)I\{R\in \mathcal R^\prime\}] = 0$. 
Next, use result \eqref{lm:fix1} and $Q_R$ satisfying the restriction \eqref{lm:fixdisp1} to conclude that
\begin{multline}\label{lm:fix5}
\int fdP \leq \sum_{(\ell,d,z)\in \mathcal L \times \mathcal D \times \mathcal Z} P\{Y\in B_{\ell,L},D=d,Z=z\}\times  \sup_{y\in B_{\ell,L}} f(y,d,z)  \\
= \sum_{(\ell,d,z)\in \mathcal L \times \mathcal D \times \mathcal Z} \sum_{(c_o,r_t,\chi)\in \mathcal I} \Big\{Q_R\{R_o\in \mathcal Y(c_o, r_t, \chi), R_t=r_t\} I\{c_o(d)=\ell,d(z)=d\}\\ \times P\{Z=z\} \sup_{y\in B_{\ell,L}} f(y,d,z) \Big\} ~.
\end{multline}
Similarly, $\tilde Q_R\{(R_o,R_t)=s(c_o,r_t,\chi)\} = Q_R\{(R_o,R_t)\in \mathcal Y(c_o, r_t, \chi) \times \{r_t\}\}$ for $(c_o,r_t,\chi)\in \mathcal I$ by definition of $\tilde Q_R$ and $\tilde P_Z = P_Z$ allow us to obtain that
\begin{multline}\label{lm:fix6}
\int fd\tilde P \geq \sum_{(\ell,d,z)\in \mathcal L \times \mathcal D \times \mathcal Z} \tilde P\{Y\in B_{\ell,L},D=d,Z=z\} \times \inf_{y\in B_{\ell,L}} f(y,d,z)  \\
= \sum_{(\ell,d,z)\in \mathcal L \times \mathcal D \times \mathcal Z} \sum_{(c_o,r_t,\chi)\in \mathcal I} \Big\{ Q_R\{R_o\in \mathcal Y(c_o, r_t, \chi), R_t=r_t\}I\{c_o(d)=\ell,d(z)=d\} \\ \times P\{Z=z\} \inf_{y\in B_{\ell,L}} f(y,d,z) \Big\}~.
\end{multline}
Results \eqref{lm:fix5} and \eqref{lm:fix6} together with $\text{diam}\{B_{\ell,L}\} <\delta$ and the function $f$ satisfying restriction \eqref{lm:fixdisp0} then imply $\int f(dP-d\tilde P) \leq \eta$.

\noindent \underline{Case II}. Second, suppose that \eqref{lm:fixdisp2} holds with equality.
Let $\bar B_{c_o}$ denote the closure of $B_{c_o}$ and then note that $\bar B_{c_o}\cap \mathcal R^\prime_o(r_t)$ is compact since $\mathcal R^\prime_o(r_t)$ is closed by Assumption \ref{ass:cont}(iv) and $B_{\ell,L}$ is bounded.
Then note that the definition of $\mathcal Y(c_o,r_t,\chi)$ implies that $\mathcal Y(c_o,r_t,1) = B_{c_o}\cap \mathcal R_o^\prime(r_t) \neq \emptyset$ for any $(c_o,r_t)$ such that $(c_o,r_t,1)\in \mathcal I.$
Since $g$ is continuous, it follows that for each $(c_o,r_t)$ such that $(c_o,r_t,1)\in \mathcal I$, there is a $v^\star(c_o,r_t) \in \bar B_{c_o}\cap \mathcal R_o^\prime(r_t)$ satisfying
\begin{equation}\label{lm:fix7}
\sup_{r_o\in \mathcal Y(c_o,r_t,1)} g(r_o,r_t) = \sup_{r_o\in B_{c_o}\cap \mathcal R_o^\prime(r_t)} g(r_o,r_t) = g(v^\star(c_o,r_t),r_t)   ~.
\end{equation}
Note that $v^\star(c_o,r_t)$ does not necessarily belong to $B_{c_o}\cap \mathcal R_o^\prime(r_t)$.
However, since $\{B_{\tilde c_o}\}_{\tilde c_o\in \mathcal L^{|\mathcal D|}}$ is a partition of $\mathcal Y^{|\mathcal D|}$ and $\mathcal Y^{|\mathcal D|}$ is closed because $\mathcal Y$ is closed by Assumption \ref{ass:cont}(i), it follows that $v^\star(c_o,r_t)\in B_{\tilde c_o}$ for some $\tilde c_o$ and, since $v^\star(c_o,r_t) \in \mathcal R_o^\prime(r_t)$, we also have
\begin{equation}\label{lm:fix7p5}
v^\star(c_o,r_t) \in B_{\tilde c_o}\cap \mathcal R_o^\prime(r_t) = \mathcal Y(\tilde c_o,r_t,1)
\end{equation}
by definition of $\mathcal Y(\tilde c_o,r_t,\chi)$.
Moreover, for any $d\in \mathcal D$ and $v_d^\star(c_o,r_t)$ denoting the $d^{\rm th}$ coordinate of $v^\star(c_o,r_t) \in \mathcal Y^{|\mathcal D|}$ we can conclude from the triangle inequality that
\begin{align}\label{lm:fix8}
I\{&v^\star(c_o,r_t)\in  \mathcal Y(\tilde c_o,r_t,1)\} \times \sup_{y\in B_{c_o(d),L}}\sup_{y^\prime \in B_{\tilde c_o(d),L}}|y -y^\prime| \notag \\
& \leq 
I\{v^\star(c_o,r_t)\in \mathcal Y(\tilde c_o,r_t,1)\} \times \Big\{\sup_{y\in B_{c_o(d),L}} | y - v^\star_d(c_o,r_t)| + \sup_{y^\prime \in B_{\tilde c_o(d),L}}|v^\star_d(c_o,r_t) -y^\prime|\Big\} \notag \\
& \leq 2 \delta~,
\end{align}
where the final inequality follows from $v_d^\star(c_o,r_t)$ belonging to the closure of $B_{c_o(d),L}$ because $v^\star(c_o,r_t)\in \bar B_{c_o}$, $v^\star_d (c_o,r_t)\in B_{\tilde c_o(d),L}$, and ${\rm diam}\{B_{\ell,L}\} < \delta$.
Finally, we let $\tilde Q_R$ be a discrete measure with support points $(s(c_o,r_t,\chi) : (c_o,r_t,\chi)\in \mathcal I)$ given by
\begin{equation*}
s(c_o,r_t,\chi) = \begin{cases} (v^\star(c_o,r_t),r_t) & \text{ if } \chi = 1 \\ \text{ any } (r_o,r_t) \in (\mathcal Y(c_o, r_t, \chi) \times \{r_t\}) & \text{ if } \chi = 0~, 
\end{cases}
\end{equation*}
and assign them probabilities $\tilde Q_R\{(R_o,R_t) = s(c_o,r_t,\chi)\} = \sum_{\tilde c_o: (\tilde c_o, r_t, \chi) \in \mathcal I} I\{s(\tilde c_o, r_t, \chi) = s(c_o,\allowbreak r_t,\chi)\} Q_R\{(R_o,R_t)\in \mathcal Y(\tilde c_o, r_t, \chi) \times \{r_t\}\}$, where we use the fact that $s(\tilde c_o, \tilde r_t, \tilde \chi) = s(c_o, r_t,\chi)$ implies $(\tilde r_t, \tilde \chi) = (r_t, \chi)$.
Since $(v^\star(c_o,r_t),r_t)\in \mathcal R^\prime \subseteq \mathcal R$ for any $(c_o,r_t)$ such that $(c_o,r_t,1)\in \mathcal I$ and in addition $(\mathcal Y(c_o,r_t,0)\times\{r_t\})\subseteq \mathcal R$ for any $(c_o,r_t,0)\in \mathcal I$, it follows that $\tilde Q_R\{R\in \mathcal R\}=1$ due to $Q_R$ satisfying \eqref{lm:fixdisp1}.
Moreover, $\tilde Q_R$ also satisfies $E_{\tilde Q_R}[(g(R)-\theta_0)I\{R\in \mathcal R'\}] = 0$ because of result \eqref{lm:fix7} and restriction \eqref{lm:fixdisp2} holding with equality by hypothesis.
Further note  that
\begin{multline}\label{lm:fix10}
\int f d\tilde P \geq 
\sum_{(\ell,d,z)\in \mathcal L\times \mathcal D \times \mathcal Z} \sum_{(c_o,r_t,\chi)\in \mathcal I} \Big\{\tilde  Q_R\{R_o\in \mathcal Y(c_o,r_t,\chi),~ R_t=r_t\} I\{c_o(d)=\ell, d(z)=d\}\\ \times P\{Z=z\}  \inf_{y\in B_{\ell,L}} f(y,d,z) \Big\}
\end{multline}
by result \eqref{lm:fix1} applied with $\tilde P$ and $\tilde Q$ in place of $P$ and $Q$.
In addition, note that $s(\tilde c_o,\tilde r_t, \tilde \chi) \in (\mathcal Y(c_o,r_t,\chi)\times \{r_t\})$ implies that $(\tilde r_t, \tilde \chi) = (r_t,\chi)$ and therefore by definition of $\tilde Q$  we obtain
\begin{multline}\label{lm:fix11}
\tilde Q_R\{R_o\in \mathcal Y(c_o,r_t,\chi), R_t = r_t\} \\ = \sum_{\tilde c_o: (\tilde c_o,r_t,\chi)\in \mathcal I} I\{s(\tilde c_o,r_t,\chi) \in (\mathcal Y( c_o,r_t,\chi)\times \{r_t\})\} Q_R\{R_o\in \mathcal Y(\tilde c_o, r_t, \chi),R_t = r_t\}~.
\end{multline}
Furthermore, noting that $s(\tilde c_o,r_t,0) \in (\mathcal Y(c_o,r_t,0)\times \{r_t\})$ if and only if $\tilde c_o = c_o$, and using result \eqref{lm:fix8} and that the function $f:\mathcal Y \times \mathcal D \times \mathcal Z \to \mathbf R$ satisfies restriction \eqref{lm:fixdisp0} we obtain that for any $c_o, \tilde c_o$ such that $(c_o,r_t,\chi), (\tilde c_o,r_t,\chi)\in \mathcal I$  the following inequality holds
\begin{multline}\label{lm:fix12}
\sum_{\ell=1}^L I\{c_o(d)=\ell,d(z)=d\}I\{s(\tilde c_o,r_t,\chi)\in (\mathcal Y(c_o,r_t,\chi)\times \{r_t\})\} \inf_{y\in B_{\ell, L}} f(y,d,z)\\
\geq \sum_{\ell=1}^L I\{\tilde c_o(d)=\ell,d(z)=d\}I\{s(\tilde c_o,r_t,\chi)\in (\mathcal Y(c_o,r_t,\chi)\times \{r_t\})\}  \Big(\inf_{y\in B_{\ell, L}} f(y,d,z) - \eta \Big)~.
\end{multline} 
Since in addition $\sum_{c_o:(c_o,r_t,\chi)\in \mathcal I} I\{s(\tilde c_o,r_t,\chi)\in (\mathcal Y(c_o,r_t,\chi)\times \{r_t\})\} = 1$ for any $\tilde c_o$ satisfying $(\tilde c_o,r_t,\chi)\in \mathcal I$,  results \eqref{lm:fix10}--\eqref{lm:fix12} allow us to conclude that
\begin{multline*}
\int f d\tilde P 
\geq \sum_{(\ell,d,z)\in \mathcal L \times \mathcal D \times \mathcal Z}\sum_{(\tilde c_o,r_t,\chi)\in \mathcal I} \Big\{ Q\{R_o\in \mathcal Y(\tilde c_o, r_t, \chi),R_t = r_t\} I\{\tilde c_o(d) = \ell, r_t(z) = d\}\\ \times P\{Z=z\}
\times \Big(\inf_{y\in B_{\ell,L}} f(y,d,z)-\eta \Big)\Big\}~.    
\end{multline*}
Therefore, the upper bound for $\int f dP$ obtained in \eqref{lm:fix5}, the function $f:\mathcal Y \times \mathcal D \times \mathcal Z \to \mathbf R$ satisfying \eqref{lm:fixdisp0}, and ${\rm diam}\{B_{\ell,L}\} < \delta$ for all $1\leq \ell \leq L$ imply that $\int f (dP-d\tilde P)\leq 2\eta$.

\noindent \underline{Case III}. The third case with \eqref{lm:fixdisp3} holding with equality follows from the same arguments as Case II.
\end{proof}

\begin{lemma} \label{lem:sep}
For any $P\in \mathbf P$ let $P_{YD|Z}$ and $P_Z$ denote the conditional distribution of $(Y,D)$ given $Z$ and the marginal distribution of $Z$ under $P$.
Suppose $\mathbf P_0\neq \emptyset$, $\mathbf P$ is convex, $P_{YD|Z}\tilde P_Z \in \mathbf P$ for any $P,\tilde P \in \mathbf P$, $P\{(Y,D,Z)\in \mathcal Y \times \mathcal D\times \mathcal Z\} = 1$ for all $P\in \mathbf P$, and let $\mathbf Q$ be the set of distributions for $(R_o,R_t,Z)$ satisfying Assumptions \ref{as:exog}--\ref{ass:generalizedstrata} and ${\rm cl}_{\mathbf P}(\mathbf P_0)$ denote the closure of $\mathbf P_0$ under weak convergence (in $\mathbf P$).
If $P\in \mathbf P \setminus {\rm cl}_{\mathbf P}(\mathbf P_0)$, then there is a continuous and bounded function $f:\mathcal Y \times \mathcal D\times \mathcal Z\to \mathbf R$ satisfying $\int f dP = 1$ and $\int f d\tilde P \leq 0$ for all $\tilde P \in \mathbf P_0$ satisfying $\tilde P_Z = P_Z$.
\end{lemma}

\begin{proof}
Let $\mathbf V$ denote the linear span of $\mathbf P$, which forms a vector space, set $\mathcal F \equiv \mathcal Y\times \mathcal D\times \mathcal Z$ noting that $P\{(Y,D,Z)\in \mathcal F\}= 1$ for any $P\in \mathbf P$ by assumption, and define 
\begin{equation*}
C_b(\mathcal F) \equiv \{f: \mathcal F \to \mathbf R \text{ s.t. } f \text{ is continuous and bounded}\}
\end{equation*}
We equip $\mathbf V$ with the weakest topology that makes the linear maps $\tilde P \mapsto \int f d \tilde P$ continuous for each $f\in C_b(\mathcal F)$, which we denote by $\sigma(\mathbf V,C_b(\mathcal F))$.
By Theorem 5.73 and Lemmas 2.52 and 2.53 in \cite{aliprantis2006infinite}, $(\mathbf V,\sigma(\mathbf V,C_b(\mathcal F)))$ is a locally convex topological vector space and the relative topology on $\mathbf P$ induced by $\sigma(\mathbf V,C_b(\mathcal F))$ equals the topology of weak convergence on $\mathbf P$.
Next set $\mathbf P_0(P)\equiv \{\tilde P \in \mathbf P_0 : \tilde P_Z = P_Z\}$, let ${\rm cl}_{\mathbf P}(\mathbf P_0(P))$ denote the closure of $\mathbf P_0(P)$ under weak convergence (in $\mathbf P)$, and note that $\mathbf P_0(P) \subseteq \mathbf P_0$ implies $\text{cl}_{\mathbf P}(\mathbf P_0(P))\subseteq \text{cl}_{\mathbf P}(\mathbf P_0)$ and hence $P\notin \text{cl}_{\mathbf P}(\mathbf P_0(P))$ because $P\notin \text{cl}_{\mathbf P}(\mathbf P_0)$ by hypothesis.
Moreover, setting ${\rm cl}_{\mathbf V}(\mathbf P_0(P))$ to denote the $\sigma(\mathbf V, C_b(\mathcal F))$-closure of $\mathbf P_0(P)$ in $\mathbf V$, it follows from Theorem 17.4 in \cite{munkres2000topology} that ${\rm cl}_{\mathbf P}(\mathbf P_0(P)) = {\rm cl}_{\mathbf V}(\mathbf P_0(P))\cap \mathbf P$.
Therefore, since $P\in \mathbf P\setminus {\rm cl}_{\mathbf P}(\mathbf P_0(P))$, we can conclude that $P\notin {\rm cl}_{\mathbf V}(\mathbf P_0(P))$.
Also note that since $\mathbf P_0(P)$ is nonempty and convex by Lemma \ref{lem:P0P},  Lemma 5.27(6) in \cite{aliprantis2006infinite} implies ${\rm cl}_{\mathbf V}(\mathbf P_0(P))$ is nonempty and convex as well.

We have so far shown that $P\notin {\rm cl}_{\mathbf V}(\mathbf P_0(P))$ and that ${\rm cl}_{\mathbf V}(\mathbf P_0(P))$ is a nonempty closed convex subset of the locally convex topological vector space $(\mathbf V, \sigma(\mathbf V,C_b(\mathcal F)))$.
By Corollary 5.80 and Theorem 5.93 in \cite{aliprantis2006infinite}, there therefore exists a $\tilde f \in C_b(\mathcal F)$ satisfying $\int \tilde f  d\tilde P \leq 0$ for all $\tilde P\in {\rm cl}_{\mathbf V}(\mathbf P_0(P))$ and $\int \tilde f dP > 0$.
The lemma then follows by setting $f = \tilde f /\int \tilde f dP$ and using that $\mathbf P_0(P)\subseteq {\rm cl}_{\mathbf V}(\mathbf P_0(P))$.
\end{proof}

\begin{lemma} \label{lem:P0P}
For any $P\in \mathbf P$ let $P_{YD|Z}$ and $P_Z$ denote the conditional distribution of $(Y,D)$ given $Z$ and the marginal distribution of $Z$ under $P$.
Suppose $\mathbf P$ is convex, that for any $P,\tilde P \in \mathbf P$ we have $P_{YD|Z}\tilde P_Z \in \mathbf P$, and define $\mathbf P_0(P)\equiv \{\tilde P \in \mathbf P_0 : \tilde P_Z = P_Z\}$.
If $\mathbf Q$ is the set of all distributions for $(R_o,R_t,Z)$ satisfying Assumptions \ref{as:exog}--\ref{ass:generalizedstrata} and $\mathbf P_0$ is nonempty,  then it follows that $\mathbf P_0(P)$  is nonempty and convex for any $P\in \mathbf P$.
\end{lemma}

\begin{proof}
We first show $\mathbf P_0(P)$ is nonempty. 
Because $\mathbf P_0$ is nonempty, there exists $\hat P \in \mathbf P_0$ such that $\hat P = QT^{-1}$ for some $Q$ that satisfies Assumptions \ref{as:exog}--\ref{ass:generalizedstrata}, $Q \{R \in \mathcal R'\} > 0$, and $\theta(Q) = \theta_0$. 
By Assumption \ref{as:exog}, $R \indep Z$ under $Q$ and therefore $Q = Q_R\times Q_Z$ where $Q_R$ and $Q_Z$ denote the marginal distributions of $R$ and $Z$ under $Q$. 
Next, define $\tilde Q = Q_R \times P_Z$ and note that $\tilde Q \in \mathbf Q$ because $Q\in \mathbf Q$, while $\theta(\tilde Q) = \theta(Q) = \theta_0$ and $\tilde Q\{R\in \mathcal R^\prime\} = Q\{R\in \mathcal R^\prime\} > 0$ because $Q_R = \tilde Q_R$.
Set $\tilde P = \tilde Q T^{-1}$ and note that $\tilde P = \hat P_{YD|Z} P_Z$ and therefore $\tilde P\in \mathbf P$ by our supposition on $\mathbf P$ because $\hat P,P\in \mathbf P$.
Hence, we obtain that $\tilde P \in \mathbf P_0(P)$, which shows that $\mathbf P_0(P)$ is nonempty.

Next, we show that $\mathbf P_0(P)$ is convex.
To this end, fix arbitrary distributions $P^1, P^2 \in \mathbf P_0(P)$ and $\lambda \in [0, 1]$. 
By definition of $\mathbf P_0(P)$, there then exist $Q^1$ and $Q^2$ that satisfy Assumptions \ref{as:exog}--\ref{ass:generalizedstrata} and $P^j = Q^j T^{-1}$, $Q^j \{R \in \mathcal R'\} > 0$, $\theta(Q^j) = \theta_0$, and $Q_Z^j = P_Z$ for $j \in \{1, 2\}$. 
Defining $Q = \lambda Q^1 + (1 - \lambda) Q^2$, it then follows that
\begin{equation*}
Q = \lambda Q^1_R \times P_Z + (1-\lambda)Q^2_R\times P_Z = (\lambda Q^1_R + (1-\lambda)Q^2_R)\times P_Z    
\end{equation*}
and hence that $Q$ satisfies Assumption \ref{as:exog} with $Q_Z = P_Z$.
In addition, for any set $\mathcal A$ we have $Q \{R \in \mathcal A\} = \lambda Q^1 \{R \in \mathcal A\} + (1 - \lambda) Q^2 \{R \in \mathcal A\}$, which implies $Q\{R\in \mathcal R\} = 1$ and $Q\{R\in \mathcal R^\prime\} > 0$.
Moreover, because $\theta(Q^1)=\theta(Q^2) = \theta_0$, we can also conclude that
\begin{multline*}
\theta(Q)  = \frac{\lambda E_{Q^1}[g(R) I \{R \in \mathcal R'\}] + (1 - \lambda) E_{Q^2}[g(R) I \{R \in \mathcal R'\}]}{\lambda Q^1 \{R \in \mathcal R'\} + (1 - \lambda) Q^2 \{R \in \mathcal R'\}} \\
 = \frac{\lambda \theta_0 Q^1 \{R \in \mathcal R'\} + (1 - \lambda) \theta_0 Q^2 \{R \in \mathcal R'\}}{\lambda Q^1 \{R \in \mathcal R'\} + (1 - \lambda) Q^2 \{R \in \mathcal R'\}} = \theta_0 .
\end{multline*}
Since $\lambda P^1 + (1 - \lambda) P^2 = Q T^{-1}$ and $\mathbf P$ is convex, it follows that $\lambda P^1 + (1 - \lambda) P^2\in \mathbf P_0(P)$, and therefore that $\mathbf P_0(P)$ is convex. 
\end{proof}

\section{Details for Testable Implications} \label{app:testability}
In this section, we discuss a method to obtain sharp testable restrictions of the model $\mathbf Q$ in terms of analytical inequalities. In this section, we assume $|\mathcal Y| < \infty$. We will make heavy use of the fact that a nonempty bounded polyhedron can be represented in two equivalent ways. Recall from Definition 2.1 in \cite{bertsimas1997introduction} that a polyhedron in $\mathbf R^k$ is a set $\{x \in \mathbf R^k: A x \leq b\}$, known as the $H$-representation of a polyhedron. If the polyhedron is nonempty and bounded, then Theorem 2.9 in \cite{bertsimas1997introduction} implies that it can equivalently be represented as the convex hull of its finite number of vertices, known as the $V$-representation of a polyhedron. We can compute one representation from the other using the double description method of \cite{avis1991pivoting}, which is implemented in the \texttt{rcdd} package in R and the \texttt{mpt3} package in MATLAB. 

If we represent $Q$ by the column vector $q = (q(r): r \in \mathcal R)$, the model $\mathbf Q$ can then be represented by the polyhedron $\{q: \mathbf 1_{|\mathcal R|}' q = 1, q \geq 0\}$, where $\mathbf 1_{|\mathcal R|}$ denotes the $|\mathcal R| \times 1$ vector of ones. In the rest of this section, with some abuse of notation, we use $\mathbf Q$ to denote this polyhedron. Further let $A_0$ denote the first $|\mathcal M|$ rows of $A$. It follows from Corollary \ref{cor:validity} that $\mathbf Q_0(P, \mathbf Q) \neq \emptyset$ if and only if
\[ (P_{yd|z} : (y,d,z) \in \mathcal{M}) \in A_0 \mathbf Q~. \]
It follows from Corollary 2.5 in \cite{bertsimas1997introduction} that $A_0 \mathbf Q$ is a bounded polyhedron and is nonempty as long as $\mathbf Q$ is nonempty. In that case, Theorem 2.9 in \cite{bertsimas1997introduction} implies it is the convex hull of its (finite number of) vertices. Let $\mathrm{co}(\cdot)$ denote the convex hull of a set of points. We have the following algorithm for obtaining the $H$-representation of $A_0 \mathbf Q$:
\begin{algorithm}
\caption{Algorithm for computing testable implications}
\begin{algorithmic}[1]
\State Compute the set of all vertices of $\mathbf Q$, denoted by $V = \{V_i: 1 \leq i \leq n\}$.
\State Compute $A_0 V = \{A_0 V_i: 1 \leq i \leq n\}$.
\State Compute $\mathrm{co}(A_0 V)$.
\State Obtain the $H$-representation of $\mathrm{co}(A_0 V)$.   
\end{algorithmic}
\label{alg:testable}
\end{algorithm}

Step 3 in Algorithm \ref{alg:testable} is justified by the following lemma:
\begin{lemma} \label{lem:convex-hull}
Suppose $\mathbf Q$ is a non-empty bounded polyhedron and $V$ is the set of vertices of $\mathbf Q$. Then, $A_0 \mathbf Q = \mathrm{co}(A_0 V)$.
\end{lemma}

\begin{proof}
We first show $\mathrm{co}(A_0 V) \subseteq A_0 \mathbf Q$. Indeed, each $P \in \mathrm{co}(A_0 V)$ could be written as $\sum_{1 \leq i \leq n} \lambda_i A_0 V_i = A_0 \sum_{1 \leq i \leq n} \lambda_i V_i$ where $\lambda_i \geq 0$ for all $i$ and $\sum_i \lambda_i = 1$, but $\sum_{1 \leq i \leq n} \lambda_i V_i \in \mathbf Q$ since $\mathbf Q$ is a polyhedron and hence convex. To show $A_0 \mathbf Q \subseteq \mathrm{co}(A_0 V)$, fix $Q \in \mathbf Q$. Because $\mathbf Q = \mathrm{co}(V)$, $Q = \sum_{1 \leq i \leq n} \lambda_i V_i$, where $\lambda_i \geq 0$ for all $i$ and $\sum_i \lambda_i = 1$, so $A_0 Q = \sum_{1 \leq i \leq n} \lambda_i A_0 V_i \in \mathrm{co}(A_0 V)$.
\end{proof}

\section{Additional Results for Sections \ref{sec:sims} and \ref{sec:app}} \label{app:simdetails}

Table \ref{tab:sims-trueval} presents the true values of $\theta(Q)$ and $\Theta_0(P, \mathbf Q)$ in Tables \ref{tab:FSST_MC_1sided_fixQ} and \ref{tab:BP_MC_3Y_summary}, as well as $\theta(Q)$ and $\Theta_{0, L}(P, \mathbf Q)$ in Table \ref{tab:FSST_MC_cont_Ibound}.

\begin{table}[t!]
\centering
\begin{tabular}{llcc}
\toprule
\multicolumn{4}{c}{Discrete Outcomes} \\
\midrule
Design & Parameter & $\theta_0$ & $\Theta_0(P,\mathbf Q)$ \\
\midrule
One-sided noncompliance
& $\theta_{\rm ATE}(Q)$
& 0.492
& $[0.148,1.030]$ \\
One-sided noncompliance
& $\theta_{\rm Prob}(Q)$
& 0.869
& $[0.620,1.000]$ \\
Encouragement
& $\theta_{\rm ATE}(Q)$
& 0.085
& $[-1.284,1.552]$ \\
Encouragement
& $\theta_{\rm Prob}(Q)$
& 0.772
& $[0.107,1.000]$ \\
\addlinespace[1.25em]
\multicolumn{4}{c}{General-Valued Outcomes: One-Sided Noncompliance Design} \\
\midrule
Partition & $L$ & value & $\Theta_{0,L}(P,\mathbf Q_{\rm 1s})$ \\
\midrule
{\tt bin1} & 2  & 0.492 & $[-5.126,7.260]$ \\
{\tt bin2} & 5  & 0.492 & $[-2.325,3.463]$ \\
{\tt bin3} & 10 & 0.492 & $[-1.287,2.381]$ \\
\addlinespace[1.25em]
\multicolumn{4}{c}{Comparison with Moment Inequalities Testing} \\
\midrule
Design & Parameter & $\theta_0$ & $\Theta_0(P,\mathbf Q)$ \\
\midrule
$\mathcal Y = \{-1, 0, 1\}$
& $\theta_{\rm ATE}(Q)$
& 0.189
& $[-0.411, 0.650]$ \\
$\mathcal Y = \{-1, 0, 1, 2\}$
& $\theta_{\rm ATE}(Q)$
& 0.0116
& $[-1.197, 1.195]$ \\
\bottomrule
\end{tabular}
\caption{True parameter values and identified sets for the simulation designs reported in Table \ref{tab:FSST_MC_1sided_fixQ}, true parameter values and $\Theta_{0, L}(P, \mathbf Q)$ for Table \ref{tab:FSST_MC_cont_Ibound}, and true parameter values and identified sets for the simulation designs in Table \ref{tab:BP_MC_3Y_summary}. For the general-valued outcome, the parameter is $\theta_{\rm ATE}(Q)$.}
\label{tab:sims-trueval}
\end{table}

Table \ref{tab:BP_MC_2Y_summary} studies a setting with $\mathcal D = \mathcal Z = \{0, 1\}$ and either $\mathcal Y = \{0, 1, 2\}$ or $\mathcal Y = \{0, 1, 2, 3\}$. We consider the model of instrument monotonicity as in \cite{imbens1994identification} --- i.e., we set $\mathbf Q_{\rm mon}$ to be the set of distributions for $(R, Z)$ satisfying Assumptions \ref{as:exog} and \ref{ass:generalizedstrata} with $\mathcal R = \{(y(0), y(1), d(0), d(1)): d(1) \geq d(0)\}$. The parameter of interest is the average treatment effect $\theta(Q) = E_Q[Y(1) - Y(0)]$. 
For each $\mathcal Y$, we again fix a distribution $Q \in \mathbf Q_{\mathrm{mon}}$, specified in the online repository. 
For the design in which $\mathcal Y = \{0,1,2\}$, the confidence regions for both procedures exhibit remarkably similar average lengths and coverage properties. However, the computational cost of ``FSST latent" is larger than that of ``FSST ineq," reflecting the larger dimension of $A$ in the former procedure.
For the design with $\mathcal  Y = \{0,1,2,3\}$, the ``FSST ineq" confidence intervals are shorter than those based on ``FSST latent" while retaining a lower computational cost.

\begin{table}[t!]
\centering
\begin{tabular}[t]{llccccccl}
\toprule
\multicolumn{9}{c}{Encouragement Design with $\mathcal Y = \{0, 1, 2\}$}\\ \midrule
\multirow{2}{*}{Method} & \multirow{2}{*}{$\dim(A)$} & \multirow{2}{*}{$|\Theta_0(P, \mathbf{Q}_{\mathrm{mon}})|$} & \multirow{2}{*}{$n$} & \multicolumn{1}{c}{95\% CI} & Cov. & Cov. & Cov. & \multirow{2}{*}{Runtime (s.d.)} \\
& & & & length & $\theta_0$ & $\min \Theta_0$ & $\max \Theta_0$ \\
\midrule
\multirow{2}{*}{FSST latent} & $(12+2) \times$ & \multirow{2}{*}{0.115} & 3000 & 0.160 & 0.972 & 0.938 & 0.946 & 129.87 (30.84) \\
& $27$ & & 12000 & 0.136 & 0.970 & 0.943 & 0.950 & 144.05 (35.56) \\
\cmidrule{1-9}
\multirow{2}{*}{FSST ineq} & $2 \times$ & \multirow{2}{*}{0.115} & 3000 & 0.160 & 1.000 & 0.955 & 0.955 & 86.66 (20.77) \\
& $2$ & & 12000 & 0.137 & 1.000 & 0.966 & 0.958 & 100.66 (25.49) \\
\addlinespace[1.25em]
\multicolumn{9}{c}{Encouragement Design with $\mathcal Y = \{0, 1, 2, 3\}$}\\ \midrule
\multirow{2}{*}{Method} & \multirow{2}{*}{$\dim(A)$} & \multirow{2}{*}{$|\Theta_0(P, \mathbf{Q}_{\mathrm{mon}})|$} & \multirow{2}{*}{$n$} & \multicolumn{1}{c}{95\% CI} & Cov. & Cov. & Cov. & \multirow{2}{*}{Runtime (s.d.)} \\
& & & & length & $\theta_0$ & $\min \Theta_0$ & $\max \Theta_0$ \\
\midrule
\multirow{2}{*}{FSST latent} & $(16+2) \times$ & \multirow{2}{*}{0.185} & 3000 & 0.330 & 0.994 & 0.978 & 0.978 & 145.74 (42.61) \\
& $48$ & & 12000 & 0.254 & 0.999 & 0.978 & 0.979 & 162.54 (49.70) \\
\cmidrule{1-9}
\multirow{2}{*}{FSST ineq} & $2 \times$ & \multirow{2}{*}{0.185} & 3000 & 0.294 & 0.995 & 0.960 & 0.944 & 96.22 (28.58) \\
& $2$ & & 12000 & 0.240 & 1.000 & 0.958 & 0.953 & 112.45 (35.50) \\
\bottomrule
\end{tabular}
\caption{Summary results for confidence intervals under $Q\in \mathbf Q_{\rm mon}$ for $\mathcal Y = \{0, 1, 2\}$ and $\mathcal Y = \{0, 1, 2, 3\}$. Each row reports the length of the identified set, average length of the 95\% confidence intervals, coverage rate of the true parameter, coverage rates of the lower and upper endpoints of the identified set, and average runtime in seconds for obtaining one confidence region (with standard deviation of runtime in parentheses).}
\label{tab:BP_MC_2Y_summary}
\end{table}

Table \ref{tab:blattman_empirical_LATE_allcomp_onesided_LIL} repeats the results in Table \ref{tab:blattman_empirical}, but with the parameter $\lambda_n$ in the test of \cite{fang2023inference} chosen by the second recommendation in that paper, which is motivated by the law of the iterated logarithm: $\lambda_n = 1/ \sqrt{\log(e \vee p) \log(e \vee \log(e \vee n))}$, where $p$ is the number of rows of the matrix $A$.
Across all entries, the maximum absolute difference in confidence interval endpoints is 0.028. The significance of the results at the 95\% level is unchanged except for the parameter $\theta_{\rm Pr}(tc,c|\mathcal R_{\rm full})$ for long-term effects under the restriction $\mathcal R = \mathcal R_{\rm 1s,nh}$: In Table \ref{tab:blattman_empirical}, the confidence interval is $[0.007, 0.180]$, whereas in Table \ref{tab:blattman_empirical_LATE_allcomp_onesided_LIL} it is $[0.000, 0.184]$.

\begin{table}[t!]
\centering
{\small
\begin{tabular}{ccllll}
\toprule
\multicolumn{6}{c}{One-sided Noncompliance ($\mathcal R = \mathcal R_{\rm 1s}$)} \\
\cmidrule(lr){1-6}
 Strata ($\mathcal R^\prime$) & Outcome & $\theta_{\rm ATE}(tc,t|\mathcal R^\prime)$ & $\theta_{\rm ATE}(tc,c |\mathcal R^\prime)$ & $\theta_{\rm Pr}(tc,t|\mathcal R^\prime)$ & $\theta_{\rm Pr}(tc,c |\mathcal R^\prime)$  \\
\cmidrule(lr){1-6} 
Unrestricted & Short Term  & [-0.124, 0.162] & [-0.101, 0.103] & [0.000, 0.242] & [0.000, 0.219] \\
($\mathcal R^\prime = \mathcal R$) & Long Term  & [-0.183, 0.116] & [-0.200, 0.020] & [0.000, 0.259] & [0.046, 0.296]\\
\addlinespace[0.75em]
Pairwise Compliers & Short Term & [-0.164, 0.205] & [-0.121, 0.047] & [0.000, 0.167] & [0.000, 0.230]  \\
($\mathcal R^\prime = \mathcal R_{j,k}$) & Long Term & [-0.201, 0.163] & [-0.239, -0.039] & [0.000, 0.200] & [0.050, 0.306]\\
\addlinespace[0.75em]
Full Compliers & Short Term & [-0.169, 0.211] & [-0.262, 0.209] & [0.000, 0.169] & [0.000, 0.258]  \\
($\mathcal R^\prime = \mathcal R_{\rm full}$) & Long Term  & [-0.211, 0.182] & [-0.341, 0.070] & [0.000, 0.204] & [0.000, 0.344] \\
\addlinespace[1.25em]
\multicolumn{6}{c}{One-Sided Noncompliance + No-Harm Restriction ($\mathcal R = \mathcal R_{\rm 1s, nh}$)} \\
\cmidrule(lr){1-6}
 Strata ($\mathcal R^\prime$) & Outcome & $\theta_{\rm ATE}(tc,t|\mathcal R^\prime)$ & $\theta_{\rm ATE}(tc,c |\mathcal R^\prime)$ & $\theta_{\rm Pr}(tc,t|\mathcal R^\prime)$ & $\theta_{\rm Pr}(tc,c |\mathcal R^\prime)$  \\
\cmidrule(lr){1-6} 
Unrestricted & Short Term  & [-0.066, 0.000] & [-0.122, 0.000] & [0.000, 0.060] & [0.000, 0.123] \\
($\mathcal R^\prime = \mathcal R$) & Long Term & [-0.125, 0.000] & [-0.161, -0.027] & [0.000, 0.115] & [0.017, 0.153] \\
\addlinespace[0.75em]
Pairwise Compliers & Short Term & [-0.082, 0.000] & [-0.127, 0.000] & [0.000, 0.094] & [0.000, 0.127]\\
($\mathcal R^\prime = \mathcal R_{j,k}$) & Long Term  & [-0.158, 0.000] & [-0.166, -0.031] & [0.000, 0.144] & [0.027, 0.165] \\
\addlinespace[0.75em]
Full Compliers & Short Term  & [-0.081, 0.000] & [-0.142, 0.000] & [0.000, 0.090] & [0.000, 0.144] \\
($\mathcal R^\prime = \mathcal R_{\rm full}$) & Long Term & [-0.150, 0.000] & [-0.194, -0.001] & [0.000, 0.152] & [0.000, 0.184] \\
\bottomrule
\end{tabular}}
\caption{95\% confidence intervals for causal parameters $\theta_{\rm ATE}(tc,d|\mathcal R^\prime)$ and $\theta_{\rm Pr}(tc,d|\mathcal R^\prime)$ (as defined in \eqref{app:causaldef}) for different subsets of generalized principal strata $\mathcal R^\prime$. Pairwise complier subset $\mathcal R_{j,k}$ corresponds to the $(j,k)$ in the corresponding causal parameter --- e.g., $\mathcal R_{j,k} = \mathcal R_{tc,t}$ in column 3 and $\mathcal R_{j,k} = \mathcal R_{tc,c}$ in column 4.  Confidence regions were obtained through test inversion of the test in Theorem \ref{thm:inference} implemented removing linearly dependent rows (as in Remark \ref{rm:dropzdisc}), two thousand bootstrap draws, and $\lambda_n$ set to equal $1/\sqrt{\log(e\vee p)\log(e\vee\log(e\vee n))}$  in the \texttt{lpinfer} R package.}\label{tab:blattman_empirical_LATE_allcomp_onesided_LIL}

\end{table}

Table \ref{tab:blattman_empirical_pval} considers the same parameters as in Table \ref{tab:blattman_empirical} but instead reports the $p$-values for testing if $\theta_{\rm ATE}(\cdot \mid \cdot) = 0$, $\theta_{\rm Pr}(\cdot \mid \cdot) = 0$, and $\theta_{\rm Pr}(\cdot \mid \cdot) = 1$.

\begin{table}[t!]
\centering
{\small
\begin{tabular}{cccccc}
\midrule
\multicolumn{6}{c}{One-Sided Noncompliance ($\mathcal R = \mathcal R_{\rm{1s}}$)} \\
\midrule
Strata $(\mathcal R^\prime)$ & Outcome & $\theta_{\mathrm{ATE}}(tc,t|\mathcal R^\prime)$ & $\theta_{\mathrm{ATE}}(tc,c|\mathcal R^\prime)$ & $\theta_{\mathrm{Pr}}(tc,t|\mathcal R^\prime)$ & $\theta_{\mathrm{Pr}}(tc,c|\mathcal R^\prime)$ \\
\cmidrule(lr){1-6}
Unrestricted & Short Term & 1.000 & 1.000 & 1.000, 0.000 & 0.252, 0.000 \\
$(\mathcal R^\prime = \mathcal R)$ & Long Term & 1.000 & 0.126 & 1.000, 0.000 & 0.003, 0.000 \\
\addlinespace[0.75em]
Pairwise Compliers & Short Term & 1.000 & 0.457 & 1.000, 0.000 & 0.289, 0.000 \\
$(\mathcal R^\prime = \mathcal R_{j,k})$ & Long Term & 1.000 & 0.002 & 1.000, 0.000 & 0.001, 0.000 \\
\addlinespace[0.75em]
Full Compliers & Short Term & 1.000 & 1.000 & 1.000, 0.000 & 1.000, 0.000 \\
$(\mathcal R^\prime = \mathcal R_{\rm full})$ & Long Term & 1.000 & 0.380 & 1.000, 0.000 & 0.374, 0.000 \\
\addlinespace[1.25em]
\multicolumn{6}{c}{One-sided Noncompliance + No-Harm Restriction ($\mathcal R = \mathcal R_{\rm{1s,nh}}$)} \\
\midrule
Strata $(\mathcal R^\prime)$ & Outcome & $\theta_{\mathrm{ATE}}(tc,t|\mathcal R^\prime)$ & $\theta_{\mathrm{ATE}}(tc,c|\mathcal R^\prime)$ & $\theta_{\mathrm{Pr}}(tc,t|\mathcal R^\prime)$ & $\theta_{\mathrm{Pr}}(tc,c|\mathcal R^\prime)$ \\
\cmidrule(lr){1-6}
Unrestricted & Short Term & 0.883 & 0.381 & 0.886, 0.000 & 0.413, 0.000 \\
$(\mathcal R^\prime = \mathcal R)$ & Long Term & 0.224 & 0.007 & 0.196, 0.000 & 0.005, 0.000 \\
\addlinespace[0.75em]
Pairwise Compliers & Short Term & 0.851 & 0.471 & 0.848, 0.000 & 0.466, 0.000 \\
$(\mathcal R^\prime = \mathcal R_{j,k})$ & Long Term & 0.133 & 0.002 & 0.137, 0.000 & 0.010, 0.000 \\
\addlinespace[0.75em]
Full Compliers & Short Term & 0.852 & 0.603 & 0.858, 0.000 & 0.508, 0.000 \\
$(\mathcal R^\prime = \mathcal R_{\rm full})$ & Long Term & 0.140 & 0.049 & 0.120, 0.000 & 0.045, 0.000 \\
\bottomrule
\end{tabular}}
\caption{$p$-values for $\theta_{\rm ATE}(tc,d|\mathcal R^\prime)$ and $\theta_{\rm Pr}(tc,d|\mathcal R^\prime)$ (as defined in \eqref{app:causaldef}) for different subsets of generalized principal strata $\mathcal R^\prime$. ATE cells report the $p$-value for $H_0:\theta_{\rm ATE}=0$. Probability cells report the $p$-values for $H_0:\theta_{\rm Pr}=0$ and $H_0:\theta_{\rm Pr}=1$, in that order, separated by a comma. Pairwise complier subset $\mathcal R_{j,k}$ corresponds to the $(j,k)$ in the corresponding causal parameter --- e.g., $\mathcal R_{j,k} = \mathcal R_{tc,t}$ in column 3 and $\mathcal R_{j,k} = \mathcal R_{tc,c}$ in column 4. The test is implemented after removing linearly dependent rows (as in Remark \ref{rm:dropzdisc}), using two thousand bootstrap draws and the default choice of penalty $\lambda_n$ in the \texttt{lpinfer} R package.}
\label{tab:blattman_empirical_pval}
\end{table}

\end{appendices}

\end{document}